\documentclass[12pt]{article}

\usepackage{junpan}
\usepackage{graphicx,subfigure}
\usepackage{amsmath}
\usepackage[title]{appendix}
\usepackage{setspace}
\usepackage{url}
\usepackage{chngcntr}

\usepackage{amsmath,amssymb,amsthm}
\usepackage[usenames,dvipsnames]{xcolor}
\usepackage{verbatim}
\usepackage{caption}
\usepackage[flushleft]{threeparttable}
\usepackage{booktabs}
\usepackage{hyperref}
\usepackage{xcolor}
\usepackage{tabularx}
\usepackage{geometry}
\usepackage{pdfpages}

\graphicspath{{../figures/}}
\usepackage{hyperref}
\usepackage{epstopdf}
\usepackage{bm}
\usepackage{enumerate}

\hypersetup{
	colorlinks,
	linkcolor={red!50!black},
	citecolor={blue!50!black},
	urlcolor={blue!80!black}
}

\usepackage{verbatim}	% The verbatim package includes the \begin{comment} command; allows to comment out large sections at once 

\newtheorem{theorem}{Theorem}[section]

\newtheorem{definition}{Definition}[section]
\newtheorem{example}[theorem]{Example}

\newtheorem{lemma}{Lemma}[section]

\newtheorem{proposition}{Proposition}[section]

\let\oldfootnote\footnote
\renewcommand\footnote[1]{%
\oldfootnote{\hspace{0.6mm}#1}}

\begin{document}

\title{Investor Experiences and Financial Market Dynamics\thanks{
		%\title{Asset Pricing with Experience Effects\thanks{
		We thank Marianne Andries, Nick Barberis, Dirk Bergemann, Julien Cujean, Xavier Gabaix, Lawrence Jin, and workshop participants at LBS, LSE, NYU, Pompeu Fabra, Stanford, UC Berkeley, as well as the ASSA, NBER EFG Behavioral Macro, NBER Behavioral Finance, SITE (Psychology and Economics segment),
		SFB TR 15 (Tutzing, Germany) conferences for helpful comments.
		We also thank Felix Chopra, Marius Guenzel, Canyao Liu, Leslie Shen, and Jonas Sobott for excellent research assistance. 
}}
\author{Ulrike Malmendier\thanks{%
		Department of Economics and Haas School of Business, University of California, 501 Evans Hall, Berkeley, CA 94720-3880, ulrike@berkeley.edu} \\
	%EndAName
	\emph{UC\ Berkeley, NBER, and CEPR} \and Demian Pouzo\thanks{%
		Department of Economics, University of California, 501 Evans Hall, Berkeley, CA 94720-3880, dpouzo@berkeley.edu} \\
	%EndAName
	\emph{UC\ Berkeley} \and Victoria Vanasco\thanks{%
		CREI, UPF, and BGSE, Ramon Trias Fargas 25-27,
		Barcelona, 08005, Spain, vvanasco@crei.cat} \\
	%EndAName
	\emph{CREI, UPF, and BGSE} \\ 
}
\date{\today}
\maketitle
\thispagestyle{empty}
\vspace{-.2in}

\onehalfspacing

\begin{abstract}

\medskip
	
How do macro-financial shocks affect investor behavior and market dynamics? Recent evidence on experience effects suggests a long-lasting influence of personally experienced outcomes on investor beliefs and investment, but also significant differences across older and younger generations.
%that individuals over-weigh personal experiences when forming beliefs and making investment decisions. 
We formalize experience-based learning in an OLG model, 
where different cross-cohort experiences generate persistent heterogeneity in beliefs, portfolio choices, and trade. The model allows us to characterize a novel link between investor demographics and the dependence of prices on past dividends, while also generating known features of asset prices, such as excess volatility and return predictability.
The model produces new implications for the cross-section of asset holdings, trade volume, and investors' heterogenous responses to recent financial crises, which we show to be in line with the data.

%How do financial crises and stock-market fluctuations affect investor behavior and the dynamics of financial markets? Recent evidence suggests that individuals over-weigh personal experiences when forming beliefs and making investment decisions. To study the aggregate implications of such behavior, we propose a tractable OLG equilibrium model where agents form beliefs about fundamentals from their own experience. We characterize a novel link between the demographic composition in markets and the dependence of prices on past dividends. Since different generations have different experiences, there is belief heterogeneity,  that, in turn, generates trade volume. In particular, younger generations react more strongly to recent experiences than older generations, and hence have higher demand for the risky asset in good times, but lower demand in bad times. The model captures many features of asset prices, and generates new testable implications, some of which we verify in the data. 

\end{abstract}

\doublespacing
\newpage

%I HAVE NOTHING TO DISCLOSE.

%Ulrike Malmendier
%\pagenumbering{gobble}
%\newpage

%I HAVE NOTHING TO DISCLOSE.

%Demian Pouzo
%\pagenumbering{gobble}

%\newpage

%I HAVE NOTHING TO DISCLOSE.

%Victoria Vanasco
%\pagenumbering{gobble}
%\newpage

\section{Introduction}

\pagenumbering{arabic}

Recent crises in the stock and housing markets have stimulated a new wave of macro-finance models of risk-taking. A key challenge, and motivation, has been to find tractable models of investor expectations that account not only for asset-pricing puzzles such as return predictability (\citeN{CampbellShiller1988}, \citeN{fama1988dividend}) and excess volatility (\citeN{LeRoyPorter1981}, \citeN{Shiller1981}, \citeN{LeRoy2005}), but also for micro-level stylized facts such as investors chasing past performances. As argued by \citeN{Woodford2013}, the empirical evidence suggests a need for dynamic models 
that go beyond the rational-expectations hypothesis. 
In line with Woodford's proposal, models of natural expectation formation  (\citeN{FusterHebertLaibson2011}; \citeN{FusterLaibsonMendel2010}) and over-extrapolation (\citeN{barberis2015x}; \citeN{barberis2016extrapolation}) successfully capture a wide range of the stylized facts. A core feature of these models is that agents over-weigh recent realizations of the relevant economic variables when forming beliefs. 

Another set of emerging stylized facts which focuses on the long-lasting effects of macro-financial shocks and their systematic cross-sectional differences, has been harder to capture by these approaches. As conveyed by the notion of ``depression babies'' or the ``deep scars'' of the 2008 financial crisis  (\citeN{Blanchard2012}, \citeN{Malmendier_Shen2015}), macro-economic shocks appear to alter investment and consumption behavior for decades to come, %far 
beyond the time frame of existing models, and there is significant cross-sectional heterogeneity. Younger cohorts tend to react significantly more strongly than older cohorts.
The growing empirical literature on \textit{experience effects} documents, for example, that personal lifetime experiences in the stock-market predict future willingness to invest in the stock market (\citeN{Malmendier_Nagel2008}), and the same for IPO experiences and future IPO investment (\citeN{Kaustia_Knuepfer2008};  \citeN{HirshleiferEtAl2011}). There is also evidence of experience effects in non-finance settings, e.\,g., on the long-term effects of graduating in a recession on labor market outcomes (\citeN{OreopoulosEtAl2012}) or of living in (communist) Eastern Germany to political attitudes post-reunification (\citeN{AlesinaFuchs2007}).\footnote{\hspace{0mm} See also \citeN{Giuliano_Spilimbergo2014}, who relate the effects of growing up in a recession to redistribution preferences.} 
In all of these applications, researchers identify a long-lasting impact of crisis experiences on individual risk-taking and illustrate their cohort-specific impact.

Much of the evidence on experience effects pertains directly to stated beliefs, e.\,g., beliefs about future stock returns (in the UBS/Gallup data), about future inflation (in the Michigan Survey of Consumers), or about the outlook for durable consumption (also in the MSC).\footnote{\hspace{0mm} Cf. \citeN{Malmendier_Nagel2008}, \citeN{Malmendier_Nagel2013}, \citeN{Malmendier_Shen2015}.} A key difference relative to over-extrapolation and related approaches is that experience-based learning generates cohort-specific differences in beliefs and in their updating after a common shock.
While more evidence on the exact process of household-level learning is needed (see the discussions in \citeN{Campbell2008} and \citeN{AgarwalDriscollGabaixLaibson2013}), the over-weighing of personal experiences appears to be a pervasive and robust psychological phenomenon affecting belief formation, which is related to availability bias as first put forward by \citeN{TverskyKahneman1974}, as well as the extensive evidence on the different effects of description versus experience.\footnote{See, for example, \citeN{weber1993determinants}, \citeN{Hertwigetal2004}, and \citeN{simonsohn2008tree}.}  

This growing empirical literature on experience effects and its strong psychological underpinning raise the question whether experience-based learning and the implied dynamic cross-cohort differences
%, as well as their interplay with recency bias,
have the potential to explain aggregate dynamics.
% and macro-finance implications.
For example, which generations invest in the stock market and how much? What are the dynamics of stock market investment? How will the market react to a macro-shock? 

% Motivated by this, our 
Our paper develops an equilibrium model of asset markets that formalizes experience-based learning and the resulting belief heterogeneity across investors. The model clarifies the channels through which past realizations affect future market outcomes by pinning down the effect on investors' own belief formation and the interaction with other generations' belief formation. We derive the aggregate implications of learning from experience and the implied cross-sectional differences in investor behavior. To our knowledge, this model is the first to tease out the tension between experience effects and recency bias, including the stronger reactions of the young than the old to a given macro shock.
It aims to provide a guide for testing to what extent experience-based learning can enhance our understanding of market dynamics and of the long-term effect of demographic changes.

The key model features are as follows. We consider a stylized overlapping generations (OLG) equilibrium model. Agents have CARA preferences and live for a finite number of periods.\footnote{The use of CARA preferences with Normal shocks allows us to keep our theoretical analysis tractable, and is widely used in finance for this reason (see \citeN{Vives2008}).} During their lifetimes, they choose portfolios of a risky and a risk-free security. We assume that agents maximize their per-period payoffs, i.\,e., are myopic.\footnote{Myopic agents omit the correlation between their next-period payoff and their continuation value function. This yields behavior that is analogous to the commonly used assumption of short-term traders (see \citeN{Vives2008}). In a previous version of the paper, available on \href{url}{https://arxiv.org/pdf/1612.09553.pdf}, we show that, when the myopia assumption is removed, the first-order effects of experience-based learning are identical to those derived for myopic agents.}  The risky asset is in unit net supply and pays a random dividend every period. The risk-free asset is in infinitely elastic supply and pays a fixed return. Investors do not know the true mean of %the distribution of 
dividends, but learn about it by observing the history of dividends. 

%ADD a sentence saying we contrast our results with other learning models. Asegurarnos que las citas sigan.
% We contrast these assumptions with the standard model of Full Bayesian Learners (FBL) and also and alternative we dub ``Bayesian Learning from Experience'' (BLE). In a world of FBL agents use all available data and do not display recency bias. Their beliefs do not differ across cohorts and, eventually, will converge to the truth. In a world of BLE agents, i.e., where (i) holds but not (ii), this is not true. Cohorts differ in their beliefs. However, BLE does not allow for the empirically documented recency bias. As such BLE is akin to over-extrapolation in the spirit of \shortciteN{barberis2015x} and \shortciteN{barberis2016extrapolation}, but applied to individuals' cohort-specific lifetimes rather than to some cohort-independent number of recent periods.

We begin by characterizing the benchmark economy in which agents know the true mean of dividends. In this setting, there is no heterogeneity, and thus the demands of all active market participants are equal and constant over time. Furthermore, there is a unique no-bubble equilibrium with constant prices. %Any departure from this benchmark can thus be cleanly attributed to experience-based learning. 

We then introduce experience-based learning. The  assumed belief formation process captures the two main empirical features of experience effects: First, agents over-weigh their lifetime experiences. Second, their beliefs exhibit recency bias. We identify two channels through which past dividends affect market outcomes. The first channel is the belief-formation process: shocks to dividends shape agents' beliefs about future dividends. Hence, individual demands depend on personal experiences, and the equilibrium price is a function of the history of dividends observed by the oldest market participant. The second channel is the generation of cross-sectional heterogeneity: different lifetime experiences generate persistent differences in beliefs. Agents ``agree to disagree." Furthermore, younger cohorts react more strongly to a dividend shock than older cohorts as it makes up a larger part of their lifetimes. A positive shock induces younger cohorts to invest relatively more in the risky asset, while a negative shock tilts the composition towards older cohorts. %In fact, we show that periods of booms, interpreted as periods with sustained increases in dividends, result in younger generations holding a larger share of the risky asset than older generations, and vice-versa. 
Thus, the model has implications for the time series of trade volume: Changes in the level of disagreement between cohorts lead to higher trade volume in equilibrium. %The mechanism is intuitive: an increase (decrease) in dividends induces trade since young agents become relatively more optimistic (pessimistic) than old agents, and disagreement generates trade. 

The model captures an interesting tension between heterogeneity in personal experiences (which generates belief heterogeneity across cohorts) and recency bias (which reduces belief heterogeneity).
When there is strong recency bias, all agents pay a lot of attention to the most recent dividends. Thus, their reactions to a recent shock are similar. %This increases the aggregate response to a shock and reduces heterogeneity across cohorts. 
Price volatility increases, while price auto-correlation and trade volume decrease. The opposite holds when the recency bias is weak, and agents form their beliefs using their experienced history. Hence, the reaction of prices and trade volume to changes in dividends is tightly linked to the relative extent of recency bias versus experience-based differences across cohorts in a given market, which are in turn influenced by demographics.
% and by the level of recency bias in individual agents' beliefs formation. 

We explore the connection between market demographics and the dependence of prices on past dividends by analyzing the effect of a one-time change in the fraction of young agents that participate in the market. We find that the demographic composition of markets significantly influences the dependence of prices on past dividends. For example, when the market participation of the young relative to the old  increases, the relative reliance of prices on more recent dividends increases. This is in line with evidence in \citeN{cassella2015} who find that the level of extrapolation in markets is positively related to the fraction of young traders in that market, and with \citeN{collin2016asset} who find that the price-divided ratio is higher and more sensitive to macro-shocks when the ratio of young to old market participants is larger. 

We then turn to several tests of the empirical implications of our model. First, we show that the model accommodates several key asset pricing features identified in prior literature. 
We follow the approach in \citeN{Campbell_Kyle1993} and \citeN{barberis2015x} to contrast CARA-model moments with the data. We show the CARA-model analogues of return predictability (\citeN{CampbellShiller1988}) and of predictability of the dividend-price ratio.
% exhibits predictive power for future price changes. 
This predictability of future price changes (and dividend-price ratio changes) stems solely from the experience-based learning mechanism rather than, say, a built-in dependence on dividends or past returns, and it depends on the demographic structure of the market. Similarly, the model generates excess volatility in prices and price changes as established by \citeN{LeRoyPorter1981}, \citeN{Shiller1981}, and  \citeN{LeRoy2005}, above and beyond the stochastic structure of the dividend process.%Furthermore, excess volatility increases in the level of recency bias and decreases with a widening demographic composition. Finally, the $P-D$ difference exhibits high autocorrelation at short lags, and low autocorrelation at long lags, decreases in recency bias, and increases in the number of cohorts trading in the market.

Experience-based learning generates new predictions for the cross-section of asset holdings and trade volume, which we 
%verify 
test in the data. Using the representative sample of the \emph{Survey of Consumer Finance} (SCF), merged with data from the Center for Research in Security Prices (CRSP) and historical data on stock-market performance, we first replicate and extend the evidence in \citeN{Malmendier_Nagel2008} on stock-market participation. We show that cross-cohort differences in lifetime stock-market experiences predict cohort differences in stock-market participation and in the fraction of liquid assets invested in the stock market. In other words, cross-cohort differences both on the extensive and on the intensive margin of stock market participation vary over time as predicted by the time series of cross-cohort differences in lifetime experiences. We then 
turn to the predictions regarding trade volume, and 
show that the de-trended turnover ratio is strongly correlated with differences in lifetime market experiences across cohorts. That is, changes in the experience-based level of disagreement between cohorts predict higher abnormal trade volume, as predicted by the model.

% REMOVED AS SECTION WAS CUT IN JFE REVISION
%As the final step in our analysis, we argue that our qualitative results are still present when we remove the myopic-agents assumption. We consider a version of our model where agents re-balance their portfolios every period to maximize their final-period consumption.\footnote{This approach to modeling dynamic portfolio choices is again following a widely used approach in the literature, see \citeN{Vives2008}.}  The dynamic set-up allows us to analyze how hedging concerns and lifetime-horizon effects interact with experience-based learning. Prior literature has shown that, in a rational-expectations linear equilibrium, the agents' multi-period investment problem  can be partitioned into a sequence of one-period ones (\citeN{Vives2008}). Under experience-based learning, such partitioning is no longer possible. However, exploiting the CARA-Gaussian setup, we are able to show that the demand of experience-based learners coincides with the one in a different static problem where dividends are drawn from a \emph{modified} Gaussian distribution. That is, we can still partition the multi-period investment problem into a sequence of one-period problems, albeit with a probability distribution of dividends that differs from the original one. This latter result might also be of interest as an independent technical contribution in solving belief dependencies beyond our specific model proposed.

%% VICO: I have removed this paragraph altogether
% ULRIKE: I understand why. But jump to "Related Lit" was too direct. Is this OK?
Overall, experience-based learning offers a unifying explanation for financial-market features of both prices and trade volume. It also has novel implications for the cross-sectional differences in market participation and portfolio choice, which we show are consistent with the data. 
%As the ``living memory'' of all cohorts who are active in the market affects equilibrium outcomes, experience-based learning establishes a novel link between the empirical patterns of stock-market behavior, including medium- to long-term movements in the trade volume, and demographics. 

%Furthermore, since in our model agent's extrapolate from past dividends, the model can rationalize survey evidence in \citeN{GreenwoodShleifer2014}. %if agent's recency bias is high enough.
%For example, experience-based learning generates excess trade volume when different generations adjust their beliefs differently. We then show that these predictions are  

\medskip

\textbf{Related Literature}. There is a wide literature on the role of learning in explaining asset pricing puzzles. Most closely related, \citeN{cogley2008market} propose a model in which the representative consumer uses Bayes' theorem to update estimates of transition probabilities as realizations accrue. As in our paper, agents use less data than a ``rational-expectations-without-learning econometrician'' would give them. There are two important differences in our setup. First, agents are not Bayesian. Second, different cohorts have different, finite experiences. Consequently, observations during an agent's lifetime have a non-negligible effect on beliefs and generate cross-cohort heterogeneity. 

Our paper also relates to the work on extrapolation by \citeN{barberis2015x} and \citeN{barberis2016extrapolation}. They consider a consumption-based asset pricing model with both ``rational" and ``extrapolative" agents. The latter believe that positive price changes will be followed by positive changes. In contrast, the heterogeneity in extrapolation in our model is linked to the demographic structure of the market. In addition, while cross-sectional heterogeneity in their model arises from the presence of both ``rational" and ``extrapolative" infinitely-lived agents, in our model, it results from the different experiences of different finitely-lived cohorts. This allows us to generate predictions about the cross-section of asset holdings and the relation between extrapolation and demographics in line with the data.

%%% VICO: xx moved the paragraph from earlier pages to here xxx
%%% ULRIKE: moved back up (modified)
% Much of the empirical evidence on experience effects pertains directly to stated beliefs, e.\,g., beliefs about future stock returns (in the UBS/Gallup data, cf. \citeN{Malmendier_Nagel2008}), beliefs about future inflation (in the Michigan Survey of Consumers, cf. \citeN{Malmendier_Nagel2013}), or beliefs about future unemployment rates and the outlook for durable consumption (also in the Michigan Survey of Consumers; cf. \citeN{Malmendier_Shen2015}). A key difference relative to the above mentioned models of belief formation is that only experience-based learning generates cohort-specific differences in beliefs and in their updating after a common shock due to different lifetime experiences.

More generally, our paper relates to the large asset-pricing literature that departs from the correct-beliefs paradigm. For instance, \citeN{BarskyDelong1993}, \citeN{Timmermann1993}, \citeN{Timmermann1996}, \citeN{Adametal} study the implications of learning and \citeN{CecchettiLamMark2000} and \citeN{jin2015speculative} of distorted beliefs for stock-return volatility and predictability, the equity premia, and booms and busts in markets. At the same time, our approach is different from asset pricing models with asymmetric information, as surveyed in \citeN{brunnermeier2001asset}. While in these models agents want to learn the information their counter-parties hold, in our model of experienced-based learning, information is available to all agents at all times.  %Experience-based learners choose to down-weigh  the observations they have not directly observed when forming their beliefs, even though such observations are available to them (and to all other agents). 

%This fact is a key distinction between experienced-based learning and learning with private information; in this sense, in our model, agents ``agree to disagree" about the distribution of dividends. 

Finally, there are contemporaneous papers that also explore the macroeconomic effects of learning-from-experience in OLG models. \citeN{schraeder2015information} focuses on how it impacts high-frequency trading patters, such as overreaction and reversal, while \citeN{ehling2018asset} analyze the trend-chasing behavior of the young and its implications for risk-premia and the risk-free rate. More closely related to our work, \citeN{collin2016asset} explore the role of demographics on asset pricing features, such as return predictability and excess volatility. Our paper contributes to this literature in several respects. First, we allow for recency bias in the belief formation process, as both the underlying psychology literature on availability bias (\citeN{TverskyKahneman1974}) and the prior empirical literature on experience effects identify it as an important component of how individuals assign weights to previously experienced outcomes. Allowing for recency bias turns out to be also of interest theoretically, as the analysis reveals that an increase in recency bias reduce the cross-sectional heterogeneity driven by the experiential learning bias.
Second, our agents are not Bayesian and do not update their posterior variance as they gain experience, and thus our results do not depend on heterogeneous posterior variances.\footnote{Once we depart from the Bayesian paradigm, nothing guarantees that our agents understand that more information increases the precision of their beliefs. If this was the case, for example, one should expect agents to also incorporate past data.}  Third, our CARA-normal framework allows us to obtain closed-form solutions to clearly understand the link between demographics, experience, and recency. Finally, we consider our empirical approach more comprehensive, as we test the model predictions about portfolio holdings, asset pricing features, and trade volume. 

There is also a large literature that proposes other mechanisms, such as borrowing constraints or life-cycle considerations, as the link from demographics to asset prices and other equilibrium quantities. We view these other mechanisms as complementary to our paper. They are omitted for the sake of tractability of the model.

The remainder of the paper is organized as follows. First we present the model setup and the notion of experience-based learning in Section \ref{sec:Baseline}. We illustrate the mechanics of the model in a simplified setting in Section \ref{sec: ToyModel}. The main results are in Section \ref{sec:results}. In Section \ref{sec: demographics} we extend the model to study demographic shocks and in Section \ref{sec:EmpiricalImplications} we present empirical implications. Section \ref{sec: Conclusions} concludes. All proofs are 
%relegated to 
in 
the Appendix.

\section{Model Set-Up}\label{sec:Baseline}
%\subsection{Lucas-Tree Economy} 

Consider an infinite-horizon economy with overlapping generations of a continuum of risk-averse agents. At each point in time $t \in \mathbb{Z}$, a new generation is born and lives for $q$ periods, $q \in \{1,2,3,...\}$. Hence, there are $q+1$ generations alive at any $t$. The generation born at $t=n$ is called generation $n$. Each generation has a mass of $q^{-1}$ identical agents.

Agents have CARA preferences with risk aversion $\gamma$. They %are born with no endowment and 
can transfer resources across time by investing in financial markets. Trading takes place at the beginning of each period. At the end of the last period of their lives, agents consume the wealth they have accumulated. We use $n_q$ to indicate the last time at which generation $n$ trades, $n_q = n + q - 1$. (If the generation is denoted by $t$ we use $t_q$.)
Figure \ref{f: timeline} illustrates the time line of this economy for two-period lived generations ($q=2$).

There is a risk-free asset, which is in perfectly elastic supply and 
%pays 
has a gross return of $R > 1$ at all times. And there is a single risky asset (a Lucas tree), which is in unit net supply and pays a random dividend $d_{t}\sim N\left(\theta,\sigma^{2}\right)$ at time $t$. 
To model uncertainty about fundamentals, we assume that agents do  not know the true mean of dividends $\theta$ and use past observations to estimate it. To keep the model tractable, we assume that the variance of dividends $\sigma^2$ is known at all times. 

For each generation $n \in \mathbb{Z}$, the budget constraint at any time $t \in \{n,...,n+q\}$ is 
\begin{align}\label{eqn:BC}
W^{n}_{t} = x^{n}_{t} p_{t}  + a^{n}_{t},
\end{align}
where $W_{t}^{n}$ denotes the wealth of generation $n$ at time $t$,  $x^{n}_{t}$ is the investment in the risky asset (units of Lucas tree output), $a^{n}_{t}$ is the amount invested in the riskless asset, and $p_{t}$ is the price of one unit of the risky asset at time $t$. As a result, wealth next period is
\begin{align}\label{eqn:LoM-W}
W^{n}_{t+1} = x^{n}_{t}( p_{t+1} + d_{t+1})  + a^{n}_{t} R =
x^{n}_{t} ( p_{t+1} + d_{t+1} - p_{t}R)+ W^{n}_{t}R.
\end{align}

We denote the excess payoff received in $t+1$ from investing at time $t$ in one unit of the risky asset, relative to the riskless asset, as $s_{t+1} \equiv p_{t+1} + d_{t+1} - p_{t} R$. This is analogous to the equity premium in our CARA-model. %Note that $p_{t+1} + d_{t+1}$ is the payoff received at $t+1$ from investing in one unit of the risky asset at time $t$, and $p_{t}R$ is the (opportunity) cost of investing in one unit of the risky asset at time $t$. 
Using this notation, $W^{n}_{t+1} = x^{n}_{t} s_{t+1} + W^{n}_{t}R$. 

We assume that agents maximize their per-period utility (i.\,e., are myopic). This assumption simplifies the maximization problem considerably and highlights the main determinants of portfolio choice generated by experience-based learning. 

For a given initial wealth level $W^{n}_{n}$, the problem of a generation $n$ 
%is to choose $(x^{n}_{t})_{t=n}^{n_q}$ such that, for each time $t \in \{ n,..., n_q \}$, 
at each time $t \in \{ n,..., n_q \}$ is to choose $x^{n}_{t}$ to maximize $E_{t}^{n}[-\exp(-\gamma W_{t+1}^{n})]$, and hence
\begin{align}\label{eq: Generation n Problem}
x^{n}_{t} \in \arg\max_{x \in \mathbb{R}} E_{t}^{n}\left[-\exp(-\gamma x s_{t+1})\right].
\end{align}
%where $E_{t}^{n}\left[\cdot\right]$ denotes the (subjective) expectation of generation $n$ at time $t$.
where $E_{t}^{n}\left[\cdot\right]$ is the (subjective) expectation with respect to a Gaussian distribution with variance $\sigma^2$ and a mean denoted by $\theta^n_t$.
We call $\theta^n_t$ the subjective mean of dividends, and we define it below. Note that, when $x^{n}_{t}$ is negative, generation $n$ is short-selling.

\begin{figure}[t] 
	\centering
	\includegraphics[scale=1]{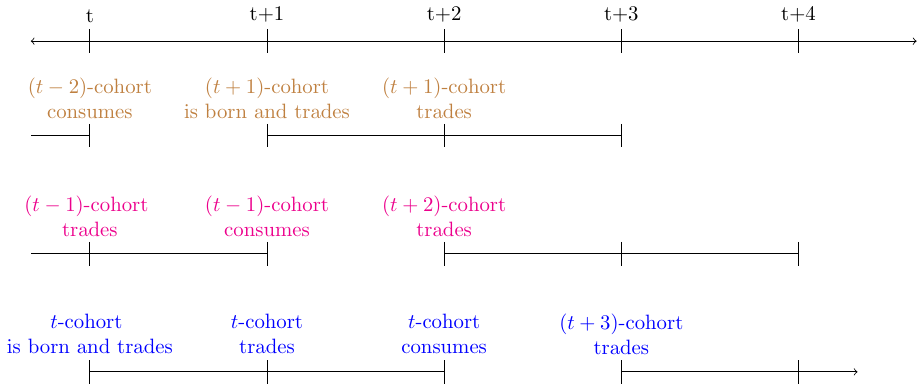} 
	\caption{A time line for an economy with two-period lived generations, $q=2$.} \label{f: timeline}
	\vspace{0.7cm}
\end{figure}

%\textbf{Observe that another difference relative to the non-myopic problem is that we kill the discounting in the myopic problem. This is necessary since it is important that the per-period problem of the agent is to maximize tomorrow's wealth, and for that there should not be discounting. The discounting appears when the agent transforms tomorrow wealth into final period wealth. 
% We do not comment on the discount since there *should* not be a discount if that is the problem ... In other words, noticing the absence of the discount rate only makes sense in comparison to a non-myopic, which we have not done yet. Instead we add a comment at the beginning of Subsection 7.1, which introduction the non-myopic problem (``Characterization of risky demands for non-myopic agents'').

\subsection{Experience-Based Learning} \label{sec:EBL} 

In this framework, experienced-based learning (EBL) means that agents over-weigh realizations observed during their lifetimes when forecasting dividends, and that they may tilt the excess weights towards the most recent observations. For simplicity, we assume that agents \emph{only} use observations realized during their lifetimes.\footnote{We only need agents to discount pre-lifetime relative to lifetime observations %when forming beliefs 
	for our results to hold.} 
That is, even though they observe the entire history of dividends, they choose to disregard earlier observations.\footnote{In our full-information setting, prices do not add any additional information. While it is possible to add private information and learning from prices to our framework, these (realistic) feature would complicate matters without necessarily adding new intuition.} 

EBL differs from reinforcement learning-type models in two ways. First, as already discussed, EBL agents understand the model and know all the primitives except the mean of the dividend process. Hence, they do not learn \emph{about} the equilibrium, they learn \emph{in} equilibrium. Second, EBL is a passive learning problem in the sense that players' actions do not affect the information they receive. This would be different if we had, say, a participation decision that links an action (participate or not) to the type of data obtained for learning. We consider this to be an interesting line to explore in the future.

Let $m$ denote the prior belief about the mean of dividends that agents are born with, and where we restrict $m$ to be Gaussian with mean $\theta$ for tractability. With this, we construct the subjective mean of dividends of generation $n$ at time $t$ following the empirical evidence on Malmendier and Nagel (2011) as follows 
\begin{equation} \label{eq: PosteriorMeanFormula}
	\theta_{t}^{n}\equiv (1-\omega_{age}) \cdot m + \omega_{age} \cdot \sum_{k=0}^{age} w(k,\lambda,age) d_{t-k},
\end{equation}
where $age=t-n$, and where, for all $k \leq age$,
\begin{align}
	\label{eq:EBL-w}
	w(k,\lambda,age) =& \frac{(age+1- k)^{\lambda}}{\sum_{k'=0}^{age} (age+1- k')^{\lambda}} 
\end{align} 
denotes the weight an agent aged $age$ assigns to the dividend observed $k$ periods earlier, and $w(k,\lambda,age) \equiv 0$ for all $k>age$. 
%That is, agents put weight $\frac{(age+1)^{\lambda}}{\sum_{k'=0}^{age} (age+1- k')^{\lambda}}$ on the most recent observation, $\frac{(age+1- 1)^{\lambda}}{\sum_{k'=0}^{age} (age+1- k')^{\lambda}}$ on the previous one, and so forth for all observations experienced during their lifetimes so far, and no weight on earlier observations. 
%The sum of all weights an agent applies to her lifetime experiences so far is always equal to one, $\sum_{k=0}^{age} w(k,\lambda,age)=1$ for all $age \in \{ 0, 1,..., q \}$. 
The denominator in (\ref{eq:EBL-w}) is a normalizing constant that depends only on $age$ and on the parameter that regulates the recency bias, $\lambda$. 
%Thus, one can rewrite $w(k,\lambda,age)$ more concisely as
%\begin{equation}
%w(k,\lambda,age) = c_{age,\lambda} \left( age + 1 - k \right)^{\lambda} 
%\end{equation}
%where $c_{age,\lambda} = 1 / \sum_{n=1}^{age+1} n^\lambda = 1 / \sum_{k'=0}^{age} (age+1- k')^{\lambda}$.
For $\lambda>0$, more recent observations receive relatively more weight, whereas for $\lambda<0$ the opposite holds.  Finally, $\omega_{age} \equiv \frac{age+1}{\tau + age+1}$  
denotes the weight that agents assign to their experience beliefs, which increases with age and decreases with the relative importance agents assign to their prior beliefs, regulated by parameter $\tau \in [0,\infty)$. Since the presence of prior beliefs has no qualitative implications in our model, unless otherwise stated, we study the case of $\tau = 0$, i.e. $\omega_{age}=1$.\footnote{We solve the model for $\tau>0$ in Appendix \ref{app: PriorBeliefs}, and we discuss the quantitative implications of prior beliefs in our model in Section \ref{sec: quantitative}.}

Here are three examples of possible weighting schemes:  
\begin{example}[Linearly Declining Weights, $\lambda=1$]
	For $\lambda=1$, weights decay linearly as the time lag increases, i.\,e., for any $0\leq k,k+j \leq age$,  \begin{align*}
		w(k+j,1,age) - w(k,1,age) = -\frac{j}{\sum_{k'=0}^{age} (age+1- k')}.
		% =	-c_{age, 1} \cdot j
		%= -\frac{2j}{(age+1)(age+2)}.
	\end{align*}
	%	$\triangle$
\end{example}

%	In this example, $c_{age,1}$ is decreasing quadratically in $age$.

\begin{example}[Equal Weights, $\lambda=0$]
	For $\lambda=0$, lifetime observations are equal-weighted, i.\,e., for any $0\leq k \leq age$,  
	%\begin{align*}
	$w(k,0,age)= \frac{1}{age+1}$.
	%\end{align*}
	%		$\triangle$	
\end{example}

\begin{example} For $\lambda \rightarrow \infty$, the weight assigned to the most recent observation converges to 1, and all other weights converge to 0, i.\,e., for any $0\leq k \leq age$,   
	%\begin{align*}
	$w(k,\lambda,age) \rightarrow 1_{\{ k = 0\}}$.
	%\end{align*}
	%		$\triangle$
\end{example}

%\begin{remark}
%	We note that the subjective mean $\theta_{t}^{n}$ in equation (\ref{eq: PosteriorMeanFormula}) can itself be thought of as an expectation with respect to the probability measure implied by the weights $w(k,\lambda,age)$, i.e., 
%	\begin{align}
%	\label{eq:3}
%	\mathbb{P}^{n}_{t}(d) = \sum_{k=0}^{age} 1_{\{ d_{t-k} \}}(d)
%	w(k,\lambda,age),~\forall d \in \mathbb{R}
%	\end{align}
%	That is, given a realization of past dividends $(d_{\tau})_{\tau=-\infty}^{t}$, $\mathbb{P}^{n}_{t}(d)$ constitutes the \emph{experience-based empirical probability measure} of generation $n$ at $t \in \{n,...,n+q\}$.\footnote{The function $x \mapsto 1_{A}(x)$ takes value 1 if $x \in A$, and 0 otherwise.}
%\end{remark}

Observe that by construction, $\theta^{n}_{t} \sim N(\theta, \sigma^{2} \sum_{k=0}^{age} (w(k,\lambda,age))^{2}  )$. Hence, $\theta^{n}_{t}$ does not necessarily converge to the truth as $t \rightarrow \infty$; it depends on whether $\sum_{k=0}^{age} (w(k,\lambda,age))^{2} \rightarrow 0$. This in turn depends on how fast the weights for ``old'' observations decay to zero (i.\,e., how small $\lambda$ is). When agents have finite lives, convergence will not occur. %In addition, since separate cohorts weight different realizations differently, we should expect belief heterogeneity at any point in time driven by different experiences. 

We conclude this section by showing a useful property of the weights, which is used in the characterization of our results. 

\begin{lemma}[Single-Crossing Property]
	Let $age'<age$ and $\lambda>0$. Then the function $w(\cdot,\lambda,age) - w(\cdot,\lambda,age')$  changes signs (from negative to non-negative) exactly once over $\{ 0,..., age'+1 \}$. %\footnote{Remember that we set $w(k,\lambda,age) \equiv 0$ for all $k>age$.}
	\label{l: SingleCrossing} 
\end{lemma}
%\begin{proof}	See Appendix \ref{app:Beliefs}.\end{proof}
	
\subsection{Comparison to Bayesian Learning}

% Ulrike: I tried to make it a subsubsection, but in the end I did not find a good way. Note that we cannot have a stand-alone sub-sub (would need at least 2 sub-sub's). Also: sub-sub is confusing/complicated. Option would be to just have it in bold-print/italics But with FBL and BLE also in bold print does not look great ...

To better understand the experience-effect mechanism, we compare the subjective mean of EBL agents to the posterior mean of agents who update their beliefs using Bayes rule. 
%The definition of EBL above implies that the expectation of generation $n$ at time $t$ denoted by $E_t^n[\cdot]$ is the expectations with respect to a Gaussian probability measure with mean $\theta_t^n$ and variance $\sigma^2$.
%We consider two cases of Bayesian learning for comparison: 
We consider two cases: \emph{Full Bayesian Learning} (FBL), wherein agents use all the available observations to form their beliefs; and \emph{Bayesian Learning from Experience} (BLE), where agents only use data realized during their lifetimes. 

\medskip
\noindent \textbf{Full Bayesian Learners.} 
%Full Bayesian learners use all the available observations ``since the beginning of time'' to form their beliefs.  Formally speaking, there is no ``beginning of time" in our economy since we are analyzing an economy that has been running forever.  Hence, loosely speaking, the beliefs of Bayesian agents who use all available information would have converged to the true $\theta$ at any point in time $t$. %, i.e., Bayesian agents would behave like agents that know the true mean. Here, however, we aim 
%
To illustrate the comparison of EBL and FBL in a common sample,
%To compare the mechanics of EBL and FBL in a common sample,
%To illustrate how the mechanis of EBL and FBL compare and construct the learning processes from a common sample,
just for this analysis, we start the economy at an initial time $t=0$, since FBL use all the available observations since ``the beginning of time." Then, all generations of FBL agents consider all observations since time 0 to form their belief. We denote the prior of FBL agents as $N(m,\sigma_m^{2})$. For simplicity, all generations have the same prior, though the analysis can easily be extended to heterogeneous Gaussian priors across generations.\footnote{The assumption of Gaussianity is also not needed but simplifies the exposition greatly.} 

The posterior mean of any generation alive at time $t$, denoted by $\hat{\theta}_{t}$, is given by
%\begin{align*}
%\gamma_{t+a} = & \frac{\tau^{-2}}{\tau^{-2} + \sigma^{-2}(t+a)} m + \frac{(t+a)\sigma^{-2}}{\tau^{-2} + \sigma^{-2}(t+a)}\sum_{k=0}^{t+a} d_{t+a-k} \frac{1}{t+a} \\
%= & \frac{\tau^{-2}}{\tau^{-2} + \sigma^{-2}(t+a)} m + \frac{(t+a)\sigma^{-2}}{\tau^{-2} + \sigma^{-2}(t+a)} \left\{ \sum_{k=0}^{a} d_{t+a-k} \frac{1}{t+a} + \sum_{k=a+1}^{t+a} d_{t+a-k} \frac{1}{t+a}    \right\}\\
%= & \frac{\tau^{-2}}{\tau^{-2} + \sigma^{-2}(t+a)} m + \frac{(t+a)\sigma^{-2}}{\tau^{-2} + \sigma^{-2}(t+a)} \left\{ \theta^{t}_{t+a} \frac{a+1}{t+a} + \sum_{k=a+1}^{t+a} d_{t+a-k} \frac{1}{t+a}    \right\}.
%\end{align*}
\begin{align*}
	\hat{\theta}_{t} = \frac{\sigma_m^{-2}}{\sigma_m^{-2} + \sigma^{-2}t} m + \frac{\sigma^{-2}t}{\sigma_m^{-2} + \sigma^{-2}t} \left( \frac{1}{t} \sum_{k=0}^{t} d_{k}  \right).
\end{align*}
The belief of an FBL agent is a convex combination 
%of the form $\alpha m + (1-\alpha) \bar d $ 
of the prior $m$ and the average of all observations $d_{k}$ realized since time $0$. The key difference to EBL agents is that differences in personal experiences do not play a role: there is no  heterogeneity in beliefs, and all generations alive in any given period have the same belief about the mean of dividends. 
In addition, beliefs of FBL agents are non-stationary, i.\,e., they depend on the time period. As $t \rightarrow \infty$, the posterior mean converges (almost surely) to the true mean. That is, with FBL the implications of learning vanish as time goes to infinity. With EBL, this is not true. Since agents have finite lives and learn from their own experiences, our model generates learning dynamics even as time diverges.

\medskip
\noindent \textbf{Bayesian Learners from Experience.}
For BLE agents, the situation is different. We assume again that each generation has a prior $N(m,\sigma_m^{2})$ when they are born. Here, the posterior mean of generation $n$ at period $t=n+age$, denoted by $\tilde{\theta}^{n}_{t}$, is given by
%\begin{align*}
%	\beta^{t}_{t+a} = & \frac{\tau^{-2}}{\tau^{-2} + \sigma^{-2}(a+1)} m + \frac{(a+1)\sigma^{-2}}{\tau^{-2} + \sigma^{-2}(a+1)}\sum_{k=0}^{a} d_{t+a-k} \frac{1}{a+1} \\
%	= & \frac{\tau^{-2}}{\tau^{-2} + \sigma^{-2}(a+1)} m + \frac{(a+1)\sigma^{-2}}{\tau^{-2} + \sigma^{-2}(a+1)} \theta^{t}_{t+a}.
%\end{align*}
%
%\begin{align*}
%\beta^{n}_{n+a} = \frac{\tau^{-2}}{\tau^{-2} + \sigma^{-2}(a+1)} m + \frac{(a+1)\sigma^{-2}}{\tau^{-2} + \sigma^{-2}(a+1)} \theta^{n}_{n+a}.
%\end{align*}
%
\begin{align*}
	\tilde{\theta}^{n}_{t} = \frac{\sigma_m^{-2}}{\sigma_m^{-2} + \sigma^{-2}(age+1)} m + \frac{\sigma^{-2}(age+1)}{\sigma_m^{-2} + \sigma^{-2}(age+1)} \left( \frac{1}{age+1} 	 \sum_{k=n}^{t} d_{k}  \right).
\end{align*}
The belief of a BLE generation is a convex combination of the prior $m$  and the average of (only) the lifetime observations $d_{k}$ available to date. The BLE belief coincides with belief $\theta^{n}_{t}$ of EBL when there is no recency bias, $\lambda =0$, and the importance assigned to prior belief is $\tau=\frac{\sigma^2}{\sigma_m^2}$. %In this case, the posterior mean of BLE agents differs from the subjective mean of EBL agents only due to the presence of the prior. As a result, if the prior of BLE agents is diffuse, i.\,e., $\tau \rightarrow \infty$, then $\tilde{\theta}^{n}_{n+a}$ coincides with the $\theta^{n}_{n+a}$ of EBL agents with $\lambda=0$. The same is true as $age\rightarrow \infty$. 
\medskip

\subsection{Equilibrium Definition} \label{sec: Equilibrium}

We now proceed to define the equilibrium of the economy with EBL agents. 

\begin{definition}[Equilibrium]
	An equilibrium is a demand profile for the risky asset 
	$\{x^{n}_{t}\}$, 
	a demand profile for the riskless asset 
	$\{a^{n}_{t}\}$, and a price schedule
	$\{ p_{t}\}$ such that: 
	\begin{enumerate}
		\item  Given the price schedule, $\{ (a^{n}_{t},x^{n}_{t}) :
		t \in \{ n,...,n_q \}\}$ solve the generation-$n$ problem. 
		\item The market clears in all periods: 
		$1 = \frac{1}{q} \sum_{n= t-q+1}^{t} x^{n}_{t}$ for all $t \in \mathbb{Z}$.
	\end{enumerate}
\end{definition}
\noindent We focus the analysis on the class of linear equilibria, i.\,e., equilibria with affine prices:

\begin{definition}[Linear Equilibrium]
	A linear equilibrium is an equilibrium wherein prices are an affine function of dividends. That is,  there exists a $K \in \mathbb{N}$, $\alpha \in \mathbb{R}$, and  $\beta_{k} \in \mathbb{R}$ for all $k \in \{0,...,K\}$ such that  
	\begin{equation} \label{e: LinearPrices}
	p_{t}=\alpha+\sum_{k=0}^{K}\beta_{k}d_{t-k}.
	\end{equation}
	%for all $t\geq K$. 
%	Thus, the price $p_t$ is a linear function of the current and the last $K$ dividends.
\end{definition}

\noindent {\bf Benchmark with known mean of dividends.} For the sake of benchmarking our results for EBL agents, we characterize equilibria in an economy where the mean of dividends, $\theta$, is known by all agents, i.\,e., $E^n_t[d_t]=\theta \; \forall \,n,t$. In this scenario, there are no disagreements across cohorts, and the demand of any cohort trading at time $t$ is 
\begin{align}\label{eq: RE_Problem}
x^{n}_{t} \in \arg\max_{x \in \mathbb{R}} E\left[-\exp(-\gamma x s_{t+1})\right].
\end{align}
The solution to this problem is standard and given by \begin{align}\label{eq: RE_Problem}
%x^n_{t} = \frac{E\left[s_{t+1}\right]-Rp_t}{\gamma V[s_{t+1}]} 
x^n_{t} = \frac{E\left[s_{t+1}\right]}{\gamma V[s_{t+1}]} 
\end{align}
for all $n\in\{t-q+1,...,t\}$, 
and zero otherwise. Since there is no heterogeneity in cohorts' demands and there is a unit supply of the risky asset, in any equilibrium, $x_t^n = 1$ for all $n\in\{t-q+1,...,t\}$, and zero otherwise. Furthermore, there exists a unique bubble-free equilibrium with constant prices $p_t = P \; \forall t$ where $P = \frac{\theta-\gamma \sigma^2}{R-1}$. 
%\begin{align*}
%\label{eq:MC}
%1 =& (q)^{-1} \sum_{n= t-q}^{t} \frac{E\left[p_{t+1}+d_{t+1}\right]-rp_t}{\gamma V[d_{t+1}]} = (q)^{-1} q \frac{P(1-r)+\theta}{\gamma \sigma^2} \\  \Rightarrow P =& \frac{\gamma \sigma^2 - \theta}{1-r}
%\end{align*}  

%I think the "thus" comes because of the following logic: Since there is no heterogeneity, all agents demand the same, since in equilibrium markets should clear, x_{t}^{n} = 1 (in equilibrium). I thought it was ok, but I might be missing something in the English?

%\section{Illustration: A Toy Model} \label{sec: ToyModel}
\section{Toy Model} \label{sec: ToyModel}

To illustrate the mechanics of the model, we first highlight the main results of the paper in a simple environment, namely, for $q=2$. 
We will solve the model for any $q>1$ in the next section.

In the toy model with $q=2$, there are three cohorts alive at each point in time: a young cohort, which enters the market for the first time; a middle-aged cohort, which is participating in the market for the second time; and an old cohort, whose agents simply consume the payoffs from their lifetime investments. %Since the old cohort has no impact on equilibrium prices or quantities, we focus our analysis on the behavior of the young and middle-aged agents. 
At time $t$, the problem of generations $n\in \{t,t-1\}$ is given by (\ref{eq: Generation n Problem}). It is easy to show that their demands for the risky asset are 
\begin{equation*}
x_t^n = \frac{E^{n}_{t}\left[s_{t+1}\right]}{\gamma V_t^n[s_{t+1}]} .
\end{equation*}

As one of our first key results in Section \ref{sec:results}, we will show that (i) prices depend on the history of dividends, and (ii) this price predictability is limited to the past dividends observed (experienced) by the oldest generation trading in the market. In other words, we show that 
$K=q-1$ in equation (\ref{e: LinearPrices}). Anticipating this result here for $q=2$, we have %$K=1$ and thus
\begin{equation}
p_t = \alpha + \beta_0 d_t + \beta_1 d_{t-1}.
\end{equation} 

The dependence of prices on past dividends is an important feature of our model, which is shared by many models of extrapolation and learning. A distinct feature of our model is that this dependence is intrinsically linked to the demographic structure of the economy. %In our framework, Price predictability extends back to the earliest dividend realizations experienced by agents still trading today, and the level of extrapolation in prices depends on the fraction of young and old cohorts in the market.
It matters which generations are participating in the market and how much.

The cross-sectional differences in lifetime experiences, and the resulting cross-sectional differences in beliefs, determine cohorts' trading behavior. Given the functional form for prices, we can re-write the demands of both cohorts that are actively trading as
\begin{align*}
%	 	x_t^t =&  (\alpha+ (1+\beta_0) E_{t}^{t}\left[d_{t+1}\right]+\beta_1 d_t- p_{t}R) / (\gamma\left(1+\beta_{0}\right)^{2}\sigma^{2}) \\
%	 	x_t^{t-1}=&(\alpha+ (1+\beta_0) E_{t}^{t-1}\left[d_{t+1}\right]+\beta_1 d_t- p_{t}R)/(\gamma\left(1+\beta_{0}\right)^{2}\sigma^{2}).\\
x_t^t =&  \frac{\alpha+ (1+\beta_0) E_{t}^{t}\left[d_{t+1}\right]+\beta_1 d_t- p_{t}R}{\gamma\left(1+\beta_{0}\right)^{2}\sigma^{2}} \\
x_t^{t-1}=&\frac{\alpha+ (1+\beta_0) E_{t}^{t-1}\left[d_{t+1}\right]+\beta_1 d_t- p_{t}R}{\gamma\left(1+\beta_{0}\right)^{2}\sigma^{2}}.
\end{align*} 
The difference between cohorts' demand arises from their different beliefs about future dividends, $E_{t}^{t}[d_{t+1}]$ and $E_{t}^{t-1}[d_{t+1}]$, given by
\begin{align*}
E_{t}^{t}\left[d_{t+1}\right] =&\; d_t, \\ 
%=\; \;\theta_t^t\;
E_{t}^{t-1}\left[d_{t+1}\right] =& \underbrace{\left(\frac{2^\lambda }{1+2^\lambda}\right)}_{ w(0,\lambda,1)} d_t +
\underbrace{\left(\frac{1}{1+2^\lambda}\right)}_{ w(1, \lambda,1)} d_{t-1}.
%= \theta_t^{t-1}
\end{align*}

These formulas illustrate the mechanics of EBL and the cause of heterogeneity among agents. 
In the simplified setting, the younger generation has only experienced the dividend $d_t$ and expects the dividends to be identical in the next period. The older generation, having more experience, incorporates the previous dividend in their weighing scheme. An implication of these formulas is that the younger generations react more optimistically  than older generations to positive changes in recent dividends, and more pessimistically to negative changes. In Section \ref{sec:cross-sec}, we show that this result continues to hold in the general model. We also see that belief heterogeneity is increasing in the change in dividends, $|d_t - d_{t-1}|$, and decreasing in the recency bias, $\lambda$. In Section \ref{sec:tradevolume}, we exploit this observation to link movements in the volume of trade to belief disagreements.

We now impose the market clearing condition, $\frac{1}{2}(x_t^t + x_t^{t-1}) =1$, to derive the equilibrium price given these demands.  
%REMOVED THIS EQUATION
%	\begin{align*}
%	1&=\frac{\alpha+\frac{1}{2}(1+\beta_0)\left[d_{t}+\frac{2^{\lambda}}{1+2^{\lambda}}d_{t}+\frac{1}{1+2^{\lambda}}d_{t-1}\right]+\beta_{1}d_{t}-R(\alpha+\beta_0 d_t+\beta_1 d_{t-1})}{\gamma\left(1+\beta_{0}\right)^{2}\sigma^{2}} 
%\end{align*}
We use the method of undetermined coefficients to solve for $\{\alpha,\beta_0,\beta_1\}$. 
%By equalizing the constants and setting 
Setting the constants and 
the terms that multiply $d_t$ and $d_{t-1}$ to zero, we obtain %the following conditions:
%\begin{align*}
%R\alpha	& =	\alpha-\gamma\left(1+\beta_{0}\right)^{2}\sigma^{2} \\
%R\beta_{0}	& =	\frac{1}{2}\left(1+\beta_{0}\right)\left(1+\frac{2^{\lambda}}{1+2^{\lambda}}\right)+\beta_{1} \\
%R\beta_{1}	& =	\frac{1}{2}\left(1+\beta_{0}\right)\left(\frac{1}{1+2^{\lambda}}\right) %.
%\end{align*}
a system of equations whose solution determines the price constant and the loadings of present and past dividends on prices,
\begin{align} 
\alpha =& -\frac{\gamma (1+\beta_0)^2 \sigma^2}{R-1} , \label{e: AlphaToyModel}\\
\beta_{0}	=&\frac{2R^{2}}{\left(R-1\right)\left(1+2R-\frac{2^{\lambda}}{1+2^{\lambda}}\right)}-1,  \label{e: Beta0ToyModel} \\
\beta_{1}	=&	\frac{R\left(1-\frac{2^{\lambda}}{1+2^{\lambda}}\right)}{\left(R-1\right)\left(1+2R-\frac{2^{\lambda}}{1+2^{\lambda}}\right)}.  \label{e: Beta1ToyModel}
\end{align}

These solutions illustrate how the price loadings on past dividends depend on the demographics of the economy and on the magnitude of the recency bias.  It is easy to derive %expressions for 
the unconditional price volatility, which is $\sigma(p_{t}) = (\beta_0^2 + \beta_1^2)^{\frac{1}{2}} \sigma$, and price auto-correlation, which is $\rho(p_{t}, p_{t+j}) = \beta_0\beta_1$ for $j=1$ and $\rho(p_{t}, p_{t+j}) =0$ for $j>1$. The variance of prices is increasing in the recency bias $\lambda$ while the price auto-correlation is decreasing in the recency bias. The intuition is straightforward: as the recency bias increases, prices become more responsive to the most recent dividend, $\frac{\partial \beta_0}{\partial \lambda} >0$, increasing price volatility, and less responsive to past dividends, $\frac{\partial \beta_1}{\partial \lambda} < 0$, decreasing price autocorrelation. In Section \ref{sec: demographics}, we present an enriched version of the model with demographic shocks and discuss how these price loadings vary with the demographic structure of the economy. 
%One final important aspect of experience-based learning that the toy model allows us to anticipate is return predictability. Since dividends in this model predict prices, they also predict future excess returns:
%\begin{equation*}
%\frac{p_{t+1}+d_{t+1}}{p_{t}}-R=\frac{\alpha+\left(1+\beta_{0}\right)d_{t+1}+\beta_{1}d_{t}}{\alpha+\beta_{0}d_{t}+\beta_{1}d_{t-1}}-R. 
%\end{equation*}   
%Moreover, adding expectation operators 
%$E_{t}^{t}$ and $E_{t-1}^{t}$ 
%$E_{t}^{t}$ and $E_{t}^{t-1}$ 
%for generation $t$ and $t-1$, respectively, dividends also predict expected excess returns.
%This equation is a first illustration how our model links demographics and market participation (i.e., which generations are trading in the market) to return predictability.
 
%\section{Results for General Model}
\section{General Model}
\label{sec:results}
We now return to the general case (i.\,e., allow for any $q>1$) and characterize the portfolio choices and resulting demands for the risky asset of the different cohorts when agents exhibit EBL.
We impose affine prices, then use market clearing to verify the affine prices guess, and fully characterize demands and prices. 
Deriving the results in the general model
%We show that the general model replicates the results obtained in the toy model, and 
allows us to discuss in more detail 
the relation between demographics, cross-sectional asset holdings, and market dynamics. % that prices depend on the dividend payments experienced by the generations trading in the market, on cross-sectional differences in (experience-based) beliefs generating trade volume, and on the volatility and autocorrelation of prices, as well as how these effects are scaled down by the extent of recency bias. In addition, we will discuss which generations will react more strongly to recent changes in dividends and characterize the abnormal trade volume under various scenarios of dividend changes.  
We obtain testable predictions, which we bring to the data in Section \ref{sec:EmpiricalImplications}.

\subsection{Characterization of Equilibrium Demands and Prices}
\label{sec:char-prices}
For any $s,t \in \mathbb{Z}$, let $d_{s:t} = (d_{s},...,d_{t})$ denote the history of dividends from time $s$ up to time $t$. For simplicity and WLOG, we assume that the initial wealth of all generations is zero, i.e., $W_n^n=0$ for all $n \in \mathbb{Z}$.  At time $t \in \{ n, ..., n_q \}$, an agent of generation $n$ determines her demand for the risky asset maximizing $E_{t}^{n}\left[-\exp\left(-\gamma x s_{t+1}\right)\right]$, as in (\ref{eq: Generation n Problem}).

The model set-up allows us to derive a standard expression for risky-asset demand:

\begin{proposition} \label{pro: demandsstatic} 
	Suppose $p_{t} = \alpha + \sum_{k=0}^{K} \beta_{k} d_{t-k}$ with $\beta_{0} \ne -1$. Then, for any generation $n \in \mathbb{Z}$ trading in period $t \in \{ n, ..., n_q \}$, demands for the risky asset are given by
	\begin{equation} \label{e: EqDemands}
	x^{n}_{t} = \frac{E^n_{t}[s_{t+1}]}{\gamma V[s_{t+1}]} = \frac{E^n_{t}[s_{t+1}]}{\gamma (1+\beta_{0})^{2} \sigma^{2}}.
	\end{equation}
\end{proposition}
%\begin{proof} %[Proof of Proposition \ref{pro: demandsstatic}]
%	The result follows by Lemma \ref{l: static} in Appendix \ref{app:results}.
%\end{proof}

%\subsection{Characterization of Equilibrium Prices}
%\label{sec:char-prices}
The expression for the risky-asset demands in equation (\ref{e: EqDemands}) allows us to derive equilibrium prices. Note that equation (\ref{e: EqDemands}) implies that demands at time $t$ are affine in $d_{t-K:t}$. It is easy to see, then, that beliefs about future dividends are linear functions of the dividends observed by each generation participating in the market, and thus prices depend on the history of dividends observed by the oldest generation in the market: 

\begin{proposition} \label{pro: prices_myopic}
	The price in any linear equilibrium is affine in the history of dividends observed by the oldest generation participating in the market, i.e., for any $t \in \mathbb{Z}$ 
%	\footnote{XXX Heuristically, an equilibrium with $\beta_{0} = -1$ is not well-defined since in this case the excess payoff, say, $s_{t+q-1}$ is deterministic given the information at time $t+q-2$ and thus the agent will take infinite positions depending on $d_{t+q-1}+p_{t+q-1} - rp_{t+q-2}$. XXX}
	\begin{equation} \label{e: EqPrices}
	p_{t}=\alpha+\sum_{k=0}^{q-1}\beta_{k}d_{t-k}, \;with
	\end{equation}
%	with \begin{align}
%	\alpha&=-\frac{1}{\left(1-\sum_{j=0}^{q-1}\frac{w_{j}}{r^{j+1}}\right)^{2}}\frac{\gamma\sigma^{2}}{r-1} \\
%	\beta_0&=\frac{\sum_{j=0}^{q-1}\frac{w_{j}}{r^{j+1}}}{1-\sum_{j=0}^{q-1}\frac{w_{j}}{r^{j+1}}} \\
%	\beta_{k}&=\left(\frac{\sum_{j=0}^{q-1-k}\frac{w_{k+j}}{r^{j+1}}}{1-\sum_{j=0}^{q-1}\frac{w_{j}}{r^{j+1}}}\right)\quad\quad k=\left\{ 1,...,q-1\right\}  \\
%	\beta_q&=\beta_{q+1}=...=0
%	\end{align}
	 \begin{align}\label{eqn:alpha}
	\alpha&=-\frac{1}{\left(1-\sum_{j=0}^{q-1}\frac{w_{j}}{R^{j+1}}\right)^{2}}\frac{\gamma\sigma^{2}}{R-1} \\\label{eqn:beta_k}
	\beta_{k}&=\frac{\sum_{j=0}^{q-1-k}\frac{w_{k+j}}{R^{j+1}}}{1-\sum_{j=0}^{q-1}\frac{w_{j}}{R^{j+1}}}\quad\quad k \in \left\{ 0,...,q-1\right\}  
	%\\
	%\beta_q&=\beta_{q+1}=...=0
	\end{align}
	where $w_{k}\equiv\frac{1}{q}\sum_{age=0}^{q-1}w\left(k,\lambda,age\right)$. %\footnote{Remember that $w(k,\lambda,age) \equiv 0$ for all $k>age$.}
%	$w_{k}\equiv\frac{1}{q}\sum_{n=t}^{t-q+1}w\left(k,\lambda,t-n\right)$.
	
\end{proposition}

%\begin{proof} %[Proof of Proposition \ref{pro: prices_myopic}]
%	See Appendix \ref{app:results}.
%\end{proof}

Proposition \ref{pro: prices_myopic} establishes a novel link between the factors influencing asset prices and demographic composition. For each $k=\{0,1,...,q-1\}$, one can interpret $w_{k}$ as the average weight placed on the dividend observed at time $t-k$ by all generations trading at time $t$.  As the formula also reveals, the relative magnitudes of the weights on past dividends, $\beta_k$, depend on the number of cohorts in the market, $q$, on the fraction of each cohort in the market, $\frac{1}{q}$, and on the extent of agents' recency bias, $\lambda$.\footnote{In our baseline model, cohorts are equally weighted. We remove this assumption in Section \ref{sec: demographics}, where we analyze demographic shocks. In those examples, there is no link between the number of cohorts and the fraction of each cohort in the market.}  

The main idea of the proposition is as follows.  In a linear equilibrium,
demands at time $t$ are affine in dividends $d_{t-K:t}$. However, from these dividends, only $d_{t-q+1:t}$ matter for forming beliefs; the dividends $d_{t-K:t-q}$ only enter through the definition of linear equilibrium. The proof shows that, under market clearing, the coefficients accompanying older dividends $d_{t-K:t-q}$ are zero. The proposition also implies that we can apply the same restriction to demands and conclude that demands at time $t$ only depend on $d_{t-q+1:t}$. %This result captures the belief channel described by Friedman and Schwartz: Prices are a function of past dividends due to the fact that agents form their beliefs using past realizations of data that they have personally experienced. 

%ADJUSTED IN REPONSE TO REFEREE #4 COMMENT.
Note that $\frac{\partial \beta_k}{\partial R} < 0$ and $\frac{\partial \alpha}{\partial R} > 0$ for any $\lambda$. That is, if the interest rate is higher, the equilibrium price of the risky asset responds less strongly to past dividends. Furthermore, higher risk aversion $\gamma$ decreases the equilibrium price by lowering $\alpha$. 

\medskip

The following proposition establishes that, as long as agents exhibit any positive recency bias (i.\,e., $\lambda>0$), the sensitivity of prices to past dividends is stronger the more recent the dividend realization.

\begin{proposition}\label{pro:prices-q2-myopic}
%For $\lambda>0$, $0<\beta_{q-1}<....<\beta_{1}<\beta_{0} \leq \frac{1}{rq-1} $.
For $\lambda>0$, more recent dividends affect prices more than less recent dividends, i.\,e., $0<\beta_{q-1}<....<\beta_{1}<\beta_{0} $.
\end{proposition}

%\begin{proof} %[Proof of Proposition \ref{pro:prices-q2-myopic}]
%	See Appendix \ref{app:results}.
%\end{proof}

This result reflects the fact that the dividends at time $t$ are observed by all generations whereas past dividends are only observed by older generations. 	
At the same time, the extent to which prices depend on the most recent dividends varies with the level recency bias, as shown in the following Lemma.
	
\begin{lemma} \label{lem: compstatics} The effect of the most recent dividend realization on prices,
 	$\beta_0$, is increasing in $\lambda$, with $\lim\limits_{\lambda\rightarrow \infty} \beta_0(\lambda) =  1/(R-1)$	%$\lim\limits_{\lambda\rightarrow \infty} \beta_0(\lambda) = \frac{(Rq)^{-1}}{1 - (Rq)^{-1}}$ 
 	and $\lim\limits_{\lambda\rightarrow \infty} \beta_k(\lambda) = 0$ for $k >0$. 
\end{lemma}
 
 %\begin{proof} %[Proof of Lemma \ref{lem: compstatics}]
% 	See Appendix \ref{app:results}.
% \end{proof}

As $\lambda \rightarrow \infty$, the average weights $w_{k}$ (defined in Proposition \ref{pro: prices_myopic}) converge to $1_{ \{ k =0   \} }$ for all $k=\{0,1,...,K\}$. Therefore, $\beta_{k} \rightarrow 0$ for all $k > 0$ and 
%$\beta_{0} \rightarrow \frac{(Rq)^{-1}}{1 - (R)^{-1}}$. 
$\beta_{0} \rightarrow \frac{1}{R-1}$. 
In other words, under extreme recency bias ($\lambda \rightarrow \infty$), only the current dividend affects prices in equilibrium, while the weights on all past dividends vanish. 

In Section \ref{sec: demographics}, we show that the dependence of prices on more recent dividends is also increasing in the fraction of young agents in the market; that is, $\beta_0$ increases as the relative measure of the youngest cohort in the market increases.

\medskip

These results on price sensitivity to past dividends, as well as the dampening effect of recency bias on cross-sectional heterogeneity, produce a range of asset pricing implications, 
from known puzzles such as the predictability of stock returns and excess volatility to new predictions about the link between asset prices and demographcis. We will derive and test these empirical implications in Section \ref{sec:EmpiricalImplications}.

\subsection{Cross-Section of Asset Holdings}
\label{sec:cross-sec}

%The theoretical framework of experience-based learning also allows us to describe long-term consequences of crises via the channel of stock-market participation. 
Experience-based learning has 
%significant
%striking
%unique
%interesting
distinctive implications 
for the cross-section of asset holdings. We show that positive shocks (booms) induce a larger representation of younger investors in the market, while negative shocks (crashes) have the opposite effect. To illustrate this, we first 
% establish
show that younger investors react more optimistically than older ones to positive changes in recent dividends, and more pessimistically to negative ones.

%\begin{proposition} \label{p: rel_demands-myopic}
%For any $t \in \mathbb{Z}$ and any generation $n$ trading at $t$, there is a threshold time lag $k_0 \leq t-n-1$ such that for dividends that date back up to $k_0$ periods, the risky-asset demand of the younger generation born at $n+1$  responds more strongly to changes than the demand of the older generation born at $n$, while for dividends that date back more than $k_0$ periods the opposite holds, i.e.,
%\begin{enumerate}
%	\item $\frac{\partial x^{n+1}_{t}}{\partial d_{t-k}} \geq \frac{\partial x^{n}_{t}}{\partial d_{t-k}}$ for $0 \leq k \leq k_0$ and
%	\item $\frac{\partial x^{n+1}_{t}}{\partial d_{t-k}} \leq \frac{\partial x^{n}_{t}}{\partial d_{t-k}}$ for $k_0 < k \leq q-1$.
%\end{enumerate}
%\end{proposition}
\begin{proposition} \label{p: rel_demands-myopic}
	For any $t \in \mathbb{Z}$ and any generations $n \leq m$ trading at $t$, there is a threshold time-lag $k_0 \leq t-m-1$ such that for dividends that date back up to $k_0$ periods, the risky-asset demand of the younger generation (born at $m$) responds more strongly to changes than the demand of the older generation (born at $n$), while for dividends that date back more than $k_0$ periods the opposite holds. That is,
	\begin{enumerate}
		\item $\frac{\partial x^{m}_{t}}{\partial d_{t-k}} \geq \frac{\partial x^{n}_{t}}{\partial d_{t-k}}$ for $0 \leq k \leq k_0$ and
		\item $\frac{\partial x^{m}_{t}}{\partial d_{t-k}} \leq \frac{\partial x^{n}_{t}}{\partial d_{t-k}}$ for $k_0 < k \leq q-1$.
	\end{enumerate}
\end{proposition}

%\begin{proof} 
%	See Appendix \ref{app:results}.
%\end{proof}

%In our model, a younger generation puts more weight on current dividends when forming beliefs than the older generation. Hence, when $d_{t}$ increases, the younger generations are ``overly optimistic" relatively to the older generation; and when current dividends decrease, younger agents are more pessimistic about the return of the risky asset than older ones. This difference is zero only when both agents have the same belief formation. 
%$ w(0,\lambda,1)=1$). In the proof of Proposition \ref{p: rel_demands-myopic}, we use Lemma \ref{l: SingleCrossing} to extend this intuition to the more recent dividends, as opposed to just the current one. 
Proposition \ref{p: rel_demands-myopic} establishes that, for any two cohorts of investors, there is a threshold time-lag up to which past dividends are weighted more by the younger generation, and beyond which past dividend realization are weighted more by the older generation. 

In what follows, we extend this insight into predictions about relative stock-market positions. We show that, as a result of the stronger impact of more recent shocks on the beliefs (and thus, demands) of younger generations, the relative positions of the young and the old in the market fluctuate. %Thus, shocks can affect the presence of older versus younger generations in the stock market.
Let us denote the difference
%discrepancy 
between generations $n$ and $n+k$ in terms of their investment in the risky asset, as $\xi(n,k,t) \equiv x^{n}_{t} - x^{n+k}_{t}$. By Proposition  \ref{pro: demandsstatic}, and some simple algebra, it follows that:
\begin{align} \label{eq: ChangeDemands}
	\xi(n,k,t) %& =  \frac{1}{\gamma (1+\beta_{0}) \sigma^{2}} \sum_{j=0}^{t-n} \Big( w(j,\lambda,t-n) - 1_{ \{j \leq t-n-k \} } \cdot w(j,\lambda,t-n-k) \Big) d_{t-j} \\
				   & = \frac{E^{n}_{t}[\theta] - E^{n+k}_{t}[\theta]}{\gamma (1+\beta_{0}) \sigma^{2}} 
				   \qquad \forall k =\{0,...,t-n\}, n =\{t-q+1,...,t\}
\end{align}
%for any $ k =\{0,...,t-n\}$ and $ n =\{t-q+1,...,t\}$. 
This formulation illustrates that the discrepancy between the positions of different generations is entirely explained by the discrepancy in beliefs. 
For instance, if for some $a>0$, $d_{n:t} \approx d_{n+a:t+a}$, then $\xi(n+a,k,t+a) \approx \xi(n,k,t)$.\footnote{This last claim follows since the inter-temporal change in discrepancies between sets of generations of the same age, $\xi(n+a,k,t+a) - \xi(n,k,t)$ for $a>0$, is given by {\footnotesize  
$(\sum_{j=0}^{t-n-k} \{w(j,\lambda,t-n) - w(j,\lambda,t-n-k)\} (d_{t+a-j}-d_{t-j}) +  \sum_{j=t-n-k+1}^{t-n}w(j,\lambda,t-n)  (d_{t+a-j}-d_{t-j}) ) /(\gamma (1+\beta_{0}) \sigma^{2})$.
%$\frac{\sum_{j=0}^{t-n-k} \{w(j,\lambda,t-n) - w(j,\lambda,t-n-k)\} (d_{t+a-j}-d_{t-j}) }{\gamma (1+\beta_{0}) \sigma^{2}}  + \frac{ \sum_{j=t-n-k+1}^{t-n}w(j,\lambda,t-n)  (d_{t+a-j}-d_{t-j}) }{\gamma (1+\beta_{0}) \sigma^{2}}$.
}}

The next result shows that, among generations born and growing up in ``boom times,'' understood as periods of increasing dividends, the younger generations have a relatively higher demand for the risky asset than the older generations. The reverse holds for ``depression babies,'' i.\,e., generations born during times of contraction. In times of depression, the younger generations exhibit a particularly low willingness to invest in the risky asset, relative to older generations born during those times.

\begin{proposition}\label{p: rel_demands-myopic-boom}
%	Suppose $\lambda \geq 0$ and that for $t_{0} \leq t_{1}$, $d_{t_{0}} \leq d_{t_{0}+1} \leq ... \leq d_{t_{1}}$, then for any $n,k$ such that $t_{0} \leq n \leq n+k \leq t_{1}$, 
%   $\xi(n,k,t) \leq 0$ for any $t = n,...,t_{1}$.
	Suppose $\lambda > 0$. Consider two points in time $t_0 \leq t_1$ such that dividends are non-decreasing from $t_0$ up to $t_1$. Then for any two generations $n \leq n+k$ born between $t_0$ and $t_1$, the demand of the older generation for the risky asset ($x^{n}_{t}$) is lower than the demand of the younger generation ($x^{n+k}_{t}$) at any point $n\leq t \leq t_1$, i.e., $\xi(n,k,t)\leq 0$.	
	On the other hand, if dividends are non-increasing, then $\xi(n,k,t) \geq 0$.
\end{proposition}

%\begin{proof} 
%	See Appendix \ref{app:results}.
%\end{proof}

The proposition illustrates that, while boom times tend to make all cohorts growing up in such times more optimistic, the effect is particularly strong for the younger generations. This induces them to be overrepresented in the market for the risky asset. The opposite holds during times of downturn. %, when we should expect younger cohorts to be underrepresented in the market. %All cohorts of investors growing up during such times will be relatively pessimistic about future returns. However, the effect is particularly strong on the younger generation, and as a result they will be underrepresented in the stock market.

\subsection{Trade Volume} \label{sec:tradevolume}

We now study how learning and disagreements affect the volume of trade observed in the market. We consider the following definition of the total volume of trade in the economy:

\begin{equation} \label{e: TradeVolumeDef}
TV_{t}\equiv\left(\frac{1}{q}\sum_{n=t-q}^{t} \left(x_{t}^{n}-x_{t-1}^{n}\right)^{2}\right)^{\frac{1}{2}}
\end{equation}

\noindent with $x^{t}_{t-1}=0$. That is, trade volume is the square root of the weighted sum (squared) of the change in positions of all agents in the economy. Using this definition, we characterize the link between trade volume and belief heterogeneity.

\begin{proposition} \label{l: TradeVolumeM}
	The trade volume defined in (\ref{e: TradeVolumeDef}) can be expressed as
	\begin{equation} \label{e: TradeVolume}
		TV_t = \left( \frac{\chi^2}{q}\sum_{n=t-q}^{t}\left(\left(\theta_{t}^{n}-\theta_{t-1}^{n}\right)-\frac{1}{q}\sum_{\tilde{n}=t-q}^{t}\left(\theta_{t}^{\tilde{n}}-\theta_{t-1}^{\tilde{n}}\right)\right)^{2}
		+ \frac{1}{q}(x_t^t)^2 + \frac{1}{q}(x_{t-1}^{t-q})^2 \right)^{\frac{1}{2}},
	\end{equation}
	where $\chi=\frac{1}{\gamma\sigma^{2}\left(1+\beta_{0}\right)}$,  $\theta_{t-1}^{t}=\theta_t^{t-q}=0$.
\end{proposition}

%\begin{proof}
%	See Appendix \ref{app:results}.
%\end{proof}

Expression \eqref{e: TradeVolume} illustrates that the presence of EBL induces trade through changes in beliefs, which in our framework are driven by shocks to dividends. More specifically, when the change in a cohort's beliefs is different from the average change in beliefs, trade volume increases. That is, trade volume increases in the dispersion of changes in beliefs.

To understand the drivers of trade volume, we need to understand not only the demands of agents that enter and exit the market, but, most importantly, how beliefs across cohorts change in response to a given shock. %We focus our analysis on the trade volume driven by changes in beliefs. In our framework, a shock to dividends impacts the belief of all generations in the market, but the effect on beliefs is stronger for the younger generations. 
From our previous analysis, %of demands, 
it follows that an increase (decrease) in dividends induces trade when it makes young agents more optimistic (pessimistic) than old agents. This mechanism is solely due to the presence of EBL, since it is essential that each generation reacts differently to the same dividend. 

We formalize this insight in the following thought experiment capturing the reaction to a dividend shock that occurs after a long period of stability. 

\medskip

{\bf Thought Experiment.} Suppose that, for $t-t_0 > q$, 
$d_{t_{0}} = d_{t_{0}+1} = ... = d_{t-1} = \bar{d}$ and that $d_{t} \neq \bar{d}$. Hence, all generations alive at time $t-2$ and $t-1$ have only observed a constant stream of dividends $\bar{d}$ over their lifetimes so far. Therefore, $E_{t-2}^{n}\left[d_{t-1}\right]=E_{t-1}^{n}\left[d_{t}\right]=\bar{d}$ for all $n \in \{t-1-q,...,t-1\}$ and thus trade volume in $t-1$ is simply given by the demand of the youngest (entering) and the oldest (exiting) agents. %: $TV_{t-1}=\frac{2}{q} |x_{t-1}^{t-1}|$, since youngest agents increase their demand and oldest agents decrease it by the same amount.

What happens at time $t$, when a dividend $d_t\neq \bar{d}$ is observed? For each generation $n$ trading at time $t$ and at time $t-1$, i.\,e., for $n=\{t-q+1,..,t-1\}$, beliefs are given by $E_{t}^{n}\left[d_{t+1}\right] = w(0,\lambda,t-n) (d_t - \bar{d}) + \bar{d}$ 
%%%
and
$E_{t-1}^{n}\left[d_{t}\right] = w(0,\lambda,t-1-n) (d_{t-1} - \bar{d}) + \bar{d}$, 
%%%
which implies the following change in cohort $n$'s beliefs:
$E_t^n\left[d_{t+1}\right] - E_{t-1}^n\left[d_{t}\right] = w(0,\lambda,t-n) (d_t - \bar{d})$. Trade volume in $t$ is therefore:
\begin{equation} 
TV_t =  \left[ \frac{\chi^2 \left(d_t-\bar{d}\right)^2}{q}  \sum_{n=t-q+1}^{t-1}  \left(w(0,\lambda,t-n)-\frac{1}{q}\sum_{\tilde{n}=t-q+1}^{t-1} w(0,\lambda,t-\tilde{n}) \right)^2 + \frac{1}{q} (x_t^t)^2+\frac{1}{q} (x_{t-1}^{t-q})^2\right]^{\frac{1}{2}}.
\end{equation}

\medskip

This thought experiment pins down two aspects of the link between the volatility in beliefs and trade volume: 
First, the trade volume increases proportionally to the change in dividends, $|d_t-\bar{d}|$, independently of whether the change is positive or negative, and also proportionally to a function that reflects the dispersion of the weights agents assign to the most recent observation in their belief formation process. Second, the increase in trade volume generated by a given change in dividends depends on the level of recency bias of the economy, which is captured by $\lambda$. For example, as $\lambda \rightarrow \infty$, the dispersion in weights decreases as $w(0,\lambda,age)\rightarrow 1$ for all $age\in\{0,...,q-1\}.$ Thus, our results suggest that higher recency bias, $\lambda$, should generate lower trade volume responses for a given shock to dividends, and vice versa. % This intuition can be seen more precisely in our toy model. 

%\medskip

%{\bf Trade Volume in Toy Model.} For $q=2$, the trade volume as stated in Proposition \ref{l: TradeVolumeM} can be re-written as:
%\begin{align*} 
%TV_t &= \frac{1}{ \gamma \sigma^{2} \left(1+\beta_{0}\right)} \left[ \frac{1}{2} \left(\theta_{t}^{t}-\bar{\theta}_{t}\right)^2+\frac{1}{2}\left(\theta_{t}^{t-1}-\bar{\theta}_{t}\right)^2 \right]^\frac{1}{2} \\
%&=\frac{|\theta_{t}^{t}-\theta_{t}^{t-1}|}{2\gamma\sigma^{2}\left(1+\beta_{0}\right)}  =\frac{\left(1-\frac{2^\lambda}{1+2^\lambda}\right)|d_t - d_{t-1}|}{2\gamma\sigma^{2}\left(1+\beta_{0}\right)}
%\end{align*}
%\begin{align*} 
%TV_t =\frac{\left(\frac{1}{1+2^\lambda}\right)|d_t - d_{t-1}|}{2\gamma\sigma^{2}\left(1+\beta_{0}\right)}.
%\end{align*}
%where $\bar{\theta}= 0.5 \left(\theta_{t}^{t} + \theta_{t}^{t-1} \right)$. 
%In this case, we see that trade volume is decreasing in the recency bias, $\lambda$, and increasing in the magnitude of the change in dividends, $|d_t - d_{t-1}|$.

%medskip 

%Overall, experience-based learning not only provides a foundation for the existence of positive trade volume, but also generates testable predictions on the cross-sectional differences in trade volume generated by dividend shocks.
%=========================================================\textbf{

%\section{Market Participation and Equilibrium Prices} \label{sec: demographics}
\section{Market Participation} \label{sec: demographics}

The results derived so far illustrate a key feature of experience-based learning: The demographic structure of an economy, and in particular the cross-sectional composition of investors, affect equilibrium prices, demand, and trade volume in a predictable direction. %Since experiences vary across cohorts, a given dividend shock will have different long-term consequences depending on the demographic structure of the economy at the time of the shock. Our model allows us to predict the direction of such differences.

In this section, we explore the link between market demographics and financial market outcomes by considering an \emph{unexpected} increase in the fraction of young market participants, e.\,g., due to a baby boom or a generation-specific event drawing a certain generation into the stock market.\footnote{We have also analyzed the implications  of a growing market population, as opposed to a one-time market demographic shock. In Online Appendix \ref{oa:demographics}, we show that population growth generates a positive trend in prices, which is independent of experience effects: 
The growing mass of agents increases the demand for the risky asset, and hence prices adjust to clear markets, since risky-asset supply is assumed to be constant. While the positive trend is independent of experience effects, experience-based learning does affect the path of the prices fluctuating around this trend. In particular, we find that the relative reliance of prices on the most recent dividend is increasing in the population growth rate.}  
The goal of this exercise is to understand how a larger fraction of young market participants affects market dynamics. 

For ease of illustration, we focus again on our $q=2$ economy. %For example, a war reduces the size of certain cohorts of the population, while a baby boom increases it. 
We denote the mass of young agents at any time $t$ by $y_t$, and the total mass of agents at $t$ by $m_t = y_t + y_{t-1}$.  We consider a one-time unexpected (exogenous) shock to the mass of young agents in the market at time $\tau$.% We assume that the fraction of young agents that participate in the market in year $\tau$ is exogenously determined (e.g., a baby boom).
\footnote{In reality, the participation of young agents in the market could also be determined endogenously (e.\,g., by entry costs). While the forces described in this section would still be present in such scenario, other forces may be at play as well. The study of these interactions is out of the scope of this paper.}
For all $t<\tau$ and $t>\tau+1$, instead, $y_t=y$ and thus $m_t=2 y =m $. 

We know from our previous results that when the market has equal-sized cohorts, prices are given by 
$p_{t} = \alpha+\beta_{0}d_{t}+\beta_{1}d_{t-1}$, with $\{\alpha,\beta_{0},\beta_{1}\}$ given by \eqref{e: AlphaToyModel}-\eqref{e: Beta1ToyModel}. Here, prices follow this path for $t> \tau+1$ and, since the shock at time $\tau$ is unexpected, for $t < \tau$ as well. %\footnote{We use the assumption that the shock is unexpected to construct the series in Figure \ref{f: Demographics}; but none of the result for $t \geq \tau$ depend on this assumption.} 
For these time periods, the market is as described in Section \ref{sec: ToyModel}. We are left to characterize demands and prices for $\tau$ and $\tau+1$, when the larger young generation enters the market and when this generation becomes old, respectively. We make the following  guesses: 
\begin{align}
	p_{\tau} & =  a_{\tau}+b_{0,\tau}d_{\tau}+b_{1,\tau}d_{\tau-1},\\
	p_{\tau+1} & =  a_{\tau+1}+ b_{0,\tau+1}d_{\tau+1}+b_{1,\tau+1}d_{\tau}.
\end{align}

We solve the problem by backwards induction. Note that the form of agents' demands remains unchanged. By imposing market clearing in $\tau+1$, with mass $y$ of young agents and $y_\tau$ of old agents, and using the method of undetermined coefficients we obtain
%$$1  =  y\frac{E_{\tau+1}^{\tau+1}\left[p_{\tau+2}+d_{\tau+2}\right]-Rp_{\tau+1}}{\gamma\left(1+\beta_{0}\right)^{2}\sigma^{2}}+y_\tau\frac{E_{\tau+1}^{\tau}\left[p_{\tau+2}+d_{\tau+2}\right]-Rp_{\tau+1}}{\gamma\left(1+\beta_{0}\right)^{2}\sigma^{2}}. $$
%$$\Rightarrow p_{t+1}  =  \frac{1}{r}\left[\alpha-\frac{\gamma\left(1+\beta_{0}\right)^{2}\sigma^{2}}{n_t}\right]+\frac{1}{r}\left[\left(1+\beta_{0}\right)\frac{(n_t-y_t)+y_t\omega}{n_t}+\beta_{1}\right]d_{t+1}+\frac{1}{r}\left[\left(1+\beta_{0}\right)\frac{y_t}{n_t}\left(1-\omega\right)\right]d_{t}$$
% where $\omega \equiv \frac{2^\lambda}{1+2^\lambda}$. 
%Our guess is verified, and we obtain the coefficients for the price function in $\tau+1$:
\begin{align}
		a_{\tau+1} & =  \alpha\frac{1}{R}\left[1+\frac{R-1}{m_\tau}\right], \nonumber %\label{eq: bplus}
		\\
		b_{0,\tau+1} & =  \beta_{0}\left[1+\frac{1}{R}\left(\frac{m_\tau-y_\tau}{m_\tau}+\frac{y_\tau}{m_\tau}\omega-\frac{y}{m}(1+\omega)\right)\right]+\frac{1}{R}\left(\frac{m_\tau-y_\tau}{m_\tau}+\frac{y_\tau}{m_\tau}\omega-\frac{y}{m}(1+\omega)\right), \nonumber \\
		{b}_{1,\tau+1} & = \beta_{1}\frac{y_\tau}{m_\tau}\frac{m}{y}, \nonumber
	\end{align}
where $\omega \equiv \frac{2^\lambda}{1+2^\lambda}$ and $m_\tau= y + y_\tau$.
Note that for $y_\tau=y$, the coefficients are as in the baseline model \eqref{e: AlphaToyModel}-\eqref{e: Beta1ToyModel}. The above expressions show that the total mass of agents $m_\tau$ only affects the price constant, while the price loadings depend %solely 
on the fraction of young agents in the market, $y_\tau/m_\tau$. We impose market clearing in $\tau$, with mass $y_\tau$ of young agents and $y$ of old agents. Using the method of undetermined coefficients, we obtain %Market clearing at time $\tau$ is
%$$	1  =  y_\tau\frac{E_{\tau}^{\tau}\left[p_{\tau+1}+d_{\tau+1}\right]-Rp_{\tau}}{\gamma\left(1+b_{0,\tau+1}\right)^{2}\sigma^{2}}+y\frac{E_{\tau}^{\tau-1}\left[p_{\tau+1}+d_{\tau+1}\right]-Rp_{\tau}}{\gamma\left(1+b_{0,\tau+1}\right)^{2}\sigma^{2}}$$ 
%$$\Rightarrow p_{t} =  \frac{1}{r}\left[{a}^o-\frac{\gamma\left(1+{b}^o_{0}\right)^{2}\sigma^{2}}{0.5+m}+\left[\left(1+{b}^o_{0}\right)\left(\frac{m+0.5\omega}{0.5+m}\right)+{b}^o_{1}\right]d_{t}+\left(1+{b}^o_{0}\right)\frac{0.5\left(1-\omega\right)}{0.5+m}d_{t-1}\right]$$
%Our guess is verified, and we obtain the following coefficients for the price function at $\tau$:
	\begin{align}
		a_{\tau} & =  \frac{1}{R}\left[a_{\tau+1}-\frac{\gamma\left(1+b_{0,\tau+1}\right)^{2}\sigma^{2}}{m_\tau}\right], \nonumber %\label{eq: b0} 
		\\
		b_{0,\tau} & =  \frac{1}{R}\left(1+b_{0,\tau+1}\right)\left(\frac{y_\tau}{m_\tau}+\frac{m_\tau-y_\tau}{m_\tau}\omega\right)+\frac{1}{R^{2}}\left(1+\beta_{0}\right)\frac{y_\tau}{m_\tau}\left(1-\omega\right), \nonumber \\
		b_{1,\tau} & =  \frac{1}{R}\left(1+b_{0,\tau+1}\right)\frac{m_\tau-y_\tau}{m_\tau}\left(1-\omega\right). \nonumber
	\end{align}

\begin{figure*}
	\centering     %%% not \center
	\subfigure[Price Loading on $d_t$]{\label{fig2a}\includegraphics[width=0.48\textwidth]{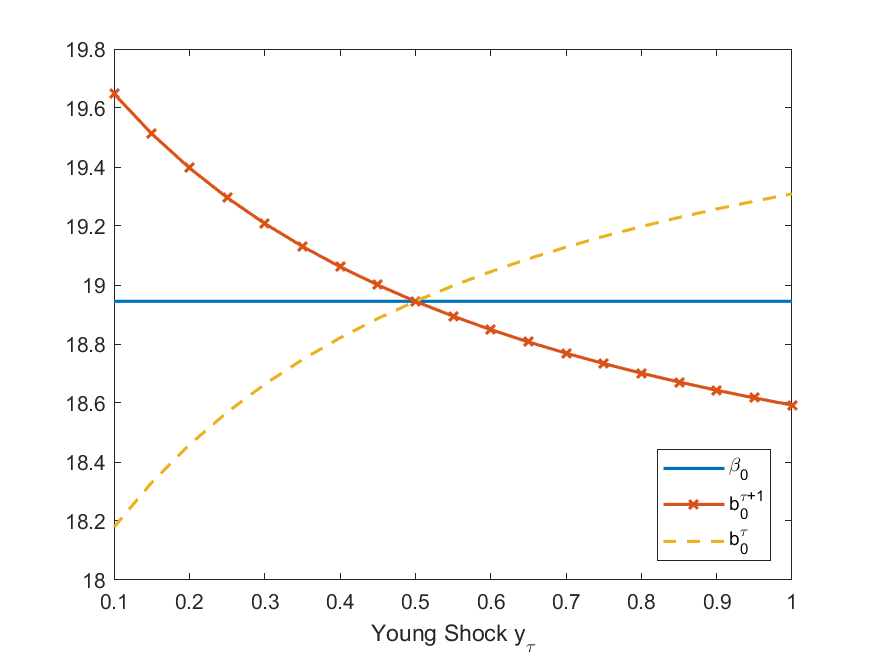}}\hfill
	\subfigure[Price Loading on $d_{t-1}$]{\label{fig3b}\includegraphics[width=0.48\textwidth]{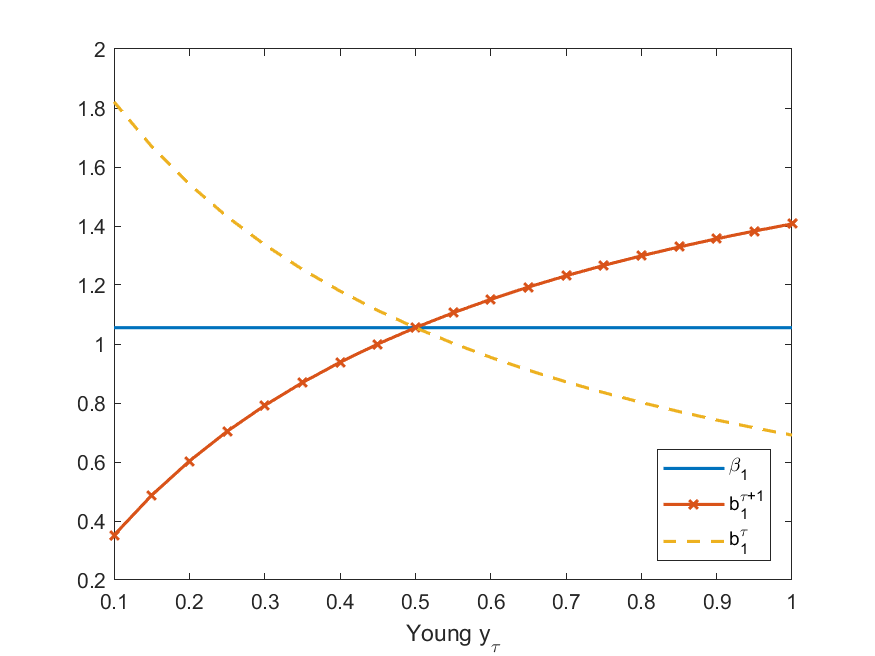}}\hfill
	\caption{Demographic Shocks and Price Coefficients.}
	\vspace{-0.2cm}
	\caption*{\textit{Notes}. This figure plots coefficients $\{\beta_0,\;b_{0,\tau},\;b_{0,\tau+1}\}$ in Panel (a) and $\{\beta_1,\;b_{1,\tau},\;b_{1,\tau+1}\}$ in Panel (b) as a function of the demographic shock $y_\tau$. The results are for $y=0.5$, $\lambda=3$, and $R=1.05$.}
	\label{f: DemographicsBetas}
	\vspace{0.4cm}
\end{figure*}

Figure \ref{f: DemographicsBetas} shows how the reliance of prices on past dividends changes with the fraction of young agents in the market at time $\tau$. We see that as the fraction of young people in the market increases ($y_{\tau}>0.5$), the more current dividends matter more relative to past dividends for the determination of prices; i.e., $b_0^{\tau}$ increases while $b_1^{\tau}$ decreases. Consistent with this, when the $\tau$-generation becomes old, prices depend less on contemporaneous dividends and more on past dividends; i.e., $b_0^{\tau+1}$ decreases while $b_1^{\tau+1}$ increases. Finally, an increase in the overall market population, and thus demand for the risky asset, generates a level increase in prices captured an increase in the price constant both at $\tau$ and $\tau +1$. All predictions are reversed when the fraction of young agents in the market decreases ($y_{\tau}<0.5$). These results are consistent with \citeN{collin2016asset}, who show both theoretically and empirically that the sensitivity of the price-dividend ratio to macro-shocks increases with the relative fraction of young market participants.

\section{Empirical Implications} \label{sec:EmpiricalImplications}

In this section, we analyze the empirical implications of our model. The analysis consists of two approaches.
First, we 
% show that the model is able to account for 
turn to asset-pricing features established in prior empirical literature: the predictability of stocks returns, 
%(\citeN{fama1988dividend}), 
the predictability of the dividend-price ratio, and the excess volatility puzzle. %(\citeN{LeRoy2005} and \citeN{Shiller1981}). 
We show that the experience-based learning model is able to quantitatively match these empirical findings, and that it generates refined predictions about their relation with the demographic composition of investors.
Second, 
we test the novel predictions generated by our model regarding 
the implications of the demographic composition on the predictability of the price-dividend ratio, trade volume, and the cross-section of asset holdings.
We show that these predictions are in line with evidence from micro-level data in the Survey of Consumer Finances (SCF) and the Center for Research in Security Prices (CRSP).

\subsection{Quantitative Implications for Asset-Pricing Moments} \label{sec: quantitative}
We first show that experience-based learning can explain 
% the above-mentioned 
several key asset-pricing puzzles.
%: the predictability of returns and of the dividend-price ratio, as well as the excess volatility puzzle. 
As the CARA-Normal framework is not well suited for a thorough calibration exercise, we follow the approach of \citeN{Campbell_Kyle1993} and \citeN{barberis2015x}, among others, to compute the moments of interest generated by our model and contrast them with the data. As in these papers, we define quantities in terms of differences rather than ratios, since variables in our model proxy for the log of prices and dividends in the data. Thus, in this section, capital letters $P$ and $D$ are used to denote prices and dividends, respectively, while small letters denote their logs, $p = \log(P)$ and $d = \log(D)$. For example, instead of stock returns we measure price changes $\Delta p$, and instead of the price-dividend ratio $P/D$ we study the difference $p - d$. 

A distinguishing feature of our model is that it establishes a link between the age profile of agents participating in the stock market and the factors that determine prices. Another feature of our model is the small number of parameters to be set for generating numerical results.\footnote{The link between demographics and price features is also studied in \cite{collin2016asset}.} 
% All of our parameters are either observable (such as the risk-free rate), or have been can be estimated from micro-data (as the recency bias). 
Following \citeN{barberis2015x}, we choose the following parameter values for our numerical solutions: the gross risk-free interest rate is $R = 1.05$, the volatility of dividends is $\sigma = 0.25$, and the coefficient of risk-aversion is $\gamma = 2$. We show our estimates for $\lambda \in \{1,3\}$ and for $q \in \{2,40\}$.

\medskip

%\noindent
\textbf{Predictability of Excess Returns.}
%\subsubsection{} 
A prominent stylized fact about stock-market returns, established by \citeN{CampbellShiller1988}, is that the dividend-price ratio predicts future returns with a positive sign.
%A feature of the stock market is that the dividend-price ratio predicts future returns with a positive sign \cite{CampbellShiller1988}. 
Experience-based learning rationalizes such predictability and, at the same time, limits it to those dividend realizations experienced by the oldest cohort participating in the market.

In order to relate the predictability generated in our model to the existing empirical evidence, and to show how it varies with the demographic composition of investors, we calculate the following measure of co-movement between the analogues of the dividend-price ratio and returns, namely, between dividend-price differences $d_t-p_t$ and price changes $p_{t+z} - p_{t}$ over return horizon $z$:
\begin{align}
\it{B}_t^R(z) \equiv \frac{Cov(d_t-p_t,p_{t+z}-p_t)}{Var(d_t-p_t)}.
\end{align}
We compute $\it{B}_t^R(z)$ 
using equation (\ref{e: EqPrices}) from Proposition \ref{pro: prices_myopic}.

We first calculate $\it{B}_t^R(z)$ for different horizons $z$. 
Panel (a) of Figure \ref{f: PredReturns} plots $\it{B}_t^R(z)$ for $z$ ranging from $1$ to $42$, and different levels of recency bias, $\lambda \in \{1,3\}$, in an economy with $q=40$. Given the number of cohorts, the obtained co-movements can be interpreted as annual; that is, $z=1$ can be interpreted as a one year horizon. As the panel shows, the experience-based learning model generates a positive (and strong) relation between the dividend price ratio and returns, which increases with the return horizon. These predicted patterns are consistent with the empirical findings described in \citeN{cochrane2011presidential}.  
% (The leveling of the slope around $z=5$ reflects  $q=5$.) 
We note that the predictability of excess returns under EBL is an equilibrium phenomenon that stems solely from our learning mechanism and not from, say, a built-in dependence on dividends or past returns. Similar to prior theoretical approaches, such as the over-extrapolation model of \shortciteN{barberis2015x} and \shortciteN{barberis2016extrapolation}, our explanation relies on agents' overweighting recent realizations.

\begin{figure*}
	\centering     %%% not \center
	\subfigure[$B_t^R(z)$ with $q = 40$]{\label{fig3a}\includegraphics[width=0.32\textwidth]{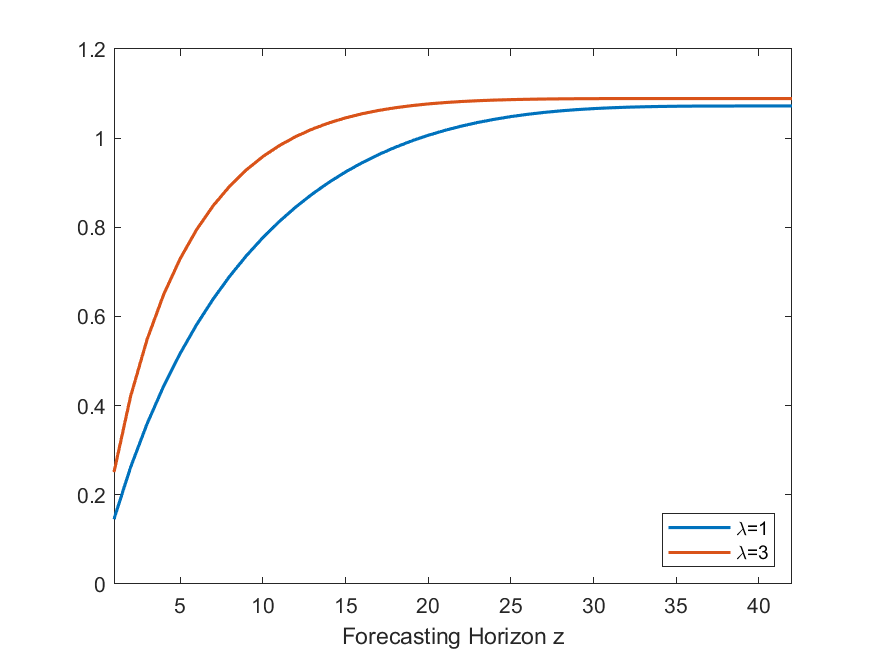}}\hfill
	\subfigure[$B_t^R(1)$ with $q = 2$]{\label{fig3b}\includegraphics[width=0.32\textwidth]{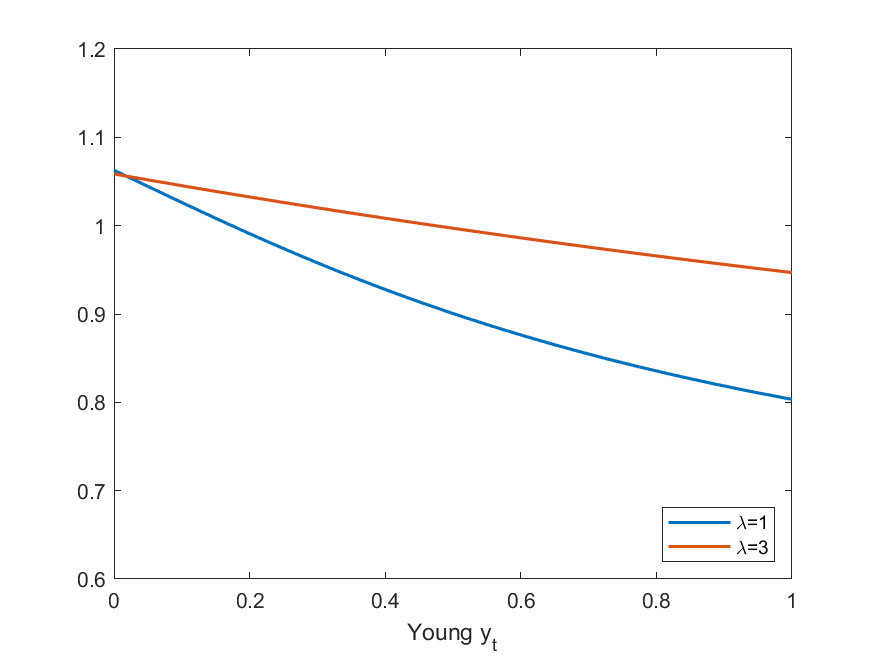}}\hfill
	\subfigure[$B_{t-1}^R(1)$ with $q = 2$]{\label{fig3c}\includegraphics[width=0.32\textwidth]{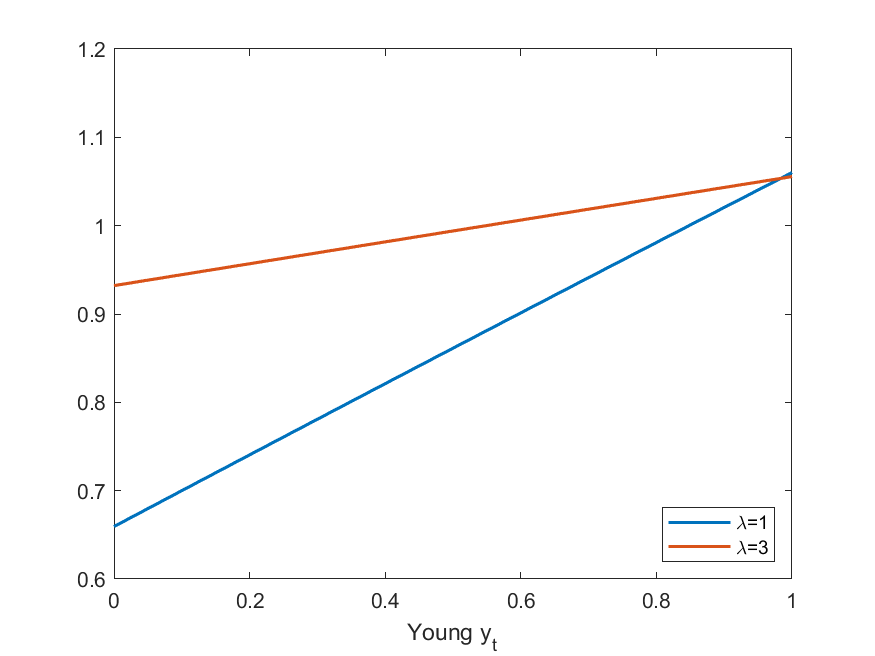}} 
	 \hspace*{\fill}%
	\caption{Predictive Power of $D_t-P_t$ for $P_{t+z}-P_t$.}
	\label{f: PredReturns}
	\vspace{-0.2cm}
	\caption*{\textit{Notes}. This figure plots the coefficient $B_t^{R}$ for two levels of recency bias, $\lambda \in \{1,3\}$. Panel (a) shows how $B_t^{R}(z)$ varies with the return horizon $z$, for $q=35$. Panel (b) shows how $B_t^{R}(1)$ varies with the fraction $y_{t}$ of young agents in the market at time $t$, and panel (c) with the fraction $y_{t+1}$ of young agents in the market at time $t+1$, in both cases for $q=2$. The results are calculated for $y=0.5$ and $R=1.05$.}
	\vspace{0.4cm}
\end{figure*}

Our model has the additional implication that different demographic structures generate different $\beta$'s, which directly determine the level of predictability (or extrapolation) of stock returns. 
We show this by studying how the coefficient $B^R_t(1)$ varies with the fraction of young market participants.
Here, we use the results from Section \ref{sec: demographics} on the effect of a one-time unexpected shock to the mass of young agents, so we focus on an economy with $q=2$. Note that with $q=2$ the co-movements cannot be interpreted as annual; in particular, $z=1$ captures approximately a 15-20 year horizon. We assume that the fraction of young agents equals $y$ at all times before and after $t$, but there is a one-time exogenous 
%increase 
shock 
in $t$, resulting in 
%$y_t>y$. 
$y_t\neq y$.
The resulting variation in next period's return predictability (based on this period's dividend-price ratio), $B^R_t(1)$, and in the current period's return predictability (based on last period's dividend-price ratio), $B^R_{t-1}(1)$, are shown in panels (b) and (c) of Figure \ref{f: PredReturns}.
For both graphs we fix $y$ at $0.5$, and we plot $y_t$ over a range 
%$(0.5,1]$.
$[0,1]$.

%As Figure \ref{f: PredReturns} shows, we find that
As the plots show, the predictability of next period's return, $B^R_t(1)$, decreases in the number of young market participants (panel b), while the predictability of this period's return, $B^R_{t-1}(1)$, increases in their number (panel c). The key channel is the differential impact on the variance of the dividend-price ratio. In both cases the covariance between future returns and the dividend-price ratio %(in the numerator of $B^R_t(1)$) 
increases in the fraction of young agents; but only the variance of the dividend-price ratio %(in the denominator of $B^R_t(1)$) 
at time $t$ increases such that it off-sets the increase in covariance. 
%Consistent with this, the predictability of returns between $t-1$ and $t$ increases in the fraction of young agents in the market at time $t$. 
Thus, return predictability is affected by the demographic composition of market participants, and 
the effect %of young market participants on return predictability 
is sensitive to the timing of the participation shock.  
With a larger generation of market participants coming in at $t$, the return experienced in that period is more predictable.

%$B^R_{t}(1)$ the co-movement between this period's dividend-price ratio (difference) and the price change between this period and next period.

%$B^R_{t-1}(1)$ the co-movement between last period's dividend-price ratio (difference) and the price change between last period and this period.

%\citeN{cassella2015} find suggestive evidence that the dividend-price ratio is a better predictor of future returns when the fraction of young agents in the market is relatively large. The authors first provide evidence that, when the level of extrapolation bias is high in the market, the predictive power of the price-dividend ratio for future returns goes up. Second, they find a positive relation between their market-wide measure of extrapolation and the relative participation of young versus old investors in the stock market. That is, as the number of young agents in the market increases, the dividend-price ratio is more correlated with future excess returns. Our model with EBL goes beyond rationalizing evidence on agents extrapolating from past dividends (cf. also \citeN{GreenwoodShleifer2014}) in that it puts structure on the extent of such extrapolation (e.g., market demographics) that is aligned with empirical observations.

\medskip
\textbf{Predictability of Price-Dividend Ratio}. In addition to the predictability %of the dividend-price ratio for 
of returns, we can also compute the predictability of the price-dividend ratio implied by the model.
That is, we relate past P/D ratios to future P/D realizations, and analyze the persistence of the price-dividend ratio. 
In particular, we study how this predictability of P/D ratios varies with the investment horizon and with the fraction of young people in the market. Our measure of predictability is constructed as follows:
\begin{align}
B^{PD}_{t}(z) = \frac{cov\left(p_{t+z} - d_{t+z},p_{t}-d_{t}\right)}{var\left(p_{t}-d_{t}\right)} %= \frac{b_{1,\tau+1}\cdot b_{0,\tau}}{(b_{0,\tau}-1)^2 \cdot b_{1,\tau}^2} 
\end{align}

\begin{figure*}
	\centering     %%% not \center
	\subfigure[$B_t^{DP}(z)$ with $q = 40$]{\label{fig4a}\includegraphics[width=0.48\textwidth]{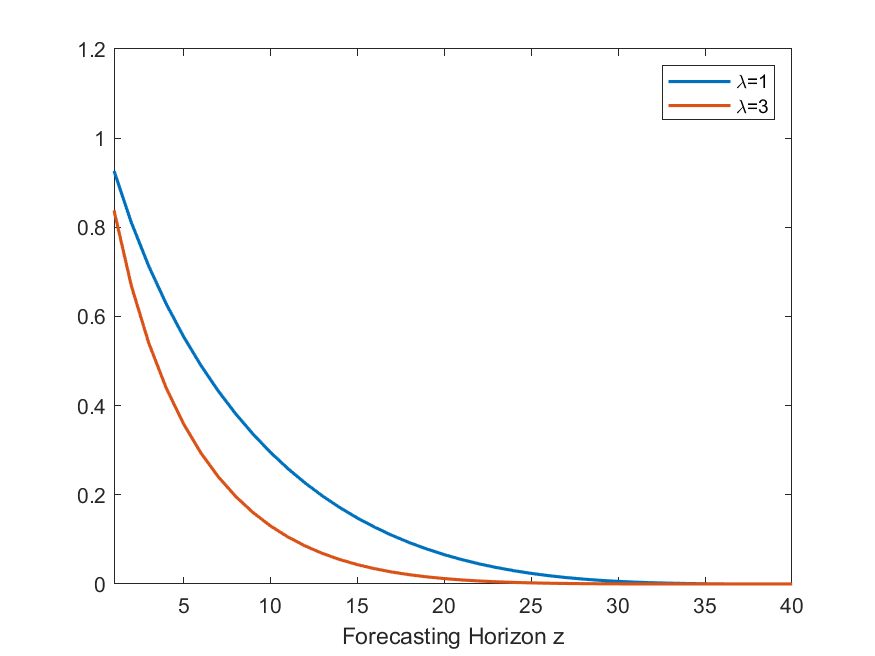}}\hfill
	\subfigure[$B_t^{DP}(1)$ with $q = 2$]{\label{fig4b}\includegraphics[width=0.48\textwidth]{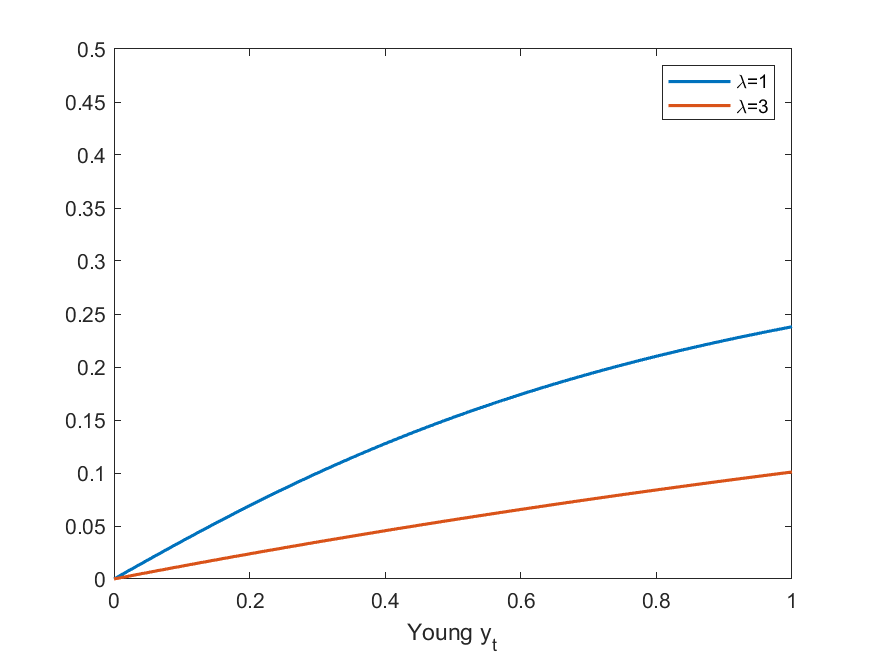}}\hfill
	\hspace*{\fill}%
\caption{$p_t - d_t$ autocorrelation.}
\vspace{-0.2cm}
\caption*{\textit{Notes}. This figure plots the coefficient $B_t^{PD}$ for two levels of recency bias, $\lambda \in \{1,3\}$. Panel (a) shows how $B_t^{PD}(z)$  varies with the investment horizon $z$ for $q=40$. Panel (b) shows how $B_t^{PD}(1)$  varies with the fraction $y_{t}$ of young agents in the market at time $t$ for $q=2$. The results are calculated for $y=0.5$ and $R=1.05$.}
\label{f: Participation_DPRatio}
\vspace{0.4cm}
\end{figure*}

We first calculate how $B_t^{PD}$ varies with the horizon $z$. Panel (a) of Figure \ref{f: Participation_DPRatio} displays how $B_t^{PD}$ varies for different horizons $z$, and for different levels of recency bias, $\lambda \in \{1,3\}$, in an economy with $q=40$. The large $q$ allows us to relate the obtained correlations to annual correlations. As in the data, we obtain that the $P/D$ is highly autocorrelated at short lags, with the autocorrelation being zero at longer horizons. From Panel (b) we see that as we reduce the number of cohorts in the market (or their horizon), autocorrelations are lower.  Furthermore, $B_t^{PD}$ decreases in the extent of recency bias present in the population for all $q$. 

We then turn to the demographic structure, using again the results from Section \ref{sec: demographics}. As shown in panel (b) of  Figure \ref{f: Participation_DPRatio}, $B_t^{PD}$ increases with the fraction of young agents in the market at time $t$, and the effect is weaker under higher recency bias. Furthermore, $B_t^{PD}$ increases with the number of cohorts in the market. A direct implication is that the dividend-price ratio is positively correlated only with lagged realizations where the number of periods lagged is below the number of cohorts in the market. %, and that the correlation vanishes as $\lambda \rightarrow \infty$ for any number of cohorts. 

\textbf{Price Dynamics.}
%\subsubsection{}
A third set of asset pricing implications are related to the dynamics of prices, and in particular the excess volatility puzzle. %\citeN{Campbell_Kyle1993} estimate the volatility of dividends to be $0.032$ that of the change in prices.
As is standard in the literature, we analyze the volatility of the log price growth and of the log price-dividend ratio, and the volatility of log prices relative to that of log dividends, both in the model and in the data. To do so, we use historical price and dividend data from Robert Shiller's website, where all log series are de-trended. Our stylized model generates ample volatility relative to our benchmark economy and to the data. The main reason being that agents' beliefs are extremely volatile when they do not put any weight on their prior belief ($\tau=0$ in \eqref{eq: PosteriorMeanFormula}), which would operate as an anchor. To highlight the quantitative effects of prior beliefs, we compute the model generated volatilities when we vary the importance that agents assign to their prior beliefs.\footnote{The solution to the model with prior beliefs is described in detail in Appendix \ref{app: PriorBeliefs}.} Table \ref{f: ExcessVolatility} presents our results for an economy with $\lambda=1$, $q=40$, and for different levels of prior relevance, captured by $\tau$.    

%Note that $\sigma(P) = \sigma (\sum_{k=0}^{q-1}\beta_k^2)^{\frac{1}{2}} $, and thus we should also expect prices to be more volatile when the fraction of young agents in the market increases (because $\beta_0$ increases, as shown in Section \ref{sec: demographics}). 

\begin{table}[t]
	\centering
	\includegraphics[width=0.7\linewidth]{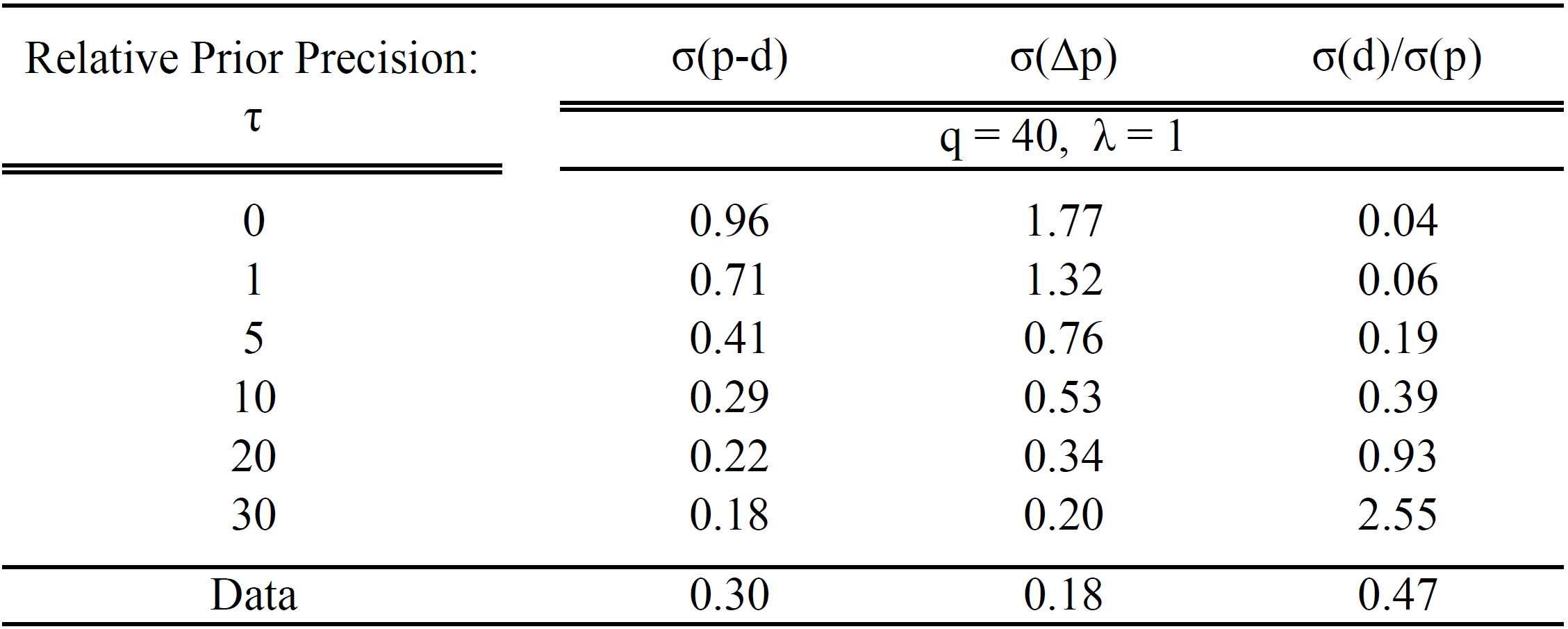}
		\caption{Excess Volatility}
	\label{f: ExcessVolatility}
\end{table}

We see that experience-based learning generates ample excess volatility in prices, returns, and price-dividend ratios. This can be seen by comparing the data with the model with no prior beliefs ($\tau=0$). However, if we allow agents to have prior beliefs, the model is able to generate moments more in line with the data. From these findings, we conclude that experience-based learning has the ability to generate volatility in line with the data.

\subsection{Demographics on Price-Dividend Predictability}

The predictability results in the previous subsection are consistent with the findings of \citeN{cassella2015}, who find a positive relation between their market-wide measure of return extrapolation and the relative participation of young versus old investors in the stock market. Our model of experience-based learning goes beyond a rationalization of the evidence on agents extrapolating from past dividends (cf.\;also \citeN{GreenwoodShleifer2014}). It puts structure on the extent of such extrapolation exhibited by different market participants and links it to market demographics.
We now bring this prediction to the data and show that is aligned with empirical observations. 

We want to test whether the predictive power of the lagged P/D ratios for the current one depends on the relative representation of younger versus older generations in the market in the manner predicted by the model.
Experience-based learning predicts that the correlation between future and current lags is higher when the current share of young market participants is large. 
Moreover, the model generates the heuristic that young people put little weight on observations of the ``distant" past (cf.\;Proposition \ref{p: rel_demands-myopic}).

In order to test these predictions, we regress the log of the P/D ratio onto lags of itself interacted with a dummy variable that indicates a larger presence of young people in the market. 
% One complication is that the empirical test 
In order to model the dynamics
of the P/D process, we depart from the standard linear AR models and postulate a Markov-Switching Regime (MSR) model, which allows us to capture richer, non-linear dynamics in a tractable way.\footnote{For a more thorough discussion of MSR see \citeN{hamilton1989new}.} The regression model is thus given by 
\begin{align}
	%\log \frac{P_{t+1}}{D_{t+1}} = \mu(S_{t+1}) + \sum_{j=1}^{3} \log \frac{P_{t+1-j}}{D_{t+1-j}} (\beta_{j} + \delta_{j} X_{t}) + \sigma(S_{t+1})\epsilon_{t+1}
	p_{t+1}-d_{t+1} = \mu(S_{t+1}) + \sum_{j=1}^{3} (p_{t+1-j}-d_{t+1-j}) (\beta_{j} + \delta_{j} \times Y_{t+1}) + \sigma \epsilon_{t+1} , \label{eqn:MSR}
\end{align}
where $p_t$ and $d_t$ denote the log of dividends and prices at time $t$, respectively, $S_{t+1} \in \{0,1\}$ is an unobserved state that evolves according to a Markov transition kernel $Q$; $Y_{t}$ is a dummy variable that takes value 1 if the share of young generations participating in the market at time $t$ is large relative to the participation of older generations, and 0 otherwise; 
%$\sigma^{2}$ is the conditional variance of $P_{t+1}-D_{t+1}$; 
and we assume $\epsilon_{t+1} \sim N(0,1)$. 
The parameters, $(\{\mu(s)\}_{s \in \{0,1\}},\sigma,Q,\{\beta_{j},\delta_{j}\}_{j=1}^{3})$,
are jointly estimated using maximum likelihood (see, e.\,g., \citeN{HAMILTON1994} for details). 

\begin{table} 
	
	\caption{Markov-Switching Regime (MSR) Model}
	\vspace{-0.4cm}
%	\begin{small}
		\begin{tablenotes}
			\item Estimation results for model specification \eqref{eqn:MSR}, where $p_t - d_t$ is the log of the price-to-dividend ratio, and regressed onto lags of itself interacted with a demographic dummy variable. $Y_t$ is the fraction of young people, which we define as an indicator equal to 1 when the fraction of investors under 50 is larger than 0.5 (in column 1), or as an indicator equal to 1 when the fraction of wealth of investors below 50 is larger than their 1960-2013 sample average (in column 2).The demographic data including age and wealth (liquid assets) of stock market participants is from the SCF, stock data from Robert Shiller's website.
		\end{tablenotes}
%	\end{small}
	\vspace{0.5em}
	\begin{center}
		\begin{tabularx}{380pt}{X c*{2}{c}}
			
			\hline\hline
			&\multicolumn{2}{c}{Dependent variable: ${p}_{t} - {d}_{t}$}\\
			\cline{2-3}
			&\multicolumn{1}{c}{(1)}&\multicolumn{1}{c}{(2)}\\
			&\multicolumn{1}{c}{$Y_t$ age-based     }&\multicolumn{1}{c}{$Y_t$ age/wealth based}\\
			
			\hline
			&        &    \\
			$\delta_{1}$   & 0.701**	& 0.475* \\
			&(0.154)	&(0.252) \\[2mm]
			$\delta_{2}$ &-0.013	&-0.115\\
			&(0.146)	&(0.366) \\[2mm]
			$\delta_{3}$ &-0.745**	&-0.329 \\
			&(0.115)	&(0.232) \\[2mm]
			$\beta_{1}$  &0.377**	&0.622** \\
			&(0.120)	&(0.159) \\[2mm]
			$\beta_{2}$ &-0.216**	&-0.074 \\
			&(0.088)	&(0.136) \\[2mm]
			$\beta_{3}$ &0.714**	&0.249**   \\
			&(0.093)	&(0.099) \\[2mm]
			$\mu(S_{1})$&5.089**	&5.741** \\
			&(1.554)	&(1.812) \\[2mm]
			$\mu(S_{2})$&19.450**	&18.350** \\
			&(3.070)	&(4.768)  \\[2mm]
			$\sigma$&3.812	&4.343 \\
			&(0.388)	&(0.600)  \\[2mm]
			$Q_{11}$&0.956	&0.978    \\
			&(0.026)	&(0.017)  \\[2mm]
			$Q_{21}$&0.365	    &0.206 \\
			&(0.204)	&(0.154)  \\
			
			\hline
			
			\textit{N}      &     51   &     51    \\
			\hline\hline
			
		\end{tabularx}
	\end{center}
	\begin{small}
		\begin{tablenotes}
			\item Standard errors in parentheses. * significant at 10\%; ** significant at 5\%.
		\end{tablenotes}
	\end{small}
	
	\label{tab:return}
\end{table}

We consider two dummies for the relative representation of younger generations in the market. First, we compute the ratio of investors who are less than 50 years old in the total population, and construct an indicator that equals 1 if their share is bigger than 50\% (or, for robustness, bigger than $55\%$ or $60\%$).
Second, we calculate young investors' share of liquid wealth, and use an indicator that equals 1 if their liquid-wealth share is above its sample average (or, for robustness above 90\% or 110\% of the sample average). Details on the variable construction and robustness checks are in Online-Appendix \ref{sec:Emperical Analysis}. 
%All results are in Online-Appendix Table \ref{tab:onapp_robustness}. 

The theoretical prediction of our model is that the correlation between future and current lags should be higher when the current share of young market participants is large. This translates into the hypothesis that $\delta_{1}>0$ in the estimation model in (\ref{eqn:MSR}). 

The estimated values are reported in Table \ref{tab:return}. In column (1), we use the fraction of young people in the population, and in column (2) the fraction of their wealth to proxy for the relative representation of younger people in the market.
In both cases, the estimates provide evidence in favor of the model hypothesis. We estimate a positive $\delta_{1}$ coefficient, which is either significant at the 5\% or at the 10\% level. Moreover, considering all three coefficients $(\delta_{i})_{i=1}^{3}$ jointly, a roughly ``decreasing" pattern emerges: $\delta_{1}$ is typically positive, $\delta_{2}$ is typically non-significant, and $\delta_{3}$ is negative or insignificantly negative, consistent with the heuristics that young people put little weight on observations of the ``distant" past. Thus, in periods when their share is relatively large, the correlation between future and distant past values is weakened.

\subsection{Cross-Section of Asset Holdings and Trade Volume} \label{sec: Facts}

We now turn to the novel empirical predictions of the experience-based learning model about the cross-section of equity holdings and stock turnover.  
%While some predictions are harder to test with existing data, and others can be attributed to several explanations (e.g., higher weights on the most recent dividend, relative to previous dividends, in determining prices), 
We investigate two sets of predictions that are directly testable and jointly hard to generate by alternative models. 

The first prediction is that cross-sectional differences in the demand for risky securities reflect cross-sectional differences in lifetime experiences of risky payoffs. That is, cohorts with more positive lifetime experiences are predicted to invest more in the risky asset than cohorts with less positive experiences (Proposition \ref{pro: demandsstatic}). We test this both in terms of stock-market participation (extensive margin) and in terms of the amount of liquid assets invested in the stock market (intensive margin). The second prediction is that changes in the cross-section of experience-based beliefs generate trade (Proposition \ref{l: TradeVolumeM}).

To test these model predictions, we combine historical stock returns data from Robert Shiller's website with SCF data on stock holdings and CRSP data on stock turnover.
% The key explanatory variable is a measure of lifetime experiences of dividends. 
The key explanatory variable is a measure of cohorts' lifetime experiences of risky-asset payoffs. Theoretically, dividends in the Lucas-tree economy capture the performance of the risky asset, or the stock market. 
Empirically dividends payments do not necessarily reflect how well firms are doing. For example, firms have incentives to smooth dividends, and also to retain earnings rather than distribute them. In other words, dividends in our model do not translate one-to-one to the dividend payments recorded in CRSP. We therefore use an array of empirical measures to capture the performance of the risky asset in our model: (1) annual stock market returns, (2) real dividends, (3) real earnings, and (4) U.S. GDP. We obtain the first three series from Robert Shiller's website, and the nominal GDP data from the Federal Reserve Bank of St. Louis (for 1929-2016) and Historical Statistics of the United States Millennial Edition Online (for 1871-1928). We convert nominal GDP into real GDP using Shiller's consumer price index variable. 

Dividends in our model are best interpreted as the performance of the risky asset at medium frequencies. Therefore, we use the \citeN{christiano2003band} band-pass filter and remove stochastic cycles at frequencies lower than 2 years and higher than 8 years,\footnote{These are the default frequencies for the CF-filter. We also remove a linear trend of the series before applying the filter and, in addition, work with the natural logarithm of earnings and GDP to remove non-linearities in these series. In unreported analyses, we also use the natural logarithm of dividends and obtain very similar results.} 
for all non-stationary series (dividend, earnings, and GDP).

In order to construct the experienced returns, dividends, earnings, and GDP of different generations over the course of their lives, we apply the formula from equation (\ref{eq:EBL-w}). We calculate generation-specific weighted averages, employing both linearly declining weights ($\lambda=1$), and a steeper weighting function ($\lambda = 3$), corresponding to the range of empirical estimates in \citeN{Malmendier_Nagel2008}. 

%We construct two measures of differences in lifetime experiences across cohorts. Our first measure is the difference between the lifetime experiences of older generations (age 61 to 74) and those of younger generations (age 24 to 39). \textbf{[TRUE?]}For both age groups, we calculate the experienced performance as a weighted average across cohorts, where weights are determined by cohort size. 
%For this, we obtain data on U.S. population by age between 19xx and 20xx from SOURCE.\footnote{We include birth-year returns in the experience measure. That is, we construct the measure as if all individuals in the survey were born on January 1 and the SCF were always conducted on January 1. The results are robust to different assumptions, e.g., the assumption that either birthday or date of survey or both are assumed to be December 31. IS ALL OF THIS TRUE?} 
%Our second measure captures disagreement across cohorts in a given year, and is calculated as the standard deviation of the change in experienced performance between that year and the previous year. As above, we take into account cohort size when computing the standard deviation. 

\medskip

\textbf{Stock market participation.} We test the first prediction relating the differences in lifetime experiences between older and younger cohorts (i.\,e., those above 60 and those below 40 years of age)
% % (i.e., ages 61 to 74) 
% and younger cohorts (i.\,e., those below 40 years)
% % (i.e., ages 24 to 39) 
to the differences in their stock-market investment. Our source of household-level micro data is 
the cross-sectional data on asset holdings and various household background characteristics in the 
Survey of Consumer Finances (SCF). 
We use all waves of the modern, triannual SCF, available from the Board of Governors of the Federal Reserve System since 1983.
We follow the variable construction of \citeN{Malmendier_Nagel2008} and extend their analysis to the most recently released data. In addition, we employ some waves of the precursor survey, available from the Inter-university Consortium for Political and Social Research at the University of Michigan
%. The precursor survey starts in 
since 1947.
%, but includes age (rather than 5- or 10-year brackets)
%, which essential for our experience measure, 
%only since 1960. 
We use all survey waves that 
%also offer 
include age and stock-market participation.\footnote{Those are 1960, 1962, 1963, 1964, 1967, 1968, 1969, 1970, 1971, and 1977.} 

For the extensive margin of stock holdings, we construct an indicator of stock-market participation. It equals 1 when a household holds more than zero dollars worth of stocks. We define stock holdings as the sum of directly held stocks (including stock held through investment clubs) and the equity portion of mutual fund holdings, including stocks held in retirement accounts (e.g., IRA, Keogh, and 401(k) plans).\footnote{For 1983 and 1986, we need to impute the stock component of retirement assets from the type of the account or the institution at which they are held and allocation information from 1989. From 1989 to 2004, the SCF offers only coarse information on retirement assets (e.\,g., mostly stocks, mostly interest bearing, or split), and we follow a refined version of the Federal Reserve Board's conventions in assigning portfolio shares. See \citeN{Malmendier_Nagel2008} for more details.}

For the intensive margin of stock holdings, we calculate the fraction of liquid assets invested in stocks as the share of directly held stocks plus the equity share of mutual funds, using all surveys from 1960-2013 other than 1971. Liquid assets are defined as the sum of stock holdings, bonds, cash, and short-term instruments (checking and savings accounts, money market mutual funds, certificates of deposit). In these analyses of the intensive margin, we drop all households that have no money in stocks. 

For both the young and old age group, we calculate their experience and their stock-market investment as a weighted average across cohorts, with the weight variable provided in the SCF. The weighted estimates are representative of the U.S. population.\footnote{The 1983-2013 SCF waves oversample high-income households with significant stock holdings. The oversampling is helpful for our analysis of asset allocation, but could induce selection bias. By applying SCF sample weights, we undo the over-weighting of high-income households and also adjust for non-response bias.}

We present the results graphically. We plot the relation between stock holdings (extensive and intensive margin) and experienced returns (Figure \ref{fig:performance_returns}), dividends (Figure \ref{fig:performance_dividends}), earnings (Figure \ref{fig:performance_earnings}), and GDP (Figure \ref{fig:performance_GDP}). Graphs \ref{fig:performance_returns}.(a) and \ref{fig:performance_returns}.(c) 
%of Figure \ref{fig:performance_returns} 
update the evidence on the extensive margin and %experienced 
returns presented in \citeN{Malmendier_Nagel2008}.)

%\FloatBarrier
\begin{figure*}
	\centering     %%% not \center
	\subfigure[Stock-market participation  ($\lambda = 1$) ]{\label{fig:Stock.a}\includegraphics[width=0.48\textwidth]{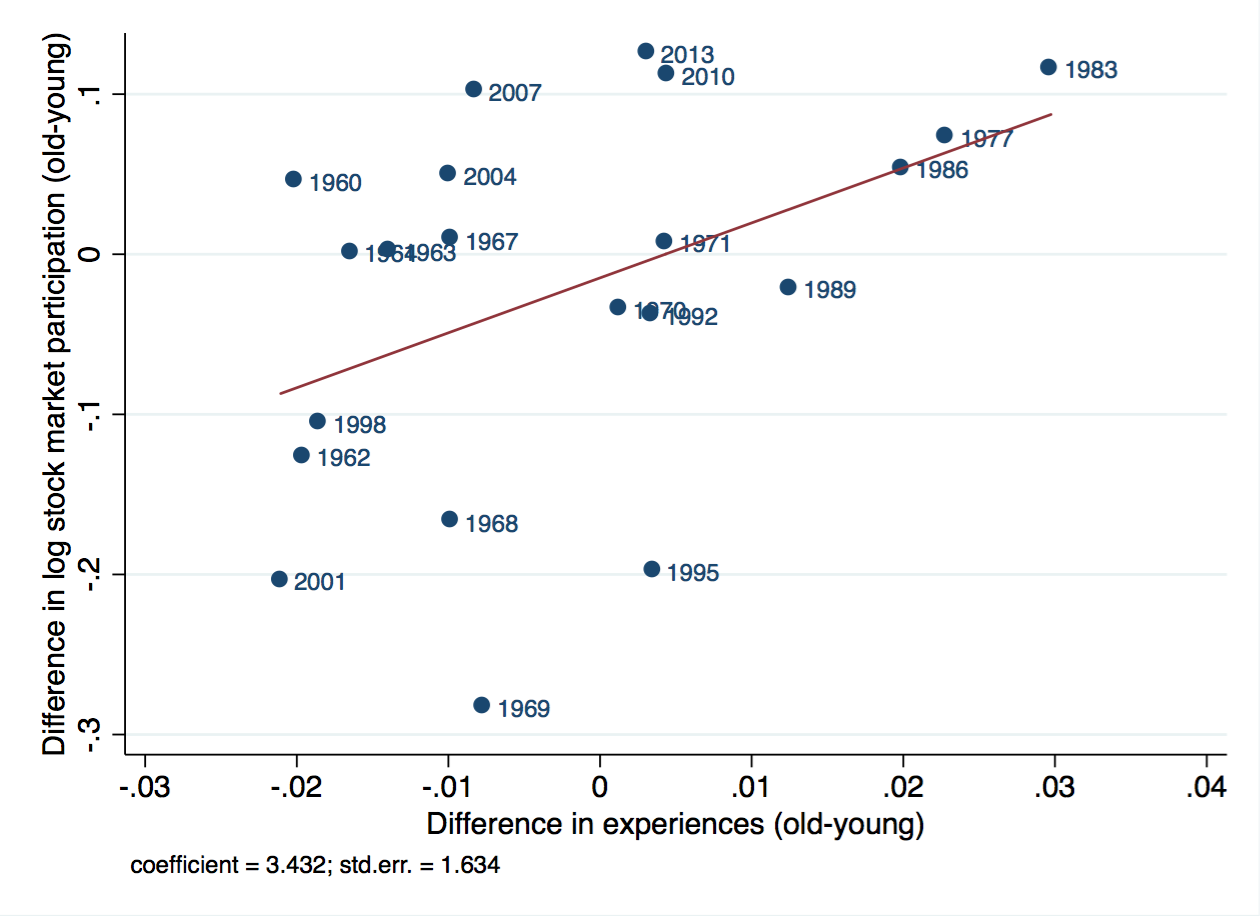}} 
	\subfigure[Fraction invested in stock ($\lambda = 1$)]{\label{fig:Stock.b}\includegraphics[width=0.48\textwidth]{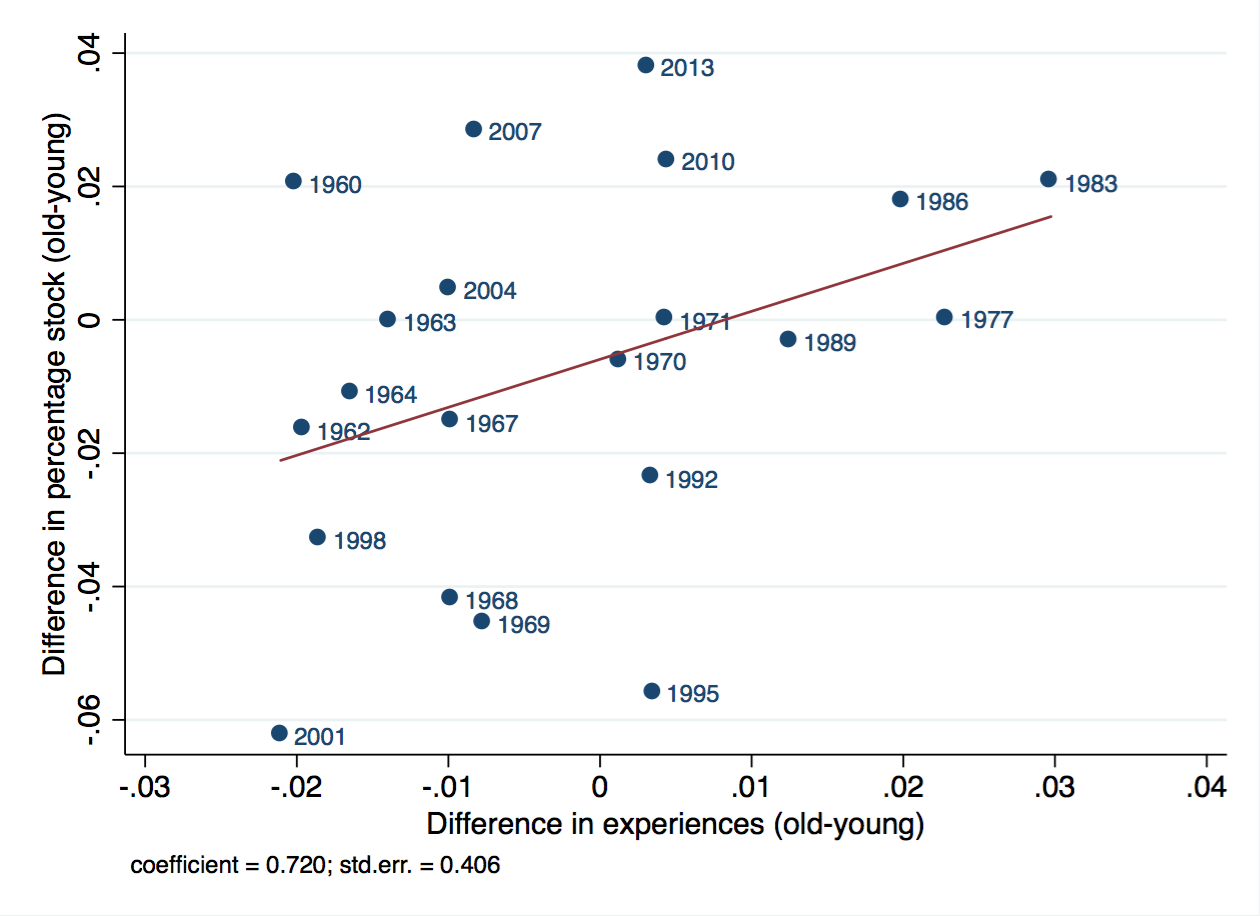}} %
	\subfigure[Stock-market participation  ($\lambda = 3$)]{\label{fig:Stock.c}\includegraphics[width=0.48\textwidth]{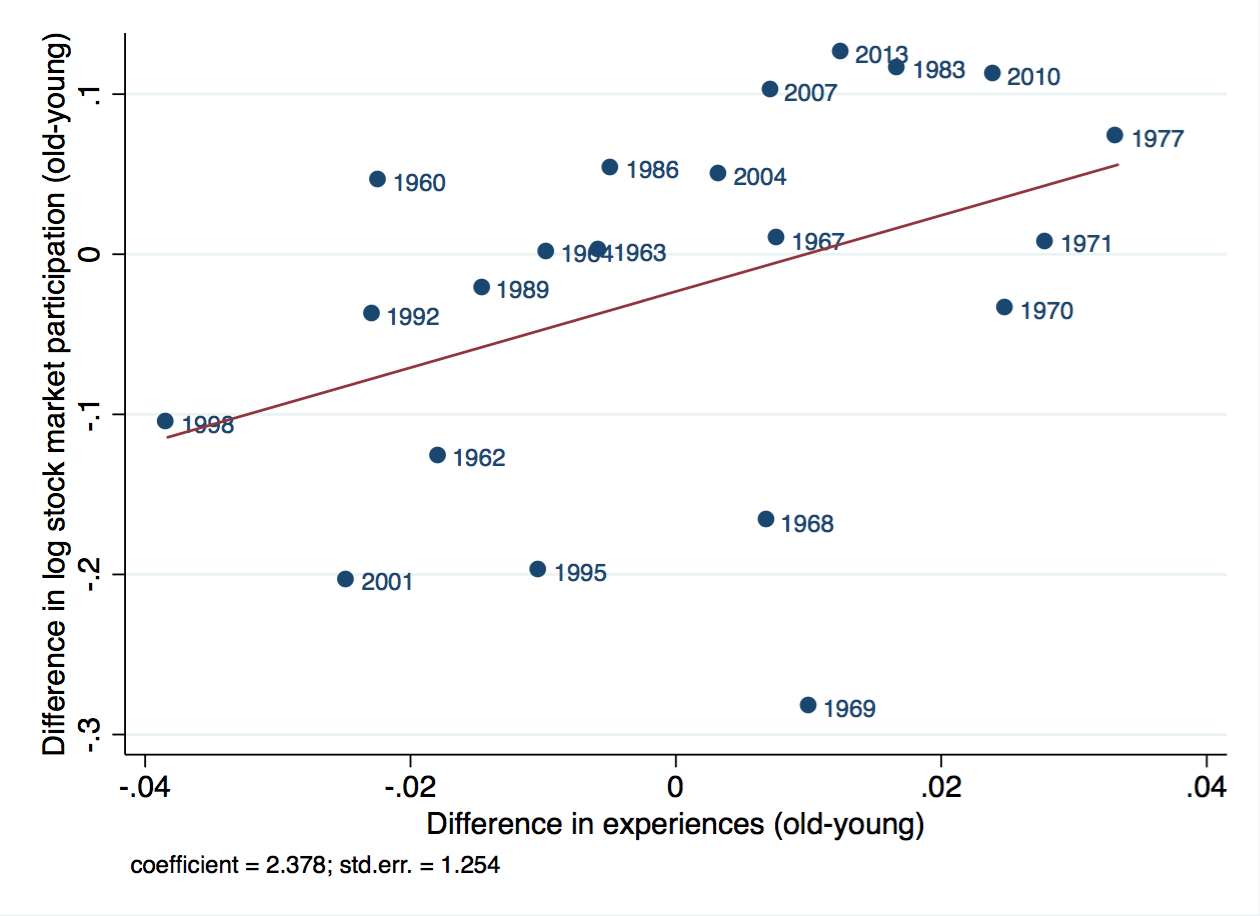}}
	\subfigure[Fraction invested in stock  ($\lambda = 3$)]{\label{fig:Stock.d}\includegraphics[width=0.48\textwidth]{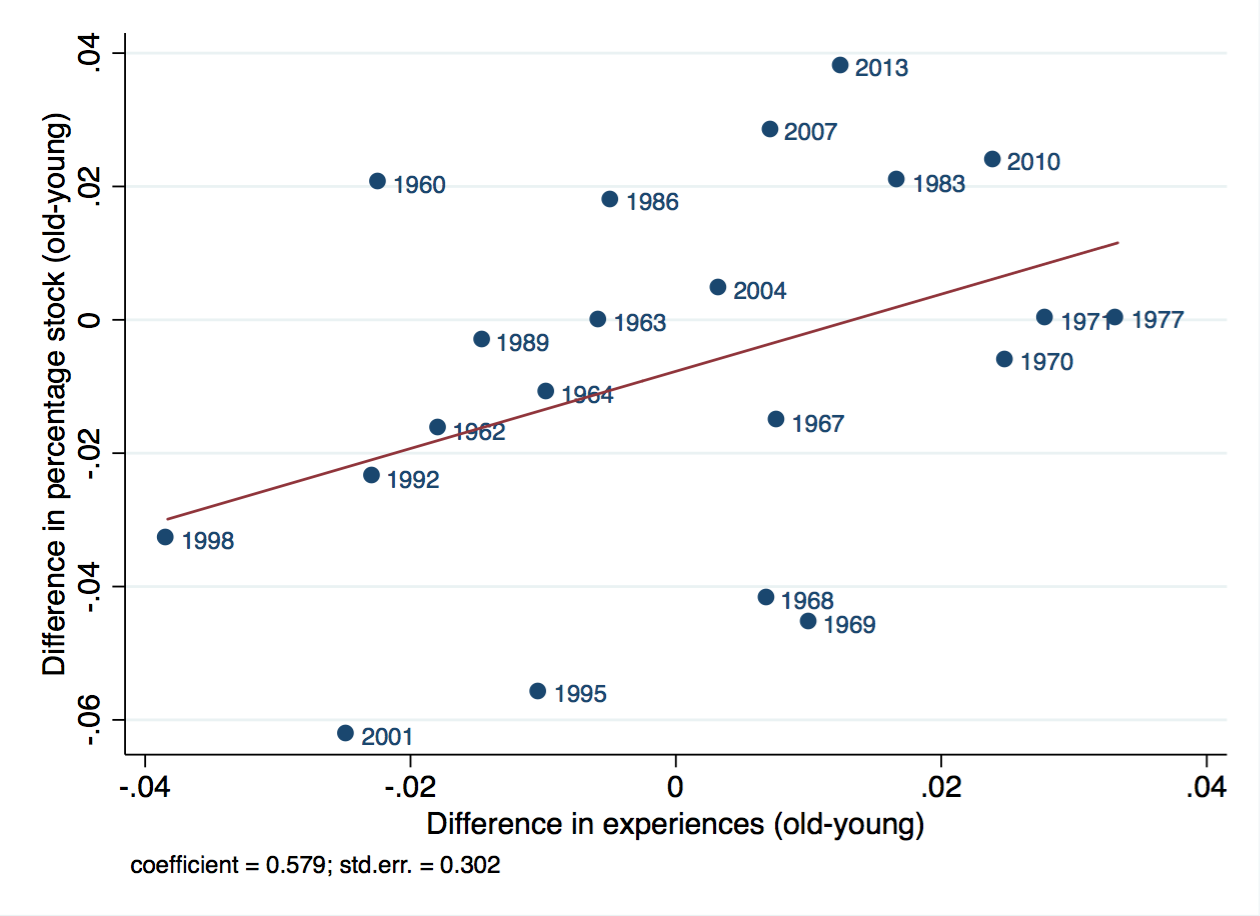}}
	\caption{Experienced Returns and Stock Holdings}
	\label{fig:performance_returns}
	\vspace{-0.2cm}
	\caption*{\textit{Notes}. \textit{Difference in experienced returns} is calculated as the lifetime average experienced returns of the S\&P500 Index as given on Robert Shiller's website, using declining weights with either $\lambda = 1$ or $\lambda = 3$ as in equation (\ref{eq:EBL-w}). \textit{Stock-market participation} is measured as the fraction of households in the respective age groups that hold at least \$1 of stock ownership, either as directly held stock or indirectly, e.\,g. via mutuals or retirement accounts. \textit{Fraction invested in stock} is the fraction of liquid assets stock-market participants invest in the stock market. We classify households whose head is above 60 years of age as ``old,'' and households whose head is below 40 years of age as ``young.''  Difference in stock holdings, the y-axis in graphs (a) and (c), is calculated as the difference between the logs of the fractions of stock holders among the old and among the young age group.  Percentage stock, the y-axis in graphs (b) and (d), is the difference in the fraction of liquid assets invested in stock.  The red line depicts the linear fit.}
\end{figure*}
%\FloatBarrier

%\FloatBarrier
\setcounter{subfigure}{0}
\begin{figure*}
	\centering     %%% not \center
	\subfigure[Stock-market participation  ($\lambda = 1$) ]{\label{fig:Stock.a}\includegraphics[width=0.48\textwidth]{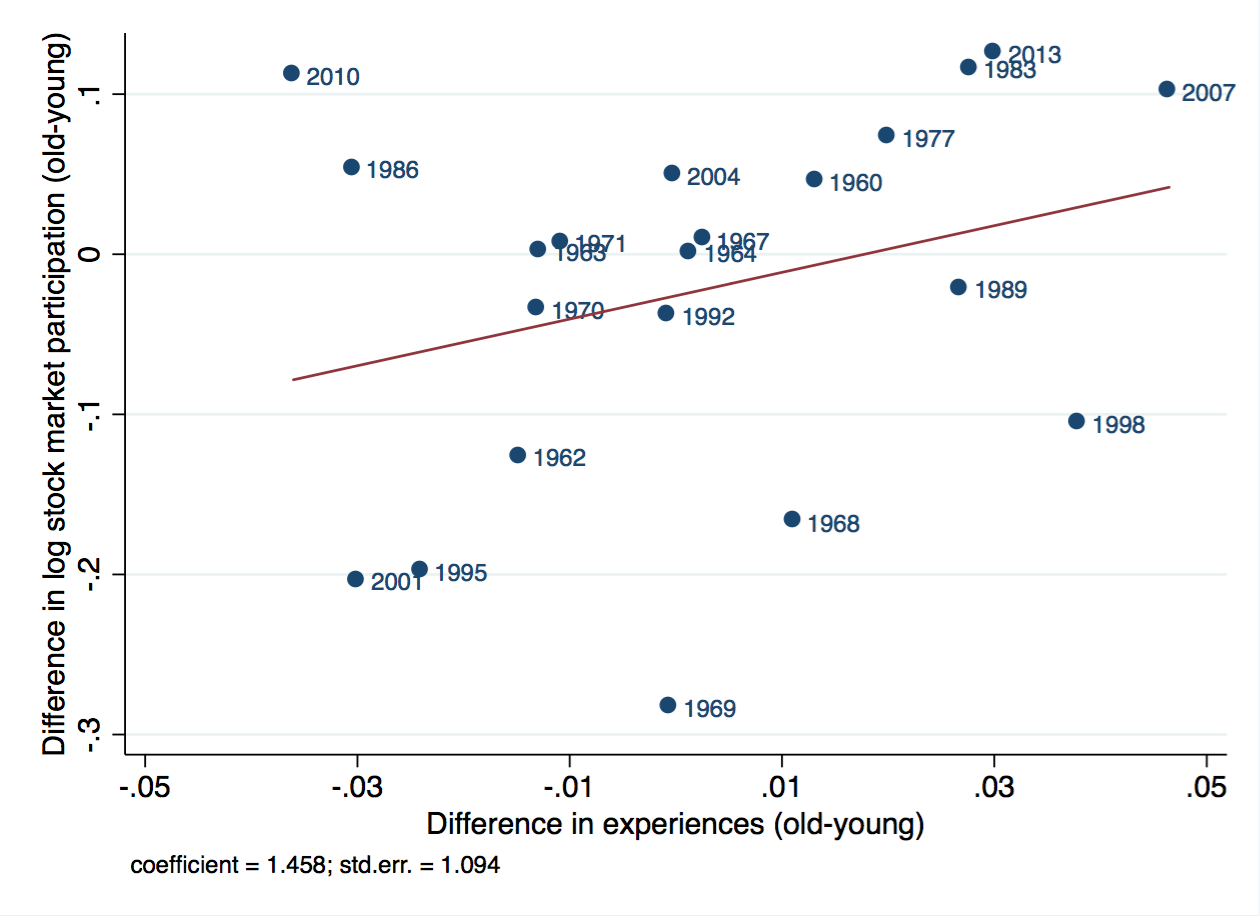}} 
	\subfigure[Fraction invested in stock ($\lambda = 1$)]{\label{fig:Stock.b}\includegraphics[width=0.48\textwidth]{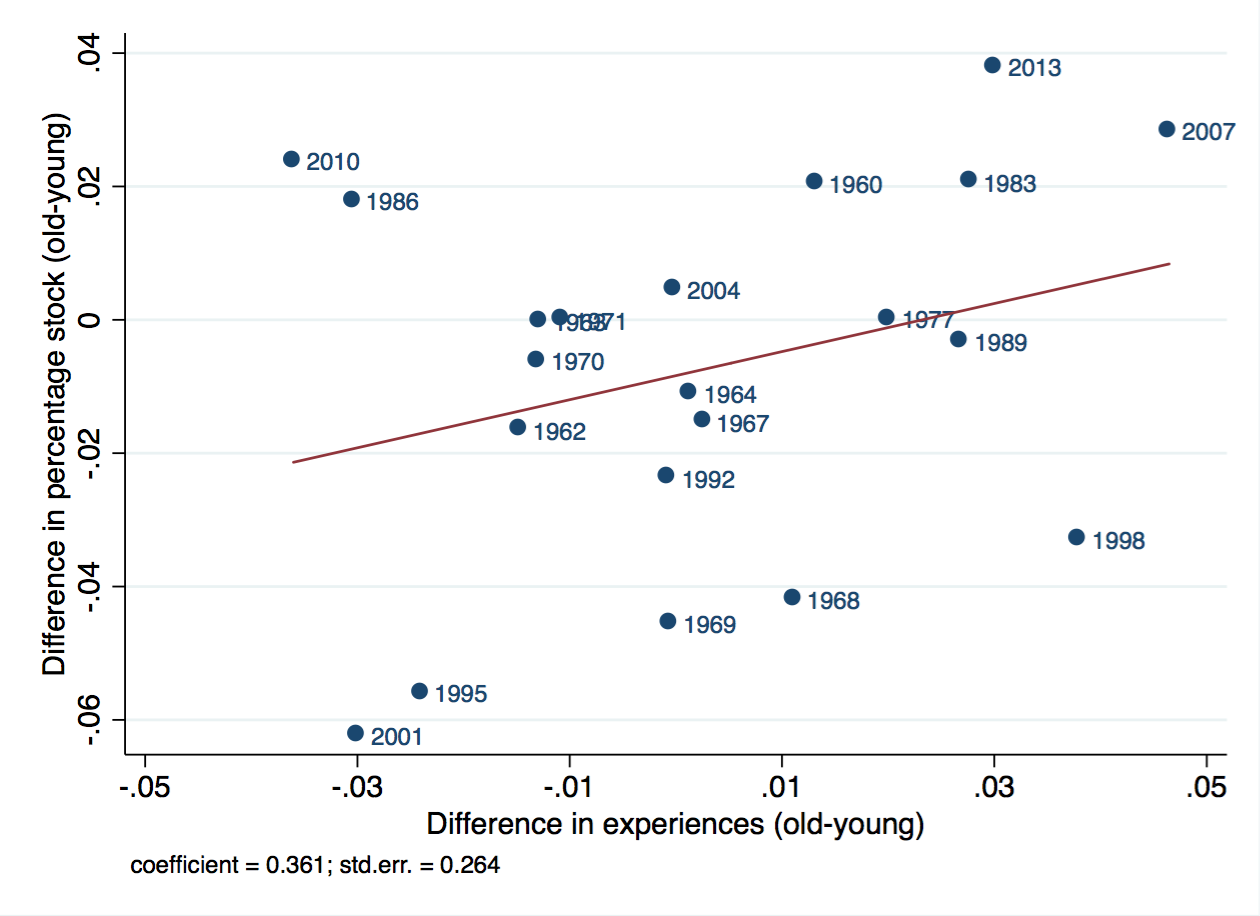}} %
	\subfigure[Stock-market participation  ($\lambda = 3$)]{\label{fig:Stock.c}\includegraphics[width=0.48\textwidth]{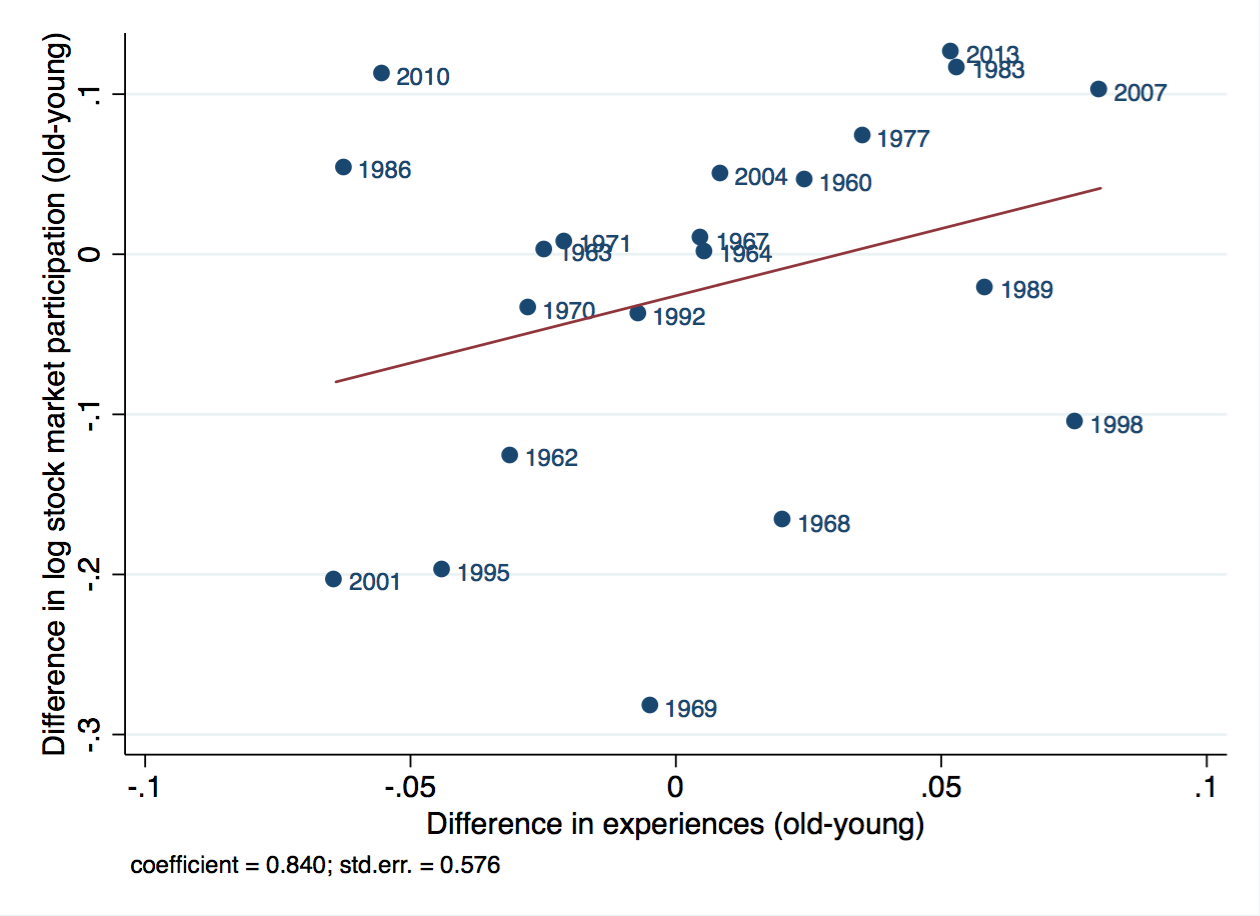}}
	\subfigure[Fraction invested in stock  ($\lambda = 3$)]{\label{fig:Stock.d}\includegraphics[width=0.48\textwidth]{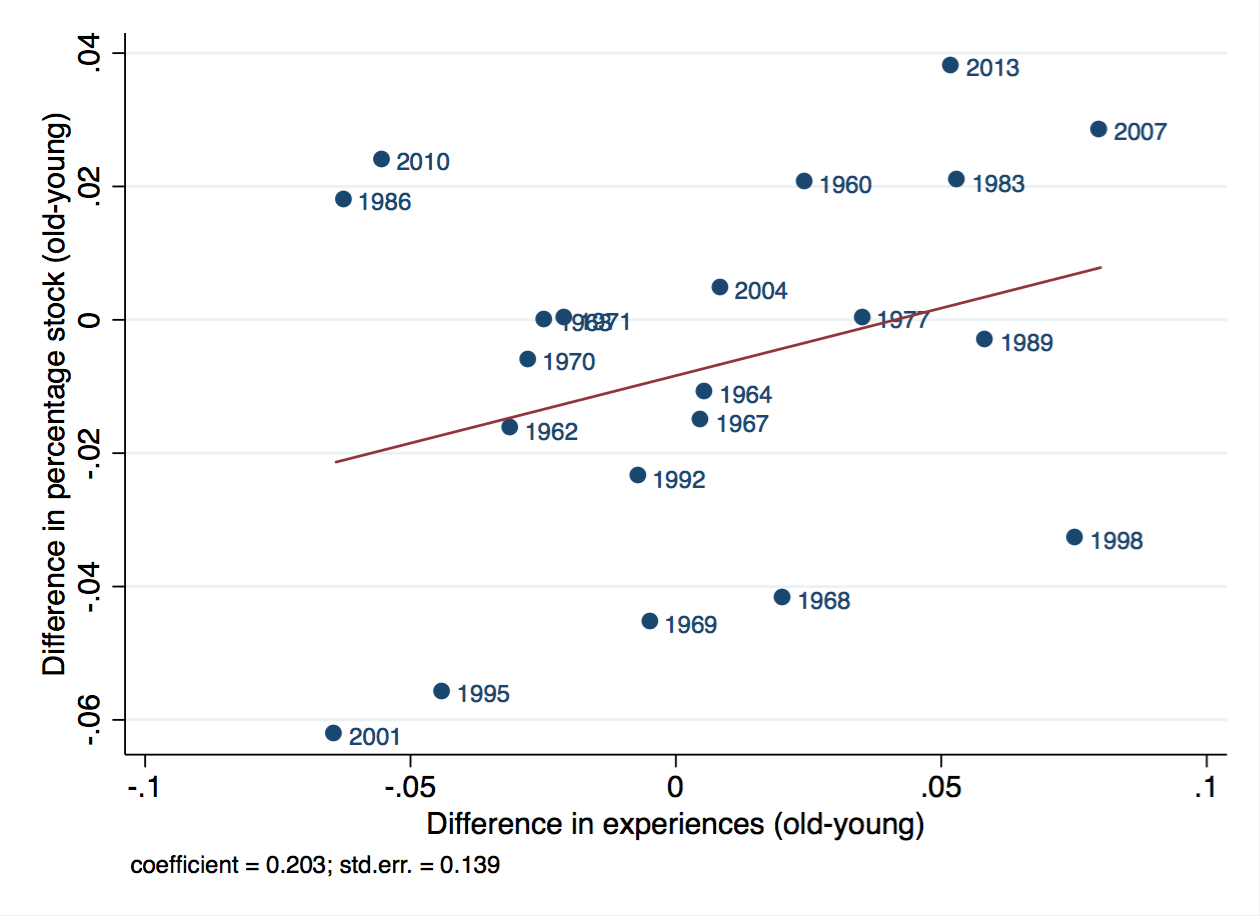}}
	\caption{Experienced Dividends and Stock Holdings}
	\label{fig:performance_dividends}
	\vspace{-0.2cm}
	\caption*{\textit{Notes}. \textit{Difference in experienced dividends} is calculated as the lifetime average experienced real dividends as given on Robert Shiller's website, using declining weights with either $\lambda = 1$ or $\lambda = 3$ as in equation (\ref{eq:EBL-w}). \textit{Stock-market participation} is measured as the fraction of households in the respective age groups that hold at least \$1 of stock ownership, either as directly held stock or indirectly, e.g. via mutuals or retirement accounts. \textit{Fraction invested in stock} is the fraction of liquid assets stock-market participants invest in the stock market. We classify households whose head is above 60 years of age as ``old,'' and households whose head is below 40 years of age as ``young.''  Difference in stock holdings, the y-axis in graphs (a) and (c), is calculated as the difference between the logs of the fractions of stock holders among the old and among the young age group.  Percentage stock, the y-axis in graphs (b) and (d), is the difference in the fraction of liquid assets invested in stock.  The red line depicts the linear fit.}
\end{figure*}
%\FloatBarrier

%\FloatBarrier
\setcounter{subfigure}{0}
\begin{figure*}
	\centering     %%% not \center
	\subfigure[Stock-market participation  ($\lambda = 1$) ]{\label{fig:Stock.a}\includegraphics[width=0.48\textwidth]{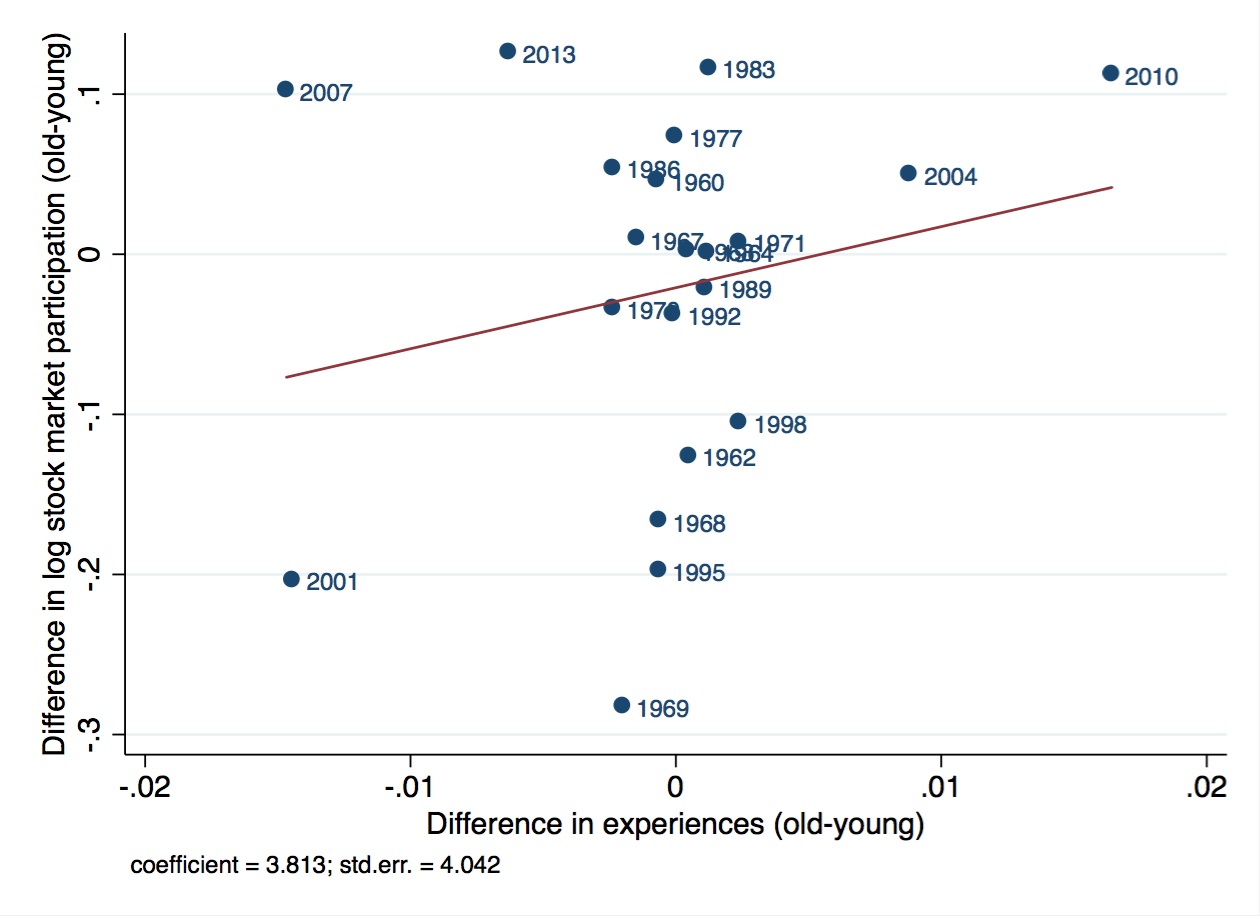}} 
	\subfigure[Fraction invested in stock ($\lambda = 1$)]{\label{fig:Stock.b}\includegraphics[width=0.48\textwidth]{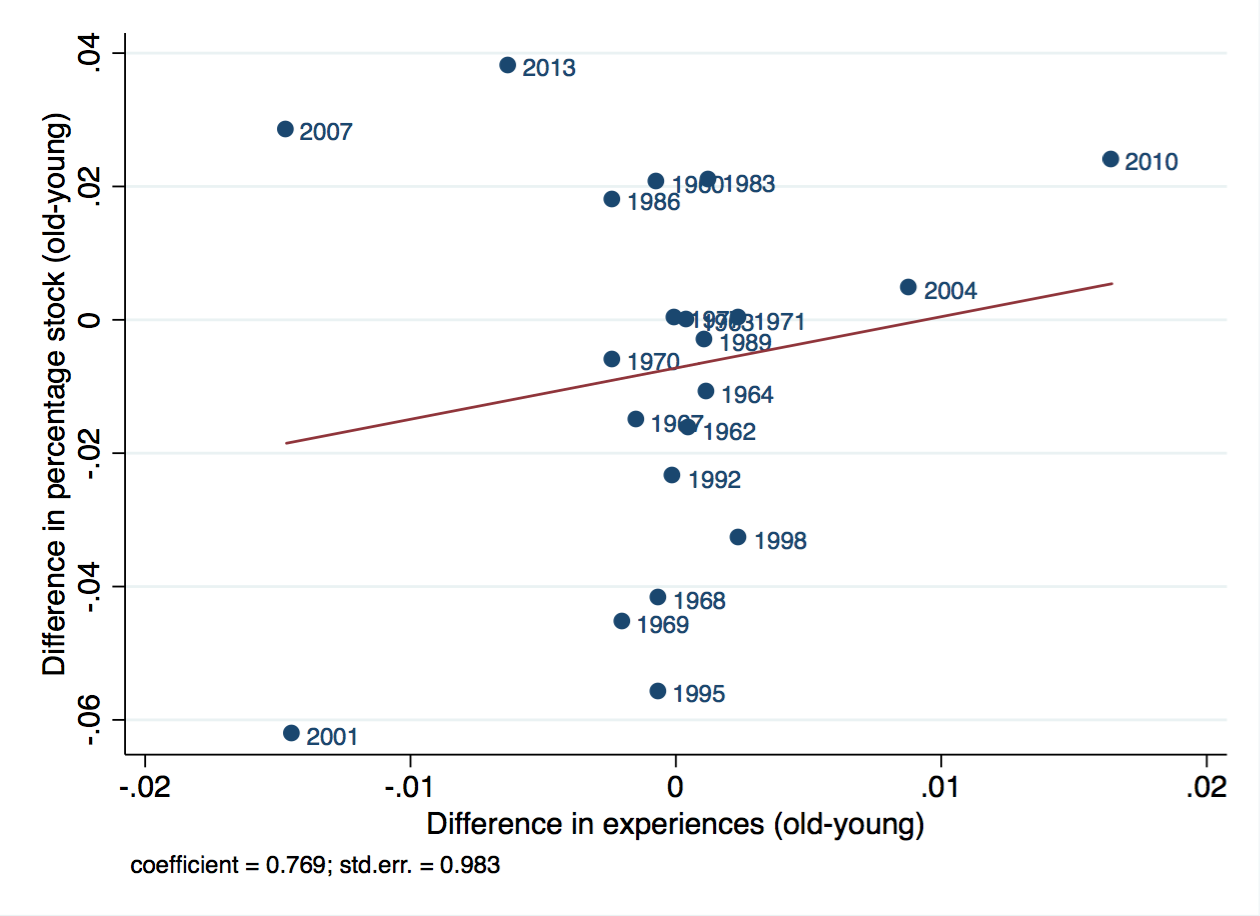}} %
	\subfigure[Stock-market participation  ($\lambda = 3$)]{\label{fig:Stock.c}\includegraphics[width=0.48\textwidth]{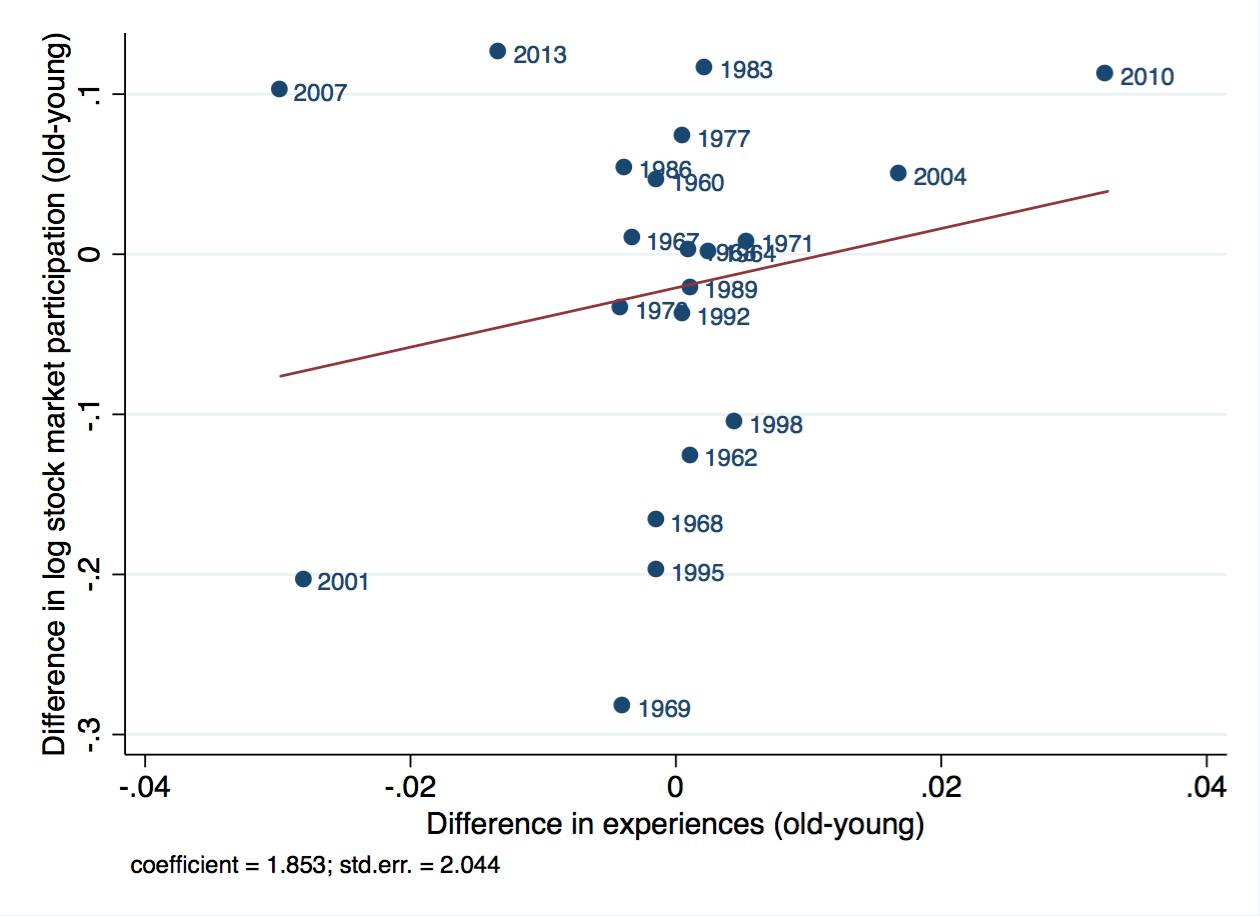}}
	\subfigure[Fraction invested in stock  ($\lambda = 3$)]{\label{fig:Stock.d}\includegraphics[width=0.48\textwidth]{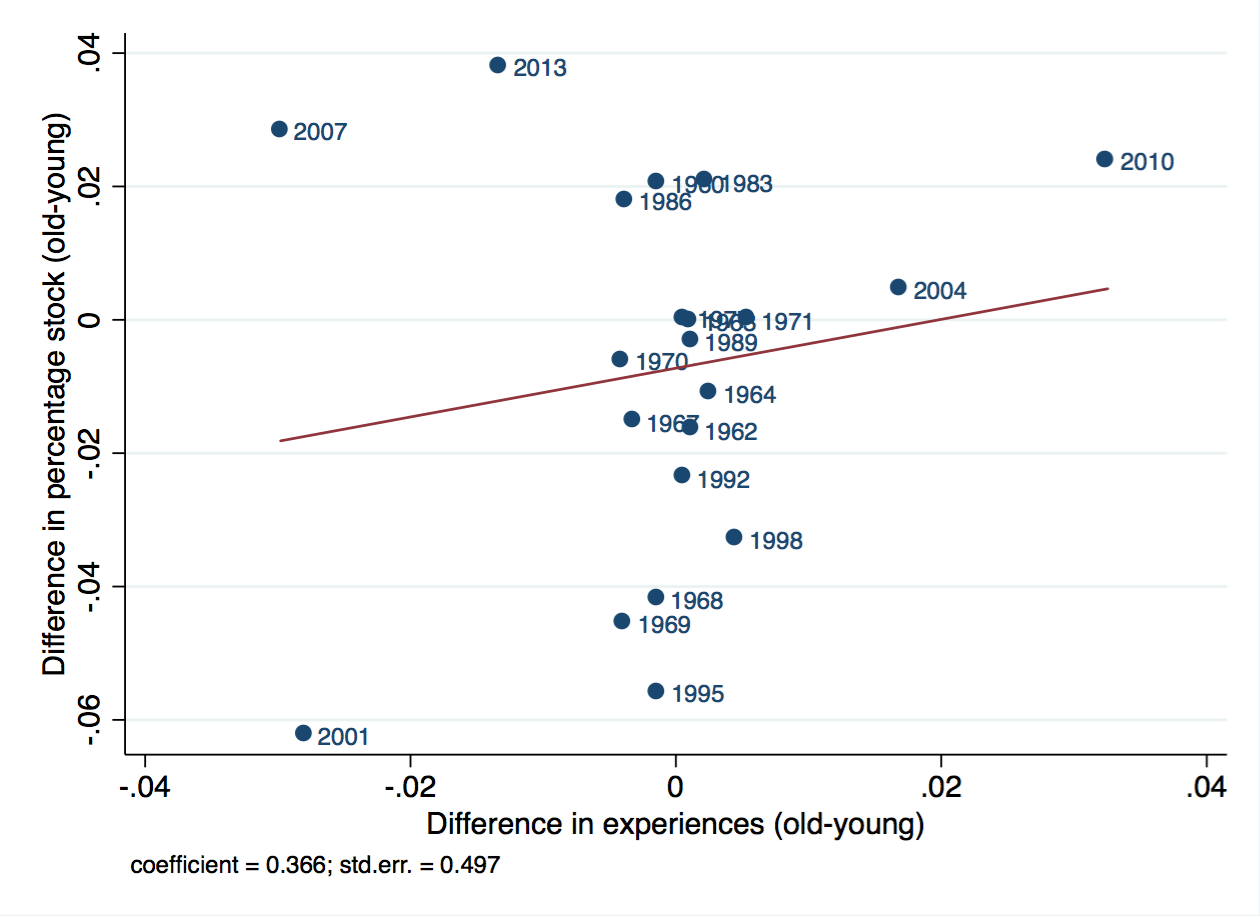}}
	\caption{Experienced Earnings and Stock Holdings}
	\label{fig:performance_earnings}
	\vspace{-0.2cm}
	\caption*{\textit{Notes}. \textit{Difference in experienced earnings} is calculated as the lifetime average experienced log real earnings as given on Robert Shiller's website, using declining weights with either $\lambda = 1$ or $\lambda = 3$ as in equation (\ref{eq:EBL-w}). \textit{Stock-market participation} is measured as the fraction of households in the respective age groups that hold at least \$1 of stock ownership, either as directly held stock or indirectly, e.g. via mutuals or retirement accounts. \textit{Fraction invested in stock} is the fraction of liquid assets stock-market participants invest in the stock market. We classify households whose head is above 60 years of age as ``old,'' and households whose head is below 40 years of age as ``young.''  Difference in stock holdings, the y-axis in graphs (a) and (c), is calculated as the difference between the logs of the fractions of stock holders among the old and among the young age group.  Percentage stock, the y-axis in graphs (b) and (d), is the difference in the fraction of liquid assets invested in stock.  The red line depicts the linear fit.}
\end{figure*}
%\FloatBarrier

%\FloatBarrier
\setcounter{subfigure}{0}
\begin{figure*}
	\centering     %%% not \center
	\subfigure[Stock-market participation  ($\lambda = 1$) ]{\label{fig:Stock.a}\includegraphics[width=0.48\textwidth]{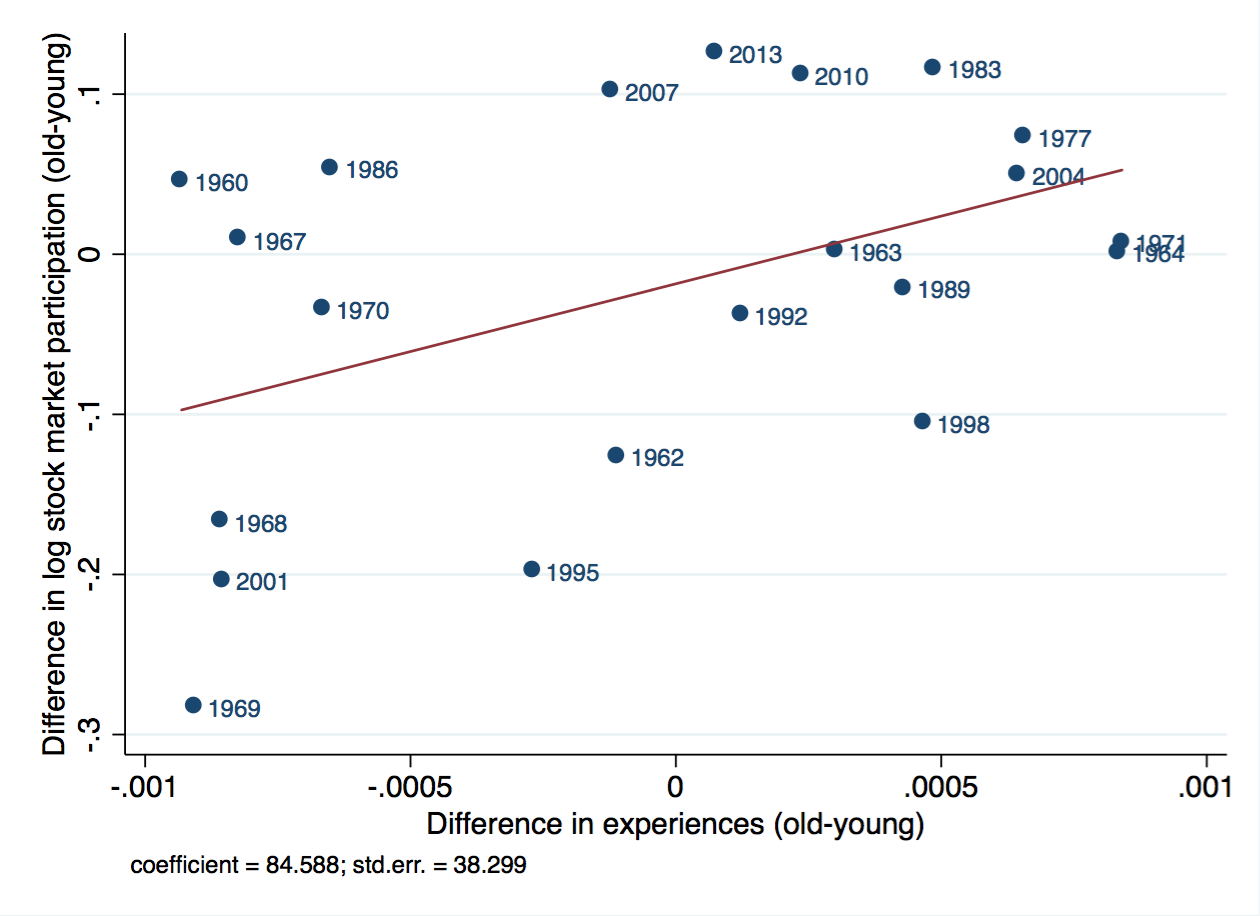}} 
	\subfigure[Fraction invested in stock ($\lambda = 1$)]{\label{fig:Stock.b}\includegraphics[width=0.48\textwidth]{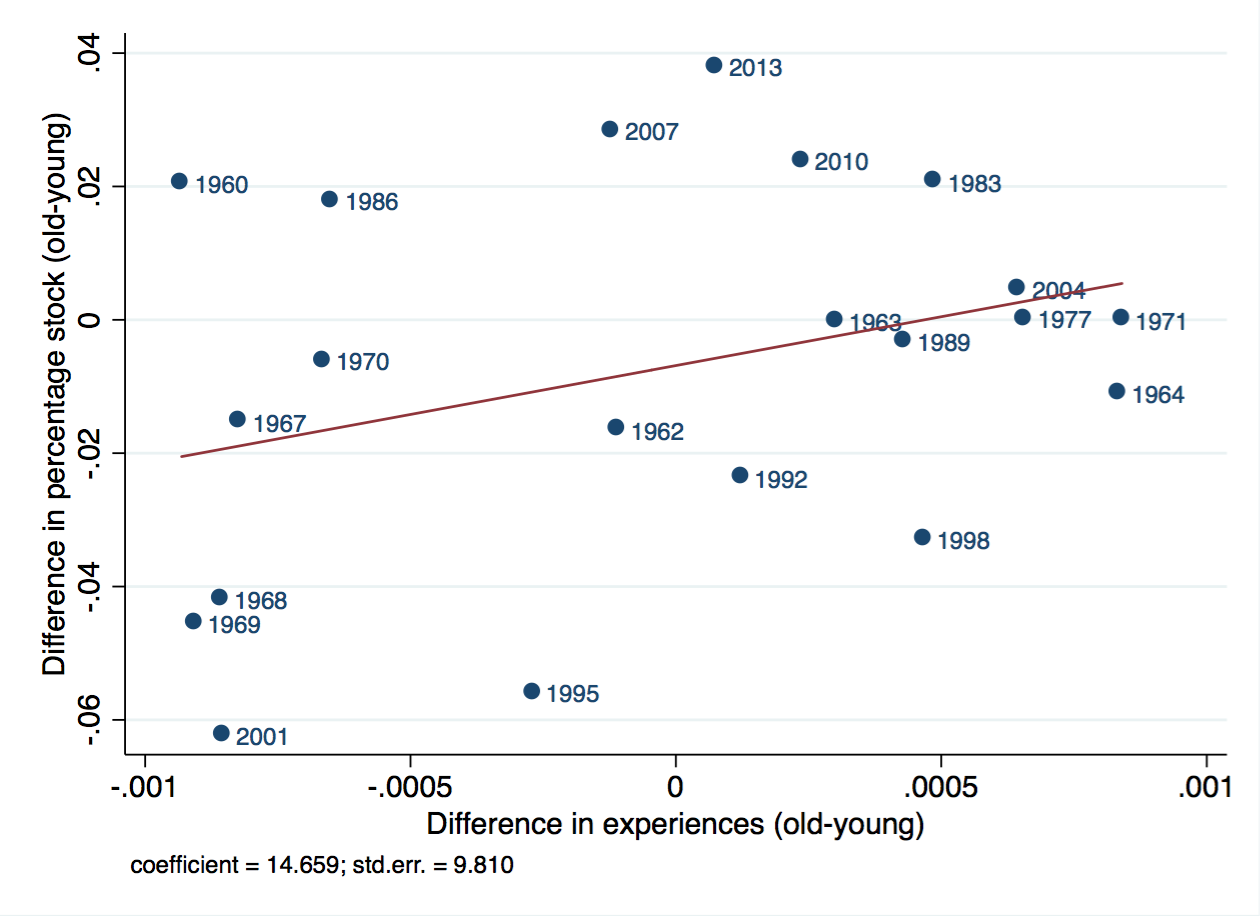}} %
	\subfigure[Stock-market participation  ($\lambda = 3$)]{\label{fig:Stock.c}\includegraphics[width=0.48\textwidth]{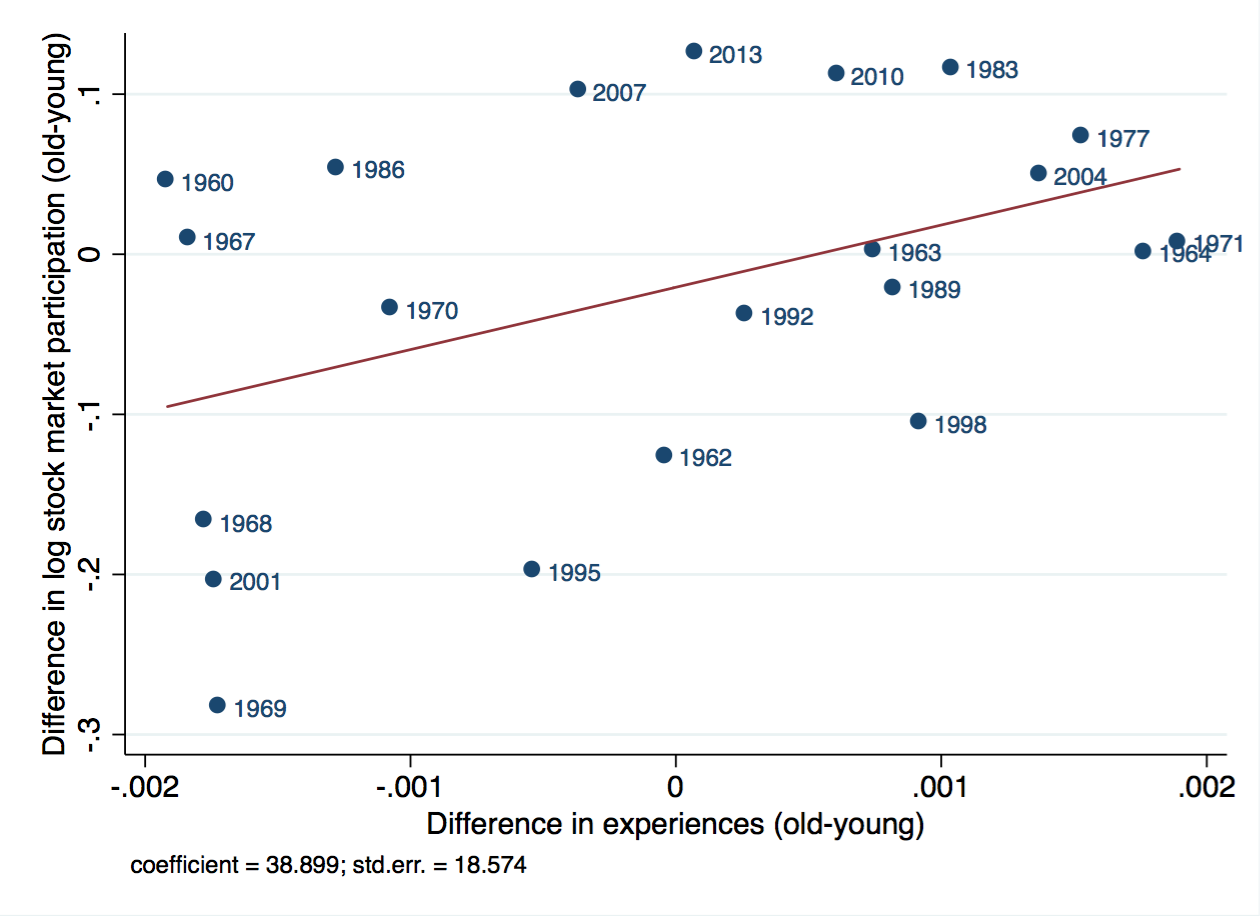}}
	\subfigure[Fraction invested in stock  ($\lambda = 3$)]{\label{fig:Stock.d}\includegraphics[width=0.48\textwidth]{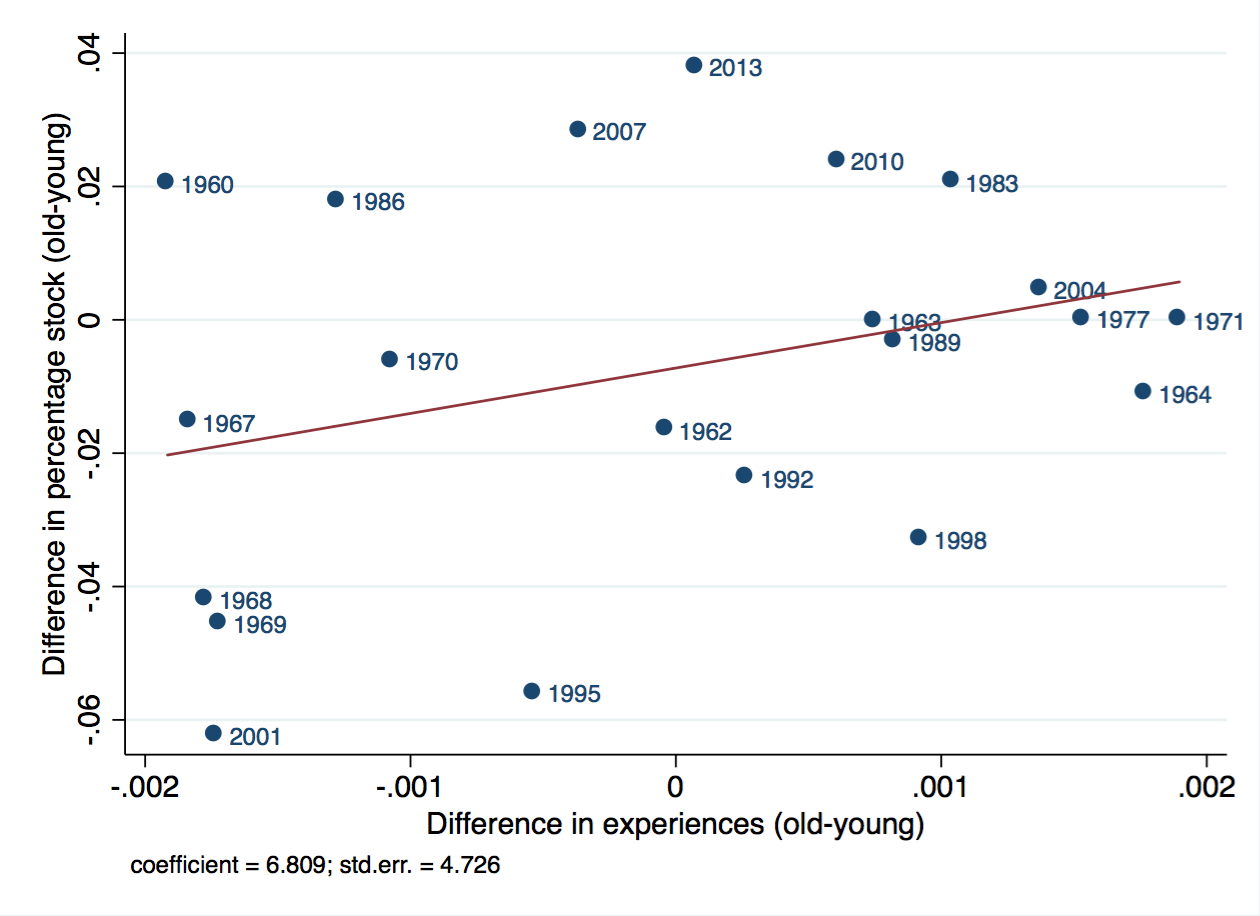}}
	\caption{Experienced Log GDP and Stock Holdings}
	\label{fig:performance_GDP}
	\vspace{-0.2cm}
	\caption*{\textit{Notes}. \textit{Difference in experienced GDP} is calculated as the lifetime average experienced log real GDP, using declining weights with either $\lambda = 1$ or $\lambda = 3$ as in equation (\ref{eq:EBL-w}). \textit{Stock-market participation} is measured as the fraction of households in the respective age groups that hold at least \$1 of stock ownership, either as directly held stock or indirectly, e.g. via mutuals or retirement accounts. \textit{Fraction invested in stock} is the fraction of liquid assets stock-market participants invest in the stock market. We classify households whose head is above 60 years of age as ``old,'' and households whose head is below 40 years of age as ``young.''  Difference in stock holdings, the y-axis in graphs (a) and (c), is calculated as the difference between the logs of the fractions of stock holders among the old and among the young age group.  Percentage stock, the y-axis in graphs (b) and (d), is the difference in the fraction of liquid assets invested in stock.  The red line depicts the linear fit.}
\end{figure*}
%\FloatBarrier

The results for all four performance measures and both for the extensive and intensive margin are in line with the predictions of our model. Starting from experienced returns with $\lambda=1$ in panel (a) of Figure 	\ref{fig:performance_returns}, we see that the older age-group is more likely to hold stock, compared to the younger age-group, when they have experienced higher stock-market returns in their lives. The opposite holds when the returns experienced by the younger generations are higher than those of the older generations. The slope coefficient of the linear line of fit is significant at 5\%.  The steepness of the weighting function, and hence the extent of imposed weight on recent data points, makes little difference, as the comparison with graph (b) for $\lambda=3$ reveals.

The analysis of the intensive margin of stock-market investment yields the same conclusion. Both graph (c) and graph (d) indicate that older generations invest a higher share of the their liquid assets in stock, compared to the younger generations, when their experienced returns have been higher than those of the younger age-group over their respective life-spans so far; and vice versa when they have experienced lower returns than the younger cohorts. Here, the slope coefficient is significant at 10\%. 

Figures \ref{fig:performance_dividends} to \ref{fig:performance_GDP} present the corresponding results for experienced dividends, earnings, and GDP. For all measures, we observe a positive relation of differences in experienced performance and stock investments between the young and the old. 
The fact that we obtain very similar findings for a wide array of performance measures 
lends support to the link between our theoretical model and the empirical facts, and ameliorates concerns about dividends %in our model 
not translating one-to-one into an empirical performance measure.

\medskip
\textbf{Trade volume.} We now turn to the second prediction, which relates trade volume to the dispersion of changes in disagreement among investors. We calculate changes in the level of disagreement as the cross-cohort standard deviation of the change in experienced performance between the current year and the previous year. We weight the cohorts by their sizes when computing the standard deviation.\footnote{For this, we obtain data on U.S. population by age between 1985 and 2015 from US Census Bureau.}
% First del x / del d implies ... explain that it translates into liquid asset share 
% NOTE VICO about using response to differences in experience rather than dividend change: The trade volume definition, if you look at it, says that trade volume is basically a sum of the difference of each person's beliefs and the average belief ... so I think we can say disagreements comfortably. On top of that, an increase in d_t does cause disagreements in our framework, and that is the only channel that will drive the trade volume, so yes, we can say disagreements.

%\FloatBarrier
%\setcounter{figure}{5}
\setcounter{subfigure}{0}
\begin{figure}
	\centering     %%% not \center
	\subfigure[Linear weights ($\lambda = 1$) ]{\label{fig:Stock.a}\includegraphics[width=0.45\textwidth]{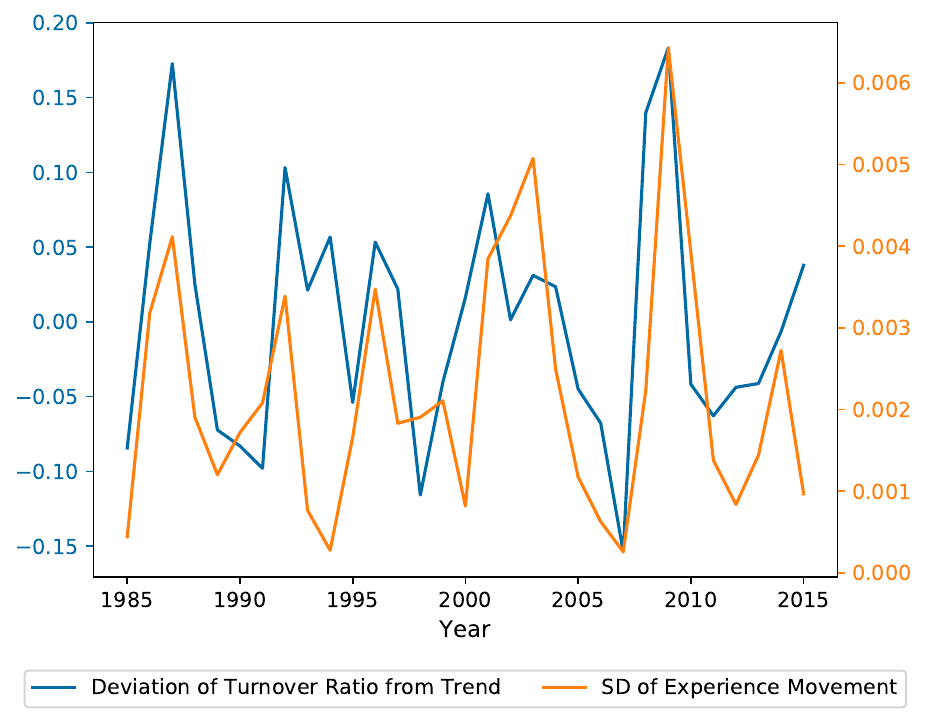}} 
	\subfigure[Superlinear weights ($\lambda = 3$)]{\label{fig:Stock.b}\includegraphics[width=0.45\textwidth]{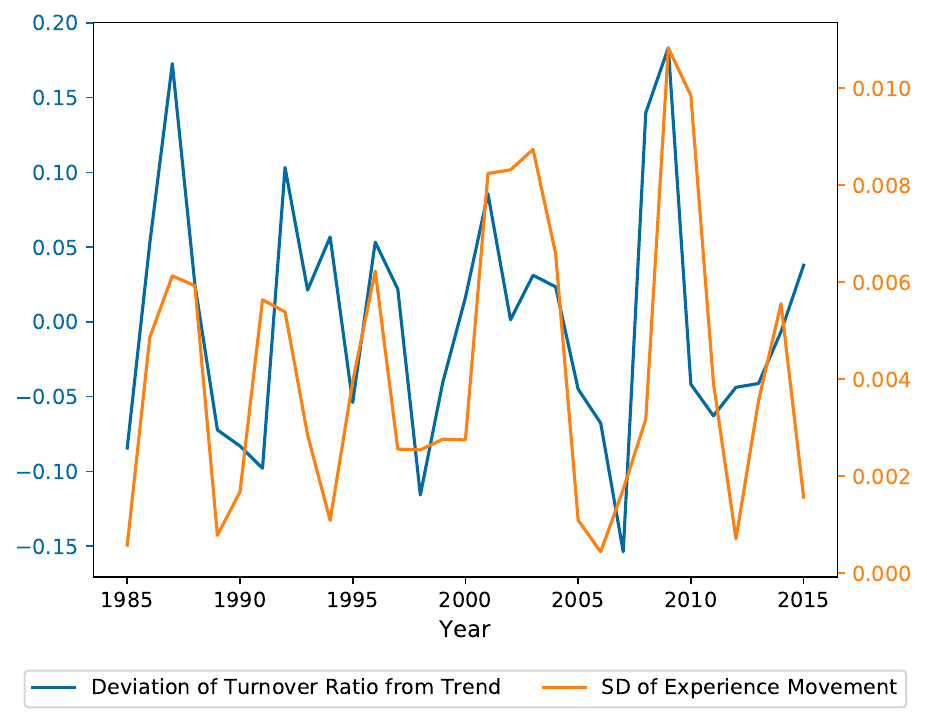}}	
	\caption{Trading Volume and Standard Deviation of Changes in Experienced Returns}
	\vspace{-0.2cm}
	\caption*{\textit{Notes}.
		%The solid line corresponds to the S.D. of experience movement and the dashed line corresponds to the deviation of turnover from trend.
		Trading volume, shown in (dark) blue, is calculated as the market-capitalization weighted average monthly turnover ratio (shares traded divided by shares outstanding) across all firms in January and in December of the preceding year. We log, linearly detrend, and CF-filter the yearly variable to obtain the deviation of turnover ratio from the trend. Returns are defined as inflation-adjusted change in price from the prior year divided by inflation-adjusted price in the prior year. Returns are linearly detrended and CF filtered. After creating the experience variables for returns, we take the change of the experience variable for individuals of a given age from the experience of % OLD VERSION: individuals of that age 
		those individuals in the prior year. We then calculate the current-year age-cohort population weighted standard deviation of this difference variable for each year as our measure of experience-based disagreement.
	}
	\label{F:turnover}
	\vspace{0.3cm}
\end{figure}
As a measure of abnormal trade volume, we calculate the deviation of the turnover ratio from its trend. Following prior literature (\citeN{statman2006investor}, \citeN{lo2000trading}), we first compute firm-level turnover ratio, i.\,e., the number of shares traded over the number of shares outstanding, on a monthly basis. We require that firms be listed on the NYSE or AMEX. We exclude NASDAQ-listed firms because the dealer market has volume measurement conventions that differ from exchange-traded securities (\citeN{atkins1997market}, \citeN{statman2006investor}). Then, we aggregate these numbers into a market-wide turnover ratio, weighting firms by their market capitalization.\footnote{This measure is equivalent to dollar turnover ratio, i.\,e., the ratio of the dollar value of all shares traded and the dollar value of the market.} Since the turnover ratio is non-stationary, we proceed in the same way as above and apply the \citeN{christiano2003band} to the logarithm of the turnover ratio series, 
so that we keep frequencies between 2 and 8 years.
We examine the co-movement between the aforementioned measure of disagreement, i.\,e., the standard deviation of the change in experienced stock returns, and the above measures of (abnormal) trade volume. 

Figure \ref{F:turnover} displays the trade volume in dark (blue) color, and changes in the experience-based disagreement about returns between cohorts in light (orange) color over time. Graph (a) shows the results when we apply linear weights for the calculation of experienced returns, and graph (b) displays the case with super-linear weights ($\lambda = 3$). Since we work with annual data for our disagreement variable, we choose the average of the turnover ratio in December of a given year and in January of the following year as our measure for trading volume of the given year. That is, Figure \ref{F:turnover} compares the variation (standard deviation) in changes in experienced returns in a given year to trading volume in December of that year and January of the following year. We choose 1985 as the starting year for this analysis, since individual investors were trading substantially less frequently when trading cost were significantly higher up to the mid-1980s, making it less likely that (individual) investors trade repeatedly based on experienced performance.%\footnote{Note that the presence of high trading cost before 1985 does not have any implications for stock holdings per se, consistent with the evidence presented in Figure \ref{F:turnover}. It does, however, have implications for the frequency of portfolio re-balancing, and hence, trading volume.}
\begin{table}[t]
	\centering
	\caption{Trading Volume and Changes in Experience-Based Disagreement}

	\begin{tabular}{llcccc} \hline
		\multicolumn{2}{l}{Experiences constructed using:} & Returns & Dividends & Log Earnings & Log GDP \\ \hline
		$\lambda = 1$ & Correlation & 0.5976 & 0.1788 & 0.3225 & 0.1780 \\
		& ($p$-value) & (0.0004) & (0.3358) & (0.0768) & (0.3379) \\
		$\lambda = 3$ & Correlation & 0.4904 & 0.1489 & 0.3099 & 0.1886 \\
		& ($p$-value) & (0.0051) & (0.4240) & (0.0898) & (0.3096) \\ \hline
	\end{tabular}
	%\justify
	\vspace{0.1cm}
	\caption*{\textit{Notes}.
		The table displays the pairwise correlations (and corresponding $p$-values in parentheses) of trading volume and eight measures of the change in experience-based disagreement. Trading volume is calculated using the market-capitalization weighted average turnover ratio (shares traded divided by shares outstanding) across all firms for January of the current year and December of the preceding year (averaged). We log, linearly detrend, and CF-filter the yearly variable.
		Experience-based disagreement is calculated separately for returns, dividends, earnings, and GDP, where returns are defined as the inflation-adjusted change in price from the prior year divided by inflation-adjusted price in the prior year, and dividends, earnings, and GDP are inflation adjusted. Returns, dividends, log earnings, and log GDP are linearly detrended and CF-filtered, and experience is calculated both with linear weights ($\lambda$=1) and with superlinear weights ($\lambda=3$).
		For each meaasure, we calculate the change in experience for individuals of a given age from the experience of the same individuals in the prior year. We then calculate the current-year age-cohort population-weighted standard deviation of the changes in experiences. 
	}
	\label{T:corr}
\end{table}
Consistent with the predictions of our model, we observe a clear co-movement between disagreement among cohorts and trading volume. Table \ref{T:corr} reveals that the co-movement is statistically significant at 1\%. 
The table presents the correlation between trading volume and our measures of changes in return disagreement, as well as the correlations when disagreement is measured using our alternative performance measures, i.\,e., using again dividends, earnings, or GDP. In each case, the correlation coefficient is again positive, albeit (marginally) significant only for changes in disagreement in experienced earnings.
%; it is however insignificant for the other two measures. 

The relationship between trade volume and changes in disagreement in experience-based beliefs about future returns in Table \ref{T:corr} and Figure \ref{F:turnover}, as well as the directionally similar correlations with the disagreement about other proxies for returns, corroborate the empirical relevance of our model for a better understanding of investor behavior. The pattern is consistent with experience-based learning and suggests that our novel explanation is worth considering. Moreover, as long as we assume that people trade based on their beliefs, it is unlikely that our channel is spurious. At the same time, %the pattern in Figure \ref{F:turnover} does not rule out that other variables affect 
other variables might also affect both the change in beliefs and the fluctuations in the trade volume. For example, if fluctuation in trade volume is caused both by variability in change in the beliefs (our model), and by another business-cycle macro variable, and both factors are positively correlated, we might still obtain a graph similar to Figure 9.
The claim in this section is not that there are no such factors, nor even that we can attribute most or all of the correlation depicted in Figure 9 to belief-based learning. Instead, the conclusion is that
all empirical findings in this section are consistent with experience-based learning
%, and that alternative explanations cannot easily explain these findings jointly.
and suggest experience-based learning as a novel and relevant factor that helps explain these empirical regularities jointly.

For a more detailed and careful empirical analysis it will be useful to analyze long-term individual-level panel data, which allows to link cumulative experiences and new experiences to trading decisions in the corresponding year.

\section{Conclusion} \label{sec: Conclusions}

\indent

In this paper, we have proposed %a stylized 
an
OLG equilibrium framework to study the effect of personal experiences on market dynamics. We incorporate the two main empirical features of experience effects,  the over-weighing of lifetime experiences and recency bias, into the belief formation process of agents. We show that experience-based learning not only generates several well-known asset pricing puzzles, that have been observed in the data, but it also produces new testable predictions about the relation between demographics, prices trading behavior, and the cross-section of asset holdings, which are in line with the data. We highlight two channels through which shocks have long-lasting effects on economic outcomes. The first is the belief formation process: all agents update their beliefs about the future after experiencing a given shock.  The second is the cross-sectional heterogeneity in the population: different experiences generate belief heterogeneity. We illustrate how the demographic composition of an economy can have important implications for the extent to which prices depend on past dividends. We consider this paper to be a first step into the exploration of the role of demographics in understanding market dynamics.

%We find that experienced-based learning generates price volatility and auto-correlation, and return predictability. These features go above and beyond the stochastic structure of the economy. Furthermore, changes in the level of disagreement between cohorts lead to higher trade volume, since shocks give rise to disagreements, which in turn generate gains from trade. As for the cross-section of asset-holdings, we find that younger cohorts react more strongly to shocks. As a result, a positive shock induces younger cohorts to invest relatively more than older cohorts, and vice-versa. We also characterize how an economy's reaction to a given shock varies with the agents' recency bias. 

%The results of this paper underline the importance of formally modeling the belief formation process of agents. This is not only relevant for improving our understanding of economic behavior, but also for effective policy making.  %The model generates several predictions about the relation between the demographics of market participants and market pricing and trading features,  some of them we test, and we hope the rest are explored in future research on this topic.

\newpage 

\singlespacing
\bibliographystyle{junpan}
%\bibliography{MPV_References}
\bibliography{MacroFinance}

% % % % % % % % % % % % % % % % % % % % % % % % % % % % % % % % % % % % % % % % % % % % % % % % % % % % % % % % % % % % % % % % % % % % % % % % % % % % % % % % % % % % % % % % % % % % % % % % % % % % % % % % % % % % % % % % % % % % % % % % % %
% % % % % % % % % % % % % % % % % % % % % % % % % % % % % % % % % % % % % % % % % % % % % % % % % % % % % % % % % % % % % % % % % % % % % % % % % % % % % % % % % % % % % % % % % % % % % % % % % % % % % % % % % % % % % % % % % % % % % % % % % %

\pagebreak
\begin{appendices}
	\setcounter{figure}{0}

	\section{Proofs for Results in Section \ref{sec:Baseline}}\label{app:Beliefs}
	
	{\small
	\begin{proof}[Proof of Lemma \ref{l: SingleCrossing}]
		Let $\Delta(k) \equiv w(k,\lambda, age) - w(k,\lambda, age')$ for all $k \in \{ 0,...,age\}$. We need to show that $\exists k_0\in \{0,...,age'\}$ such that $\Delta(k)<0$ for all $k\leq k_0$, and $\Delta(k)\geq 0$ for all $k>k_0$, with the last inequality holding strictly for some $k$.
		
		For $k>age'$, $\Delta(k)>0$ since $w(k,\lambda,age') \equiv 0 $, and hence $\Delta(k) = w(k,\lambda, age)>0$, for all $k \in \{age'+1,...,age\}$.
		
		For $k\leq age'$, we note that $
		\Delta(k) > 0 \iff Q(k) := \frac{w(k,\lambda,age)}{w(k, \lambda, age')} > 1$. Hence, it remains to be shown that $\exists k_0\in \{0,...,age'\}$ such that $Q(k)<1$ for all $k\leq k_0$, and $Q(k)\geq 1$ for all $k>k_0$. 	
		Since the normalizing constants used in the weights $w(k,\lambda,age)$ are independent of $k$ (see the definition in (\ref{eq:EBL-w})), we absorb them in a constant $c \in \mathbb{R^{+}}$ and rewrite
		\begin{equation}
		Q(k) = c \cdot \frac{ (age+1-k)^\lambda}{(age'+1-k)^\lambda} = c \cdot \Big[ \frac{ age+1-k}{age'+1-k} \Big]^\lambda = c \cdot \alpha(k)^\lambda\,~\forall k \in \{0,...,age'\}.
		\end{equation}
		
		The function $x\mapsto \alpha(x) = \frac{age+1-x}{age'+1-x}$ has derivative $\alpha'(x) = \frac{age - age'}{(age'+1-x)^2} > 0$ for $x \in [0,age'+1)$, and hence $Q(\cdot)$ is strictly increasing over $\{0,...,age'\}$.
		Thus, to complete the proof, we only have to show that $Q(k)<1$ or, equivalently, $\Delta(k)<0$ for some $k \in \{0,...,age'\}$. We know that $\sum_{k=0}^{age} \Delta(k) = 0$ because $\sum_{k=0}^{age} w(k,\lambda, age) = \sum_{k=0}^{age'} w(k,\lambda,age')=1$, and we also know that  
		$\sum_{k=age'+1}^{age} \Delta(k) > 0$ since $\Delta(k) = w(k,\lambda, age)>0$ for all $k \in \{age'+1,...,age\}$. Hence, it must be that $\Delta(k)<0$ for some $k<age'$.
	\end{proof}
}
	
	\section{Proofs for Results in  Section \ref{sec:results}}\label{app:results}
	
	%		For the proof of Proposition \ref{pro: demands} we need the following technical lemmas (their proofs are relegated to section \ref{sec:lemmas-linear-demands}). We use $z\mapsto \phi(z;\mu,\sigma^{2})$ to denote the Gaussian pdf with mean $\mu$ and $\sigma^{2}$. 
	
	Proposition \ref{pro: demandsstatic}  directly follows from the following Lemma.
	
	\begin{lemma}%[{l: static}]
		\label{l: static}
		Let $z \sim N(\mu,\sigma^{2})$, then for any $a > 0$, 
		\begin{align*}
		x^{\ast} = \arg\max_{x} E[- \exp \{ - a x z \} ] =&\frac{\mu} {a \sigma^{2}} 
		\end{align*}
		and
		\begin{align*}
		\max_{x} E[- \exp \{ - a x z \} ] =& - \exp \left\lbrace -  \frac{1}{2} ( \sigma a  x^{\ast} )^{2}   \right\rbrace  =  - \exp\left( - \frac{1}{2}\frac{ \mu^{2} }{ \sigma^2 }  \right).
		\end{align*}
	\end{lemma}
	
	{\small
		\begin{proof}[Proof of Lemma \ref{l: static}] Since $z \sim N(\mu,\sigma^2)$, we can rewrite the problem as follows:
		\begin{align*}
		x^*&=\arg\max_x -\exp\left(-axE[z]+\frac{1}{2}a^2x^2V[z]\right) \\
		&=\arg\max_x \;\;ax\mu-\frac{1}{2}a^2x^2\sigma^2
		\end{align*} 
		From FOC, $x^*=\frac{\mu}{a\sigma^2}$. Plugging $x^*$ into $-\exp\left(-a x^*\mu+\frac{1}{2}a^2(x^*)^2\sigma^2\right)$ the second result follows.
	\end{proof}

}
		{\small
	\begin{proof}[Proof of Proposition \ref{pro: prices_myopic}] 
		We show the result for the guess $p_{t}=\alpha+\beta_{0}d_{t}+...+\beta_{K}d_{t-K}$ with $K=q$. This case shows the logic of the proof; the proof for the case starting from an arbitrary lag $K \geq q$ is analogous but more involved, and omitted for simplicity.
		
		From Lemma \ref{l: static}, agents' demand for the risky asset is given by $x_{t}^{n}=\frac{E_{t}^{n}\left[s_{t+1}\right]}{\gamma V\left[s_{t+1}\right]}$. Plugging in our guess for prices, and for $\beta_0\neq -1$, we obtain:
\begin{equation}
	x_{t}^{n}=\frac{\left(1+\beta_{0}\right)\theta_{t}^{n}+\alpha+\beta_{1}d_{t}+...+
		\beta_{q}d_{t-q+1}
		-p_{t}R}{\gamma\left(1+\beta_{0}\right)^{2}\sigma^{2}}
\end{equation}
By market clearing, $\frac{1}{q}\sum_{n=t-q+1}^{t}x_{t}^{n} = 1$, which implies that
\begin{align*}
	&\frac{\left(1+\beta_{0}\right)\frac{1}{q}\sum_{n=t-q+1}^{t}\theta_{t}^{n}}{\gamma\left(1+\beta_{0}\right)^{2}\sigma^{2}}+\frac{\alpha+\beta_{1}d_{t}+...+\beta_{q}d_{t-q+1}-p_{t}R}{\gamma\left(1+\beta_{0}\right)^{2}\sigma^{2}}= 1 .
\end{align*}
By straightforward algebra and the definition of $\theta^{n}_{t}$, it follows that 
\begin{align*}	
	%			&\left(1+\beta_{0}\right)\frac{1}{q}\sum_{n=t}^{t-q+1}\theta_{t}^{n}+\left[\alpha+\beta_{1}d_{t}+...+\beta_{q}d_{t-q+1}-rp_{t}\right]= \gamma\left(1+\beta_{0}\right)^{2}\sigma^{2} \\
	%			&\left(1+\beta_{0}\right)\frac{1}{q}\sum_{n=t}^{t-q+1}\theta_{t}^{n}+\left[\alpha-\gamma\left(1+\beta_{0}\right)^{2}\sigma^{2}\right]+\beta_{1}d_{t}+...+\beta_{q}d_{t-q+1}=rp_t\\
	\left(1+\beta_{0}\right)\frac{1}{q}\sum_{n=t-q+1}^{t}\left[\sum_{k=0}^{t-n}w\left(k,\lambda,t-n\right)d_{t-k}\right]+\left[\alpha-\gamma\left(1+\beta_{0}\right)^{2}\sigma^{2}\right]+\beta_{1}d_{t}+...+\beta_{q}d_{t-q+1}=p_{t}R.
\end{align*}
Plugging in (again) our guess for $p_t$ and using the method of undetermined coefficients, we find the expressions for $\alpha$ and the $\beta$'s:
\begin{align}
	-\frac{\gamma\left(1+\beta_{0}\right)^{2}\sigma^{2}}{R-1} &= \alpha \label{eq: alpha_interm}\\
	\left(1+\beta_{0}\right)\frac{1}{q}\sum_{n=t-q+1}^{t-k}w\left(k,\lambda,t-n\right)+\beta_{k+1}&=\beta_{k}R\quad\quad \forall k\in\{0,1,...,q-1\}\\
	0&=\beta_{q}R. \label{eq: 0=bq}
\end{align}

Let $w_{k}$ be the average of the weights assigned to dividend $d_{t-k}$ by each generation in the market at time $t$, i.e., $w_{k}=\frac{1}{q}\sum_{n=t-q+1}^{t}w\left(k,\lambda,t-n\right)$.
%As defined in Proposition \ref{pro: demandsstatic}, $w_{k}$ is the average of the weights assigned to dividend $d_{t-k}$ by each generation in the market, $w_{k}=\frac{1}{q}\sum_{n=t-q+1}^{t}w\left(k,\lambda,t-n\right)$. 
Given that a weight of zero is assigned to dividends that a generation did not observe, i.e., for $k>t-n$, we can rewrite $w_{k}=\frac{1}{q}\sum_{n=t-q+1}^{t-k}w\left(k,\lambda,t-n\right)$. Also using $\beta_{q}=0$ from equation (\ref{eq: 0=bq}) we obtain:
\begin{align}
	\left(1+\beta_{0}\right)w_{k}+\beta_{k+1}&=\beta_{k}R\quad\quad \forall k\in\{0,1,...,q-2\}  \label{eq: betak}\\
	\left(1+\beta_{0}\right)w_{q-1}&=\beta_{q-1}R  \label{eq: betaq}
\end{align}
By solving this system of equations, we obtain the expressions in the proposition. In particular, $\left(1+\beta_{0}\right) (w_{q-2}+w_{q-1}/R )=\beta_{q-2}R$ for $k=q-2$, $\left(1+\beta_{0}\right) (w_{q-3} + w_{q-2}/R+ w_{q-1}/R^{2} )=\beta_{q-3}R$ for $k=q-3$, and so on. This allow us to express (\ref{eq: betak}) and (\ref{eq: betaq}) as
\begin{align}
	(1+\beta_{0}) \sum_{j=0}^{k-1} w_{q-(k-j)}/R^{j}  = \beta_{q-k}R\quad \textit{for}~k= 1,...,q. \label{eq:1+beta0}
\end{align}
The last expression (\ref{eq:1+beta0}) implies $\beta_{0} = \frac{\sum_{j=0}^{q-1}  w_{j}/R^{j}  }{R - \sum_{j=0}^{q-1} w_{j}/R^{j}} = \frac{\sum_{j=0}^{q-1}  w_{j}/R^{j+1}  }{ 1- \sum_{j=0}^{q-1} w_{j}/R^{j+1}} $ (from plugging in $k=q$), which in turn, plugged into (\ref{eq: alpha_interm}) allows us to obtain the expression for $\alpha$ from (\ref{eqn:alpha}) in Proposition \ref{pro: prices_myopic}. And expression (\ref{eq:1+beta0}) implies $\beta_{k}= \frac{\sum_{j=0}^{q-1-k}  w_{k+j}/R^{j+1}}{1- \sum_{j=0}^{q-1} w_{j}/R^{j+1}}$ (from substituting $k$ with $q-k$, and using the expression for $\beta_{0}$) as expressed in equation  (\ref{eqn:beta_k}) of the Proposition. The latter also subsumes equation (\ref{eq: betaq}), solved for $\beta_{q-1}$, and the above formula for $\beta_{0}$, and hence holds for $k=0,...q-1$.
\end{proof}

	\begin{proof}[Proof of Proposition \ref{pro:prices-q2-myopic}] 
		For this proof, we use equations (\ref{eq: betak}) and (\ref{eq: betaq}). In addition, note that by construction, $w_k < w_{k-1}$ for $\lambda>0$ since for all generations, $w(k,\lambda,age)$ is decreasing in $k$ and more agents observe the realization of $d_{t-(k-1)}$ than $d_{t-k}$. Given this, it follows that since $\beta_0>0$ then $\beta_{q-1}>0$ and
		\begin{equation}
		\beta_{q-1}=\frac{1}{R}\left(1+\beta_{0}\right)w_{q-1}< \frac{1}{R}[\left(1+\beta_{0}\right)w_{q-2}+\beta_{q-1}]=\beta_{q-2}
		\end{equation}
		In addition, if $\beta_k<\beta_{k-1}$, then:
		\begin{equation}
		\beta_{k-1}=\frac{1}{R}[\left(1+\beta_{0}\right)w_{k-1}+\beta_k] < \frac{1}{R}[\left(1+\beta_{0}\right)w_{k-2}+\beta_{k-1}]=\beta_{k-2}
		\end{equation}
		Thus, the proof that $\beta_k<\beta_{k-1}$ for all $k\in\{1,...,q-1\}$ follows by induction.	
	\end{proof}

	\begin{proof}[Proof of Lemma \ref{lem: compstatics}] To show that $\beta_0$ is increasing in $\lambda$, let $G_q(\lambda)=\sum_{k=0}^{q-1} w_{k}/R^{k+1}$. We thus have $\beta_0=\frac{G_q(\lambda)}{1-G_q(\lambda)}$, and it suffices to show that $G_q'(\lambda)>0\;\;\forall q>0$ and  $\forall \lambda>0$. After some algebra, the terms in $G_q(\cdot)$ can be re-organized as follows:
		\begin{equation}
		G_q(\lambda)=\sum_{age=0}^{q-1} \frac{1}{q} \sum_{k=0}^{age} w(k,\lambda,age)/R^{k+1}
		\end{equation}
	Note that for any $age\in\{0,...,q-1\}$: (i) $\sum_{k=0}^{age} w(k,\lambda,age)=1$ and (ii) for any $\lambda_1,\lambda_2$ such that $\lambda_1>\lambda_2>0$,  $\sum_{k=j}^{age} w(k,\lambda_1,age)<\sum_{k=j}^{age} w(k,\lambda_2,age)$. Thus, the weight distribution given by $\lambda_{2}$ first-order stochastically dominates the weight distribution given by $\lambda_{1}$. Since $1/R>1/R^2>1/R^3>...>1/R^{q-1}$, stochastic dominance implies that for all $age\in \{0,...,q-1\}$, $\sum_{k=0}^{age} c^{k+1} w(k,\lambda_1,age)>\sum_{k=0}^{age} c^{k+1} w(k,\lambda_2,age)$, and thus $G_q(\lambda_1)>G_q(\lambda_2)$. 
	
	\smallskip 
	
	To show the limit results, note that $\lim_{\lambda \rightarrow \infty} w(0,\lambda,age) = 1$, while $\lim_{\lambda \rightarrow \infty} w(k,\lambda,age) = 0$ for all $k>0$.
	\end{proof}

	\begin{proof}[Proof of Proposition \ref{p: rel_demands-myopic}]
		From Propositions \ref{pro: demandsstatic} and \ref{pro: prices_myopic}, we know that, for any $t$, any generations $m \geq n$ both in $\{t-q+1, ..., t\}$ and any $ k \in \{ 0,....,q-1\}$,
		\begin{align*}
			\frac{\partial (x^{n}_{t} - x^{m}_{t}) }{\partial d_{t-k}} = \frac{(1+\beta_{0}) }{\gamma V[s_{t+1}]} \frac{\partial (\theta^{n}_{t} - \theta^{m}_{t}) }{\partial d_{t-k}}.
		\end{align*}
		We note that, for any $n \in \{t-q+1,...,t\}$, $\frac{ \partial \theta^{n}_{t}} {\partial d_{t-k}} = w(k, \lambda , n-t)$ if $k \in \{ 0,...,t-n \}$, and $\frac{ \partial \theta^{n}_{t}} {\partial d_{t-k}} = 0$ if $k \in \{ t-n+1,...,q-1\}$. (Observe that $t -n \leq q-1$.) Hence, it suffices to compare $w(k, \lambda , t-n)$ with $w(k, \lambda , t-m)$ for any $k\in\{ 0,...,q-1  \}$. (As usual, here we adopt the convention that for any $age$, $w(k,\lambda,age) =  0$ for all $k \geq age$.)		
		From Lemma \ref{l: SingleCrossing}, there exists a $k_{0}$ such that $w(k,\lambda,t-n) < w(k,\lambda,t-m)$ for all $k \in \{ 0,..., k_{0}\}$ and $w(k,\lambda,t-n) \geq w(k,\lambda,t-m)$ for the rest of the $k$'s, $k \in \{k_0+1,...,q-1\}$. 		
	\end{proof}
}

\medskip

\noindent The proof of Proposition  \ref{p: rel_demands-myopic-boom} relies on the following first-order stochastic dominance result: 

%	\begin{lemma} \label{l: FOSD} For any $a \in \{0,1,...\}$ and any $m \in \{0,...,a\}$, let $F(m,a)\equiv \sum_{j=0}^m w(j,\lambda,a)$. 
%		If $j \mapsto w(j,\lambda,a) - w(j,\lambda,a')$ for $a'<a$ is increasing in $j$, then $F(m,a)\leq F(m,a')$ for all $m\in\{0,...,a\}$.
%	\end{lemma}

\begin{lemma} \label{l: FOSD} For any $a \in \{0,1,...\}$, $a' < a$ and any $m \in \{0,...,a\}$, let $F(m,a)\equiv \sum_{j=0}^m w(j,\lambda,a)$. 
		Suppose the conditions of Lemma \ref{l: SingleCrossing} hold; then $F(m,a)\leq F(m,a')$ for all $m\in\{0,...,a\}$.
\end{lemma}
	
	\begin{proof}[Proof of Lemma \ref{l: FOSD}]
	From Lemma \ref{l: SingleCrossing}, we know that there exists a unique $j_0$ where $w(j_0,\lambda,a')-w(j_0,\lambda,a)$ ``crosses" zero. Thus, for  $m\leq j_0$, the result is true because $w(j,\lambda,a')>w(j,\lambda,a)$ for all $j\in \{0,...m\}$. For $m> j_0$, the result follows from the fact that $w(j,\lambda,a')<w(j,\lambda,a)$ for all $j\in \{m,...a\}$ and $F(a,a)=F(a',a')=1$. 
\end{proof}

{\small
	
	\begin{proof}[Proof of Proposition \ref{p: rel_demands-myopic-boom}] 
		
		We first introduce some notation. For any $j \in \{ t-n-k+1,...,t-n\}$, let $w(j,\lambda,t-n-k) = 0$; i.\,e., we define the weights of generation $n+k$ for time periods before they were born to be zero. Thus,  $\sum_{j=0}^{t-n-k}w(j,\lambda,t-n-k)  d_{t-j} = \sum_{j=0}^{t-n}w(j,\lambda,t-n-k)  d_{t-j}$. In addition, we note that $(w(j,\lambda,t-n-k))_{j=0}^{t-n}$ and $(w(j,\lambda,t-n))_{j=0}^{t-n}$ are sequences of positive weights that add to one. 
		
		Let for any $m \in \{0,...,t-n\}$,
		\begin{align*}
		F(m,t-n-k) = \sum_{j=0}^{m} w(j,\lambda,t-n-k) \mbox{ and } F(m,t-n) = \sum_{j=0}^{m} w(j,\lambda,t-n).
		\end{align*}
		These quantities, as functions of $m$, are non-decreasing and $F(t-n,t-n-k) = F(t-n,t-n) = 1$. Moreover, $F(m+1,t-n-k)-F(m,t-n-k) = w(m+1,\lambda,t-n-k)$ and $F(m+1,t-n)-F(m,t-n) = w(m+1,\lambda,t-n)$. Finally, we set $F(-1,t-n) = F(-1,t-n-k) = 0$.   
		
		By these observations, by the definition of $\xi(n,k,t)$, and by straightforward algebra, it follows that,
		{\small
			\begin{align*}
			& \xi(n,k,t) \\
			= &  \frac{\sum_{m=0}^{t-n} (F(m,t-n)-F(m-1,t-n) )  d_{t-m} - \sum_{m=0}^{t-n} (F(m,t-n-k)-F(m-1,t-n-k) )  d_{t-m}}{\gamma (1+\beta_{0}) \sigma^{2}}\\
			%= & \frac{F(0,t-n)  d_{t} + (F(1,t-n)-F(0,t-n) )  d_{t-1} + ...+ (1-F(t-n-1,t-n) )  d_{n}  }{\gamma (1+\beta_{0}) \sigma^{2}}\\
			%& - \frac{F(0,t-n-k)  d_{t} + (F(1,t-n-k)-F(0,t-n-k) )  d_{t-1} + ...+ (1-F(t-n-1,t-n-k) )  d_{n}}{\gamma (1+\beta_{0}) \sigma^{2}}\\
			%= & \frac{(d_{t}-d_{t-1})F(0,t-n) + (d_{t-1}-d_{t-2})F(1,t-n) + ...+ (d_{n+1}-d_{n})F(t-n-1,t-n) + d_{n} }{\gamma (1+\beta_{0}) \sigma^{2}}\\
			%& - \frac{(d_{t}-d_{t-1})F(0,t-n-k) + (d_{t-1}-d_{t-2})F(1,t-n-k) + ...+ (d_{n+1}-d_{n})F(t-n-1,t-n-k) + d_{n}}{\gamma (1+\beta_{0}) \sigma^{2}}\\	
			= & \frac{\sum_{j=0}^{t-n-1} (d_{t-j}-d_{t-j-1}) (F(j,t-n) - F(j,t-n-k)) }{\gamma (1+\beta_{0}) \sigma^{2}}.  
			\end{align*}}
		
		If the weights are non-decreasing, then $d_{t-j}-d_{t-j-1} \geq 0$ for all $j=0,...,t-n-1$, and it suffices to show that $F(j,t-n) \leq F(j,t-n-k)$ for all $j=0,...,t-n-1$. This follows from applying Lemma  \ref{l: FOSD}  with $a=t-n > t-n-k=a'$.
%		 To show this, we note that by Lemma \ref{l: SingleCrossing} the hypothesis in Lemma \ref{l: FOSD} holds. Thus, the result follows from applying the latter lemma with $a=t-n > t-n-k=a'$ and $j \in \{ 0,...,t-n\}$.
		
		If the weights are non-increasing, then $d_{t-j}-d_{t-j-1} \leq 0$, and the sign of $\xi(n,k,t)$ changes accordingly.
		%			note that, for any $m=0,...,t-n+1$, 
		%			\begin{align*}
		%			F(m,\lambda,t-n) = \sum_{j=0}^{m} w(j,\lambda,t-n) = \sum_{j=0}^{m} \frac{(t-n+1-j)^{\lambda}}{\sum_{k'=0}^{t-n} (t-n+1-k')^{\lambda}} .
		%			\end{align*}
		%			Therefore, to show the inequality, it suffices to show that, $m=0,...,t-n-k$ 
		%			\begin{align}\label{eqn:prop-xi-1}
		%			\sum_{j=0}^{m} \frac{(t-n+1-j)^{\lambda}}{\sum_{k'=0}^{t-n} (t-n+1-k')^{\lambda}} \leq \sum_{j=0}^{m} \frac{(t-n-k+1-j)^{\lambda}}{\sum_{k'=0}^{t-n-k} (t-n-k+1-k')^{\lambda}} 
		%			\end{align}
		%			since for any $m=t-n-k+1,...,t-n$, $ \sum_{j=0}^{m} \frac{(t-n+1-j)^{\lambda}}{\sum_{k'=0}^{t-n} (t-n+1-k')^{\lambda}} \leq 1$, which always holds. This has been shown in Lemma \ref{l: SingleCrossing} followed by Lemma \ref{l: FOSD}.
	\end{proof}
	
	%\section{Proofs of Subsection \ref{sec:tradevolume}}\label{app:tradevolume}
	
	\begin{proof}[Proof of Proposition \ref{l: TradeVolumeM}]
		By Propositions \ref{pro: demandsstatic} and \ref{pro: prices_myopic}, it follows that for any $t$ and $n \leq t$,
		%\begin{align*}
		%		x_{t}^{n}=&\frac{\alpha_{0}+\left(1+\beta_{0}\right)\theta_{t}^{n}+\beta_{1}d_{t}+...+\beta_{q-1}d_{t-q+2}-R\left(\alpha_{0}+\beta_{0}d_{t}+...+\beta_{q-1}d_{t-q+1}\right)}{\gamma\sigma^{2}\left(1+\beta_{0}\right)^{2}},
		%\end{align*}
		% Felix: the equation above is rewritten using summation instead of "..."
		
		\begin{align}
		x_{t}^{n} %&=\frac{1}{\gamma\sigma^{2}\left(1+\beta_{0}\right)^{2}}\left(  \alpha_0 + (1+\beta_0)\theta^n_t + \sum_{k=1}^{q-1}\beta_k d_{t+1-k} - R \left(  \alpha_0 + \sum_{k=0}^{q-1} \beta_k d_{t-k}  \right) \right) \\
		 = \frac{1}{\gamma\sigma^{2}\left(1+\beta_{0}\right)^{2}}
		\left(  \alpha_0 (1-R) + (1+\beta_0)\theta^n_t - R \beta_{0}d_{t}  + \sum_{k=1}^{q-1}\beta_k ( d_{t+1-k} - R d_{t-k})  \right). \label{eqn:TradeVolumeM-0} 
		\end{align}
		
		Thus, for $n\in\{t-q+1,...,t-1\}$,
		\begin{align} 
		x_{t}^{n}-x_{t-1}^{n}= \frac{ (1+\beta_{0})(\theta_{t}^{n}-\theta_{t-1}^{n}) + \mathcal{T}(d_{t:t-q}) 
		}{\gamma\sigma^{2}\left(1+\beta_{0}\right)^{2}}  \label{eqn:TradeVolumeM-1}
		%x_{t}^{n}-x_{t-1}^{n}= \frac{1}{\gamma\sigma^{2}\left(1+\beta_{0}\right)^{2}}\left( (1+\beta_{0})(\theta_{t}^{n}-\theta_{t-1}^{n}) - R\beta_{q-1}\left(d_{t-q+1}-d_{t-q}\right) + \sum_{k=1}^{q-1}(\beta_k-R\beta_{k-1}) ( d_{t+1-k} - R d_{t-k}) \right)
		\end{align}
		where $\mathcal{T}(d_{t:t-q}) \equiv \sum_{k=1}^{q-1}\beta_k ( d_{t+1-k} - d_{t-k} - R (d_{t-k}-d_{t-1-k})) - R \beta_{0}(d_{t}-d_{t-1})$. Note that $\mathcal{T}(d_{t:t-q})$ is not cohort specific, i.e., does not depend on $n$. 
		
		The fact that $x_t^t-x_{t-1}^t=x_t^t$ and $x_t^{t-q}-x_{t-1}^{t-q}=-x_{t-1}^{t-q}$, and market clearing imply 
				\begin{align} 
		q^{-1} \left( \sum_{n=t-q}^{t} x^{n}_{t} - x^{n}_{t-1} \right) =0. 
		\end{align}
		This expression and the expression in \eqref{eqn:TradeVolumeM-1} imply that
	\begin{align*}
	%\frac{1}{q}\left(\sum_{n=t-q}^{t}x_{t}^{n}-x_{t-1}^{n}\right) = & \frac{1}{q}\sum_{n=t-q}^{t}
	%\frac{ (1+\beta_{0})(\theta_{t}^{n}-\theta_{t-1}^{n}) + \mathcal{T}(d_{t:t-q}) 
	%}{\gamma\sigma^{2}\left(1+\beta_{0}\right)^{2}} = 0  \\
	%\Rightarrow \;\; 
	\frac{1}{q} \left( \sum_{n=t-q+1}^{t-1}
	\frac{ (1+\beta_{0})(\theta_{t}^{n}-\theta_{t-1}^{n})}{\gamma\sigma^{2}\left(1+\beta_{0}\right)^{2}}  +  x^{t}_{t} - x^{t-q}_{t-1}\right) = -\frac{1}{q}\sum_{n=t-q}^{t}\frac{ \mathcal{T}(d_{t:t-q}) 
	}{\gamma\sigma^{2}\left(1+\beta_{0}\right)^{2}} = -\frac{ \mathcal{T}(d_{t:t-q}) 
}{\gamma\sigma^{2}\left(1+\beta_{0}\right)^{2}}.
\end{align*} 		
Letting $\theta^{t}_{t-1} = \theta^{t-q}_{t} = 0$, it follows that 
	\begin{align*}
	\frac{1}{q} \left( \sum_{n=t-q}^{t}
 (1+\beta_{0})(\theta_{t}^{n}-\theta_{t-1}^{n})  \right) = -\mathcal{T}(d_{t:t-q}).
\end{align*} 
Thus, we can express the change in individual demands for those agents with $n=\{t-q+1,...,t-1\}$ in expression \eqref{eqn:TradeVolumeM-1} as follows:
\begin{equation}
x_{t}^{n}-x_{t-1}^{n}=\chi \left[\left(\theta_{t}^{n}-\theta_{t-1}^{n}\right)-\frac{1}{q}\sum_{n=t-q}^{t}\left(\theta_{t}^{n}-\theta_{t-1}^{n}\right)\right],~\forall n \in \{ t,...,t-q  \}
\end{equation}
where $\chi\equiv\frac{1}{\gamma\sigma^{2}\left(1+\beta_{0}\right)}$. By squaring and summing at both sides and including the demands on the youngest ($n=t$) and oldest ($n=t-q$) market participants the desired result follows. 		

	\end{proof}
}
	
	\pagebreak

	\section{Incorporating Prior Beliefs} \label{app: PriorBeliefs}

	In this section, we show how the model can be extended to allow agents to have prior beliefs; that is, $\tau>0$ in \eqref{eq: PosteriorMeanFormula}.  We will prove results analogous to those in Proposition \ref{pro: prices_myopic}. 
	
	As a reminder, we now suppose that all cohorts are born with prior belief $N(m,\sigma_m^2)$, and update their beliefs during their lifetime as follows:
\begin{equation}
\theta_t^n = (1 - \omega_{t-n}) m + \omega_{t-n} \left[\sum_{k=0}^{t-n}w\left(k,\lambda,t-n\right)d_{t-k}\right]
\end{equation}
where $\omega_{t-n}$ is given by $$\omega_{t-n} = \frac{t-n+1}{\tau + (t-n+1)},$$ and where $\tau$ captures the relative importance of prior beliefs to experience-based beliefs. For example, if agents are Bayesian from experience as described in Section \ref{sec:EBL}, then $\tau = \frac{\sigma^{2}}{\sigma_m^{2}}$. For the purpose of our analysis, however, all that is important is how results vary with $\tau$.

We continue to guess that prices are affine in past dividends, $$p_{t}=\alpha +\beta_{0} d_{t}+...+\beta_{K}d_{t-K}$$ with $K=q$, as in the baseline model. From Lemma \ref{l: static}, agents' demand for the risky asset is given by $x_{t}^{n}=\frac{E_{t}^{n}\left[s_{t+1}\right]}{\gamma V\left[s_{t+1}\right]}$. Plugging in our guess for prices, and for $\beta_0\neq -1$, we obtain:

\begin{equation}
x_{t}^{n}=\frac{\left(1+\beta_{0}\right)\theta_{t}^{n}+\alpha+\beta_{1}d_{t}+...+
	\beta_{q}d_{t-q+1}
	-p_{t}R}{\gamma\left(1+\beta_{0}\right)^{2}\sigma^{2}}
\end{equation}

By market clearing, $\frac{1}{q}\sum_{n=t-q+1}^{t}x_{t}^{n} = 1$, which implies that

\begin{align*}
\frac{\left(1+\beta_{0}\right)\frac{1}{q}\sum_{n=t-q+1}^{t}\theta_{t}^{n}}{\gamma\left(1+\beta_{0}\right)^{2}\sigma^{2}}+\frac{\alpha+\beta_{1}d_{t}+...+\beta_{q}d_{t-q+1}-p_{t}R}{\gamma\left(1+\beta_{0}\right)^{2}\sigma^{2}}= 1 .
\end{align*}

By straightforward algebra and the definition of $\theta^{n}_{t}$, it follows that 

\begin{align*}	
\left(\frac{1}{\gamma\left(1+\beta_{0}\right)^{2}\sigma^{2}}\right)\left[ \left(1+\beta_{0}\right)\frac{1}{q}\sum_{n=t-q+1}^{t}	\theta_{t}^{n}+ \alpha+\beta_{1}d_{t}+...+\beta_{q}d_{t-q+1} - p_{t} R\right] & =  1 \\
%\left(1+\beta_{0}\right)\frac{1}{q}\sum_{n=t-q+1}^{t}	\theta_{t}^{n}+\alpha-\gamma\left(1+\beta_{0}\right)^{2}\sigma^{2}+\beta_{1}d_{t}+...+\beta_{q}d_{t-q+1} &=  p_t R\\
\left(1+\beta_{0}\right)\frac{1}{q}\sum_{n=t-q+1}^{t}	\left[(1-\omega_{t-n}) m + \omega_{t-n} \sum_{k=0}^{t-n}w\left(k,\lambda,t-n\right)d_{t-k}\right]+\alpha-\gamma\left(1+\beta_{0}\right)^{2}\sigma^{2} ... &\\
+\beta_{1}d_{t}+...+\beta_{q}d_{t-q+1} & =  p_{t}R. \\
\left[\left(1+\beta_{0}\right)\frac{1}{q}\sum_{n=t-q+1}^{t} (1- \omega_{t-n}) m +\alpha-\gamma\left(1+\beta_{0}\right)^{2}\sigma^{2}\right]  + ...& \\  \left(1+\beta_{0}\right)\frac{1}{q}\sum_{n=t-q+1}^{t}	\sum_{k=0}^{t-n}  \omega_{t-n}  w\left(k,\lambda,t-n\right)d_{t-k}
+\beta_{1}d_{t}+...+\beta_{q}d_{t-q+1} & =  p_{t}R.
\end{align*}

Plugging in (again) our guess for $p_t$ and using the method of undetermined coefficients, we find the expressions for $\alpha$ and the $\beta$'s:

\begin{align}
\frac{\gamma\left(1+\beta_{0}\right)^{2}\sigma^{2}+\left(1+\beta_{0}\right)\frac{1}{q}\sum_{n=t-q+1}^{t} (1- \omega_{t-n}) m}{1-R} &= \alpha \label{eq: alpha_interm}\\
\left(1+\beta_{0}\right)\frac{1}{q}\sum_{n=t-q+1}^{t-k} \omega_{t-n} w\left(k,\lambda,t-n\right)+\beta_{k+1}&=\beta_{k}R\quad\quad \forall k\in\{0,1,...,q-1\}\\
0&=\beta_{q}R. \label{eq: 0=bq}
\end{align}

Where $w_{k}$ is now the average of the weights assigned to dividend $d_{t-k}$ by each generation in the market at time $t$, i.e., $w_{k}=\frac{1}{q}\sum_{n=t-q+1}^{t} \omega_{t-n} w\left(k,\lambda,t-n\right)$. 

Introducing prior beliefs requires two adjustments. First, the constant in prices, $\alpha$, now increases to incorporate the demand driven by prior belief, $m$. Second, all the weights that an agent with age $t-n$ gives to past dividends are now adjusted by $\omega(t-n)$, which keeps track of the importance that these agents assign to their experience-based learning. Such adjustment affects the $\beta's$ in the pricing equation. Given these adjustments, the model is isomorphic to the baseline model. 
	
	\pagebreak 
	
	\section{Population Growth} \label{oa:demographics}
	
	In addition to considering the effects of a one-time shock to population structure, we also explore the implications of population growth. 
	
	In this section of the Online Appendix, we consider an OLG model two-period lived agents where the mass of young agents born every period grows at rate $g$. For this growth setting, we need to set an initial date for the economy, which we define to be $t=0$. Let $y_{t}$ denote the mass of young agents born at time $t$; then $y_{t+1}=\left(1+g\right)y_{t}=y_0(1+g)^t$. We further denote the total mass of people at any point in time $t>0$ as $n_{t}$, and hence $n_{t}=y_t+y_{t-1}=\left(2+g\right)y_{t-1}$. It is easy to check that $n_t=\left(1+g\right)n_{t-1}$; that is, total population grows at rate $g$.
	
	The framework is otherwise as in the `toy model" in Section \ref{sec: ToyModel} of the main paper. The main difference is that now population is growing over time. As a result, we make a different guess for the price function:
	\[p_{t}=\alpha_{0}\left(1+g\right)^{-t}+\beta_{0}d_{t}+\beta_{1}d_{t-1}	\]		
	We verify this guess using our market clearing condition, which requires the demand of the young and the old  to add up to total supply of the asset, one:
	{\small
		\begin{align*}
		1&=y_{t}\frac{E_{t}^{t}\left[p_{t+1}+d_{t+1}\right]-Rp_{t}}{\gamma V\left[p_{t+1}+d_{t+1}\right]}+y_{t-1}\frac{E_{t}^{t-1}\left[p_{t+1}+d_{t+1}\right]-Rp_{t}}{\gamma V\left[p_{t+1}+d_{t+1}\right]} \iff \\
		1 & =  \frac{y_{0}\left(1+g\right)^{t-1}}{\gamma\left(1+\beta_{0}\right)^{2}\sigma^{2}}\left[\left(1+\beta_{0}\right)\left[\left(1+g\right)E_{t}^{t}\left[d_{t+1}\right]+E_{t}^{t-1}\left[d_{t+1}\right]\right]+\left(2+g\right)\left[\alpha_{0}\left(1+g\right)^{-\left(t+1\right)}+\beta_{1}d_{t}-Rp_{t}\right]\right]\\
		\end{align*}}
	and after simple algebra,
	\begin{align*}
		Rp_t  =  \left(1+\beta_{0}\right)\left\{ \frac{1+g}{2+g}d_{t}+\frac{1}{2+g}\left[\left(1-\omega\right)d_{t-1}+\omega d_{t}\right]\right\} + \frac{\alpha_{0}}{\left(1+g\right)^{t+1}}+\beta_{1}d_{t}-\frac{\gamma \sigma^{2} (1+\beta_{0})^{2}}{y_{0}\left(2+g\right)\left(1+g\right)^{t-1}}
	\end{align*}
	We plug in $p_t=\alpha_{0}\left(1+g\right)^{-t}+\beta_{0}d_{t}+\beta_{1}d_{t-1}$ and we use the method of undetermined coefficients to obtain:		
	\begin{align*}
	%\frac{r\alpha_{0}}{\left(1+g\right)^{t}} & =  \frac{\alpha_{0}}{\left(1+g\right)^{t+1}}-\frac{\gamma\left(1+\beta_{0}\right)^{2}\sigma^{2}}{y_{0}\left(2+g\right)\left(1+g\right)^{t-1}}\\
	%r\alpha_{0} & =  \frac{\alpha_{0}}{\left(1+g\right)}-\frac{\left(1+g\right)\gamma\left(1+\beta_{0}\right)^{2}\sigma^{2}}{y_{0}\left(2+g\right)}\\
	\alpha_{0} & =  -\frac{\gamma\left(1+\beta_{0}\right)^{2}\sigma^{2}}{R-\frac{1}{1+g}}\frac{\left(1+g\right)}{y_{0}\left(2+g\right)} \\
	R\beta_{0} &=\left(1+\beta_{0}\right)\left(\frac{1+g}{2+g}+\frac{1}{2+g}\omega\right) +\beta_{1} \\
	R\beta_{1} &=\left(1+\beta_{0}\right)\frac{1-\omega}{2+g}
	\end{align*}
	
	Let $\alpha_t \equiv \alpha_0 (1+g)^{-t}$ and $\gamma \equiv \frac{y_t}{n_t}$ denote the fraction of young agents, which is easy to verity is constant over time. Then, we can rewrite the above equations as
	
	\begin{align*}
	\alpha_{t} & =  -\frac{\gamma\left(1+\beta_{0}\right)^{2}\sigma^{2}}{R-\frac{1}{1+g}}\frac{1+g}{n_t} \\
	R\beta_{0} &=\left(1+\beta_{0}\right)\left(\gamma+(1-\gamma)\omega\right) +\beta_{1} \\
	R\beta_{1} &=\left(1+\beta_{0}\right)(1-\gamma)(1-\omega).
	\end{align*}
	
The latter expressions reveal that the total mass of agents in the market is reflected only in the price constant, while the fraction of young people in the market determines the dividend loadings $\beta_0$ and $\beta_1$. Overall, we see that adding population growth generates to our model generates a positive trend in prices. The relative reliance of prices on the most recent experiences (dividends) is increasing  in the population growth rate.

\pagebreak 
\counterwithin{table}{section}
\renewcommand\thetable{OA.\arabic{table}}
\section{Empirical Analysis}\label{sec:Emperical Analysis}

We use two alternative approaches to measure the fraction of younger agents (below 50 years of age) in the market. First, we compute an indicator variable that equals one when the fraction of young agents in the market is above 0.5 and zero otherwise, 
%\begin{align*}
%Y^1_t = \mathbb{I}\{\text{Fraction of young investors}_t > 0.5\}.
%\end{align*}
$\mathbb{I}\{\text{Fraction of young investors}_t > 0.5\}$.
Here, the fraction of young investors is based on their relative cohort sizes, with
\begin{align*}
\text{Fraction of young investors}_t =  \frac{\sum_{j}\mathbb{I}(\text{age}_{j,t}<50) \cdot w^{\text{scf}}_{j,t}}{\sum_{j=1} w^{\text{scf}}_{j,t}},
\end{align*}
where $\text{age}_{j,t}$ is the age of household head $j$ in year $t$, and $w^{\text{scf}}_{j,t}$ is the weight given to household head $j$ in year $t$ in the Survey of Consumer Finance to compensate for unequal probabilities of household selection in the original design and for unit nonresponse (failure to obtain an interview). 

Our second proxy captures the wealth of younger generations of investors. We construct an indicator variable that equals 1 when the fraction of liquid wealth owned by agents below 50 is above the 1960-2013 sample average of their liquid wealth and zero otherwise, $\mathbb{I}(\text{Fraction of young investors' wealth}_t > \text{Sample average})$,\footnote{The SCF documents the age and wealth (liquid assets) information of each respondent in 1960, 1962, 1963, 1964, 1967, 1968, 1969, 1970, 1971, 1977, 1983, 1986, 1989, 1992, 1995, 1998, 2001, 2004, 2007, 2010, and 2013. We use linear interpolation to fill the missing years and construct a yearly sample from 1960 to 2013. The liquid assets variable is defined to be the sum of assets in an investor's checking, savings, and money-market accounts, as well as any call accounts at brokerages and prepaid cards.} 
%\begin{align*}
%Y^2_t = \mathbb{I}(\text{Fraction of young investors' wealth}_t > \text{Sample average})
%\end{align*}
i.\,e., \begin{align*}
\text{Fraction of young investors' wealth}_t =&  \frac{\sum_{j} \mathbb{I}(\text{age}_{j,t}<50) \cdot w^{scf}_{j,t} \cdot \text{Wealth}_{j,t}} {\sum_{j} w^{scf}_{j,t} \cdot \text{Wealth}_{j,t}} %\\
%\text{Sample average} =& \frac{ \sum_{t=1960}^{2013} \text{Fraction of young investors' wealth}_t}{2013-1960+1}
\end{align*}

For robustness, we also consider thresholds $0.55$ and $0.60$ for the age-based first proxy, and $0.9\times \text{Sample average}$ and $1.1\times \text{Sample average}$ for the age- and wealth-based second proxy. Results are presented in Online-Appendix Table \ref{tab:onapp_robustness}. We estimate a positive $\delta_1$ coefficient, which is significant when requiring a fraction of 0.6 for the age-based coefficient and when requiring wealth above 0.9 of the sample average for the age- and wealth-based coefficient.

\begin{table} 
%\ref{tab:onapp_robustness}.
	\centering
	\caption{Markov-Switching Regime (MSR) model} 
	\label{tab:onapp_robustness}
	
	\begin{small}
		\begin{tablenotes}
		\item Robustness checks of the estimation results in Table \ref{tab:return}. $p_t - d_t$ is the log of the price-to-dividend ratio, and regressed on its lagged values interacted with the demographic indicator variable $Y_t$ for the fraction of young investors. We use different thresholds to construct $Y_t$. In column (1), $Y_t$ equals 1 when the fraction of investors below 50 is larger than 0.55, and in column (2), the threshold is 0.60. In column (3), $Y_t$ equals 1 when the fraction of wealth of investors below 50 is larger than 90\% of their 1960-2013 sample average, and in column (4), the threshold is 110\% of the sample average. As in Table \ref{tab:return}, the demographic data including age and wealth (liquid assets) of stock-market participants is from the SCF, stock data from Robert Shiller's website.
		\end{tablenotes}
	\end{small}
	\begin{threeparttable}
		
		\vspace{0.5em}
		{ \def\sym#1{\ifmmode^{#1}\else\(^{#1}\)\fi} 
			\centering 
						\begin{tabularx}{\textwidth}{X c*{4}{c}}
		  \hline\hline
			                             & \multicolumn{4}{c}{Dependent variable: $p_{t} - d_{t}$}\\
		  \cline{2-5}	
			%\multicolumn{1}{r}{Fraction of young investors:} 
			&\multicolumn{2}{c}{$Y_t$ age-based}&\multicolumn{2}{c}{$Y_t$ age/wealth-based}\\
			&\multicolumn{1}{c}{(1)}&\multicolumn{1}{c}{(2)}&\multicolumn{1}{c}{(3)}&\multicolumn{1}{c}{(4)}\\
				
				\hline    
				
				&  &  &  & \\
				$\delta_{1}$ &0.430	&1.134**	&0.681**	&-0.339	\\
				&(0.288)	&(0.188)	&(0.209)	&(0.268)	\\[2mm]
				$\delta_{2}$ &0.016	&-0.194		&-0.082		&-0.137	\\
				&(0.304)	&(0.322)	&(0.210)	&(0.422)	\\[2mm]
				$\delta_{3}$ &-0.460**		&-0.778**	&-0.629**	&0.428*	\\
				&(0.225)	&(0.282)	&(0.161)	&(0.220)	\\[2mm]
				$\beta_{1}$ &0.623**	&0.297**	&0.417**	&1.084**	\\
				&(0.259)	&(0.134)	&(0.177)	&(0.081)	\\[2mm]
				$\beta_{2}$&-0.099	&-0.132		&-0.168	&-0.223**	\\
				&(0.219)	&(0.140)	&(0.166)	&(0.100)	\\[2mm]
				$\beta_{3}$&0.298**	&0.452**	&0.580**	&-0.044	\\
				&(0.151)	&(0.106)	&(0.140)	&(0.064)	\\[2mm]
				$\mu(S_{1})$&4.925**		&5.780**	&5.734**	&5.430**	\\
				&(1.541)	&(2.915)	&(1.513)	&(1.683)	\\[2mm]
				$\mu(S_{2})$&19.590**		&14.100**	&20.310**	&19.670**	\\
				&(3.399)	&(2.959)	&(3.144)	&(3.102)	\\[2mm]
				$\sigma$&3.948	&3.259	&3.792	&3.737	\\
				&(0.392)	&(0.393)	&(0.397)	&(0.375)	\\[2mm]
				$Q_{11}$&0.953	&0.779		&0.956	&0.931	\\
				&(0.030)	&(0.152)	&(0.026)	&(0.035)	\\[2mm]
				$Q_{21}$&0.374	&0.115	&0.365		&0.473	\\
				&(0.214)	&(0.068)	&(0.210)	&(0.220)	\\
				\hline
				\textit{N} &51	&51	&51	&51 \\
				\hline\hline
				
			\end{tabularx}
		}
		\begin{small}
			\begin{tablenotes}
				\item Standard errors in parentheses. * significant at 10\%; ** significant at 5\%.
			\end{tablenotes}
		\end{small}		
	\end{threeparttable}
\end{table}

\end{appendices}
\end{document}